\documentclass[a4paper,11pt]{article}
\usepackage{libertine}
\usepackage[utf8]{inputenc}
\usepackage[T1]{fontenc}
\usepackage[noadjust,compress]{cite}
\usepackage{a4wide}
\usepackage{hyperref}
\hypersetup{colorlinks=true,citecolor=blue,linkcolor=blue,urlcolor=blue}
\usepackage{amsmath, mathtools}
\usepackage{amssymb}
\usepackage{amsthm}   
\usepackage{enumerate}
\usepackage{thm-restate}
\usepackage{caption}
\usepackage{subcaption}
\usepackage{xspace}
\usepackage{cleveref}
\usepackage{comment}
\usepackage{stmaryrd}
\renewcommand{\cref}{\Cref}

\usepackage{soul}

\usepackage{tikz}
\usetikzlibrary{fit,positioning, intersections, shapes.geometric,shapes.symbols,arrows,decorations.markings,decorations.pathreplacing,calc}

\usepackage{algorithm}
\usepackage[noend]{algpseudocode}
\makeatletter

\newcommand{\algorithmicbreak}{\textbf{break}}
\newcommand{\Break}{\State \algorithmicbreak}
\makeatother

\theoremstyle{plain}
\newtheorem{thm}{Theorem}[section]
\newtheorem{cor}[thm]{Corollary}
\newtheorem{lem}[thm]{Lemma}
\newtheorem{obs}[thm]{Observation}

\newtheorem{innercustomthm}{Theorem}
\newenvironment{customthm}[1]
  {\renewcommand\theinnercustomthm{#1}\innercustomthm}
  {\endinnercustomthm}

\crefname{thm}{Theorem}{Theorems}
\Crefname{thm}{Theorem}{Theorems}

\Crefmultiformat{lem}{Lemmas~#2#1#3}%
{ and~#2#1#3}{, #2#1#3}{, and~#2#1#3}

\theoremstyle{definition}
\newtheorem{defn}[thm]{Definition}
\newtheorem{prob}[thm]{Problem}

\newcommand{\bigO}{\mathcal{O}}

\newcommand{\calI}{\mathcal{I}}
\newcommand{\calJ}{\mathcal{J}}
\newcommand{\calN}{\mathcal{N}}
\newcommand{\calS}{\mathcal{S}}

\newcommand{\calL}{\mathcal{L}}
\newcommand{\calR}{\mathcal{R}}

\newcommand{\calF}{\mathcal{F}}

\newcommand{\calA}{\mathcal{A}}
\newcommand{\calB}{\mathcal{B}}
\newcommand{\calM}{\mathcal{M}}

\newcommand{\boldx}{\mathbf{x}}

\newcommand{\boldd}{\mathbf{d}}
\newcommand{\boldz}{\mathbf{z}}
\newcommand{\boldp}{\mathbf{p}}

\newcommand{\rom}[1]{%
	\textup{\uppercase\expandafter{\romannumeral#1}}%
}

\newcommand{\NEQ}{\mathrm{NEQ}}

\newcommand{\PP}{\textsf{P}}
\newcommand{\NP}{\textsf{NP}}

\newcommand{\coloringVD}[1]{\ensuremath{#1\textsc{-ColoringVD}}\xspace}
\newcommand{\coloringED}[1]{\ensuremath{#1\textsc{-ColoringED}}\xspace}

\newcommand{\LHomVD}{\textsc{LHomVD}}
\newcommand{\LHomED}{\textsc{LHomED}}

\newcommand{\Hom}{\textsc{Hom}}
\newcommand{\LHom}{\textsc{LHom}}
\newcommand{\ED}{\mathrm{LHomED}}

\newcommand{\id}{\mathsf{id}}

\mathchardef\hyph="2D
\newcommand{\abs}[1]{|#1|}
\newcommand{\eps}{\varepsilon}
\renewcommand{\epsilon}{\varepsilon}
\renewcommand{\phi}{\varphi}
\renewcommand{\geq}{\geqslant}
\renewcommand{\leq}{\leqslant}
\renewcommand{\ge}{\geqslant}
\renewcommand{\le}{\leqslant}
\renewcommand{\tilde}{\widetilde}
\newcommand{\from}{\colon}

\usepackage{todonotes}

\newenvironment{myitemize}
{ \begin{itemize}
		\setlength{\itemsep}{0pt}
		\setlength{\parskip}{0pt}
		\setlength{\parsep}{0pt}     }
	{ \end{itemize}                  }

\newenvironment{myenumerate}[1][1.]
{  \begin{enumerate}[#1]
		\setlength{\itemsep}{0pt}
		\setlength{\parskip}{0pt}
		\setlength{\parsep}{0pt}     }
{ \end{enumerate} 				 }

\graphicspath{{figures/}}

\newcommand{\executeiffilenewer}[3]{%
\ifnum\pdfstrcmp{\pdffilemoddate{#1}}%
{\pdffilemoddate{#2}}>0%
{\immediate\write18{#3}}\fi%
} 

\usepackage{bm}
\newcommand{\del}{\bm{\times}}
\newcommand{\precarc}{\prec_{\text{arc}}}
\newcommand{\aux}{\mathrm{Aux}}

\newcommand{\nh}{\Gamma}
\newcommand{\edcount}{\mathsf{cost}_\text{ed}}
\newcommand{\vdcount}{\mathsf{cost}_\text{vd}}

\DeclareMathOperator*{\trans}{transition}

\newcommand{\param}{i^\bullet}

\newcommand{\core}[2]{$(#1,#2)$-hub}

\newcommand{\cpar}{p}

\newcommand{\coreword}{hub}

\usepackage[margin=1in]{geometry}

\begin{document}
\title{List homomorphisms by deleting edges and vertices:\\ tight complexity bounds for bounded-treewidth graphs}

\author{
Bar\i\c{s} Can Esmer\thanks{CISPA Helmholtz Center for Information Security. The first author is also affiliated with the Saarbr\"ucken Graduate School of Computer Science, Saarland Informatics Campus, Germany. Research of the third author is supported by the European Research Council (ERC) consolidator grant No.~725978 SYSTEMATICGRAPH.} \and 
Jacob Focke$^\ast$ 
\and 
D\'{a}niel Marx$^\ast$ \and
Paweł Rzążewski\thanks{Warsaw University of Technology, Faculty of Mathematics and Information Science and University of Warsaw, Institute of Informatics, \texttt{pawel.rzazewski@pw.edu.pl}. Supported by the Polish National Science Centre grant no. 2018/31/D/ST6/00062.}
}

\date{}

\maketitle

\begin{abstract}
The goal of this paper is to investigate a family of optimization problems arising from list homomorphisms, and to understand what the best possible algorithms are if we restrict the problem to bounded-treewidth graphs. Given graphs $G$, $H$, and lists $L(v)\subseteq V(H)$ for every $v\in V(G)$, a {\em list homomorphism} from $(G,L)$ to $H$ is a function $f:V(G)\to V(H)$ that preserves the edges (i.e., $uv\in E(G)$ implies $f(u)f(v)\in E(H)$) and respects the lists (i.e., $f(v)\in L(v)$). The graph $H$ may have loops. For a fixed $H$, the input of the optimization problem $\textsc{LHomVD}(H)$ is a graph $G$ with lists $L(v)$, and the task is to find a set $X$ of vertices having minimum size such that $(G-X,L)$ has a list homomorphism to $H$. We define analogously the edge-deletion variant $\textsc{LHomED}(H)$, where we have to delete as few edges as possible from $G$ to obtain a graph that has a list homomorphism. This expressive family of problems includes members that are essentially equivalent to fundamental problems such as \textsc{Vertex Cover}, \textsc{Max Cut}, \textsc{Odd Cycle Transversal}, and \textsc{Edge/Vertex Multiway Cut}.

  For both variants, we first characterize those graphs $H$ that make the problem polynomial-time solvable and show that the problem is \textsf{NP}-hard for every other fixed $H$. Second, as our main result, we determine for every graph $H$ for which the problem is \textsf{NP}-hard, the smallest possible constant $c_H$ such that the problem can be solved in time $c^t_H\cdot n^{\mathcal{O}(1)}$ if a tree decomposition of $G$ having width $t$ is given in the input. Let $i(H)$ be the maximum size of a set of vertices in $H$ that have pairwise incomparable neighborhoods. For the vertex-deletion variant $\textsc{LHomVD}(H)$, we show that the smallest possible constant is $i(H)+1$ for every $H$:
\begin{itemize}
\item Given a tree decomposition of width $t$ of $G$, $\textsc{LHomVD}(H)$ can be solved in time $(i(H)+1)^t\cdot n^{\mathcal{O}(1)}$.
\item For any $\epsilon>0$ and $H$, an $(i(H)+1-\epsilon)^t\cdot n^{\mathcal{O}(1)}$ algorithm would violate the Strong Exponential-Time Hypothesis (SETH).
\end{itemize}

The situation is more complex for the edge-deletion version. For every $H$, one can solve $\textsc{LHomED}(H)$ in time $i(H)^t\cdot n^{\mathcal{O}(1)}$ if a tree decomposition of width $t$ is given. However, the existence of a specific type of decomposition of $H$ shows that there are graphs $H$ where $\textsc{LHomED}(H)$ can be solved significantly more efficiently and the best possible constant can be arbitrarily smaller than $i(H)$. Nevertheless, we determine this best possible constant and (assuming the SETH) prove tight bounds for every fixed $H$.

\end{abstract}

\newpage
\pagestyle{empty}
\tableofcontents
\newpage
\setcounter{page}{1}
\pagestyle{plain}

\section{Introduction}
Typical \NP-hard graph problems are known to be solvable in polynomial time when the input graph is restricted to be of bounded treewidth. In many cases, the problem is actually fixed-parameter tractable (FPT) parameterized by treewidth: given a tree decomposition of width $t$, the problem can be solved in time $f(t)\cdot n^{\bigO(1)}$ for some function $f$ \cite{DBLP:books/sp/CyganFKLMPPS15,DBLP:journals/tcs/DiazST02,DBLP:conf/dimacs/DiazST01}. While early work focused on just establishing this form of running time, more recently there is increased interest in obtaining algorithms where the function $f$ is growing as slowly as possible. New techniques such as representative sets, cut-and-count, subset convolution, and generalized convolution were developed to optimize the function $f(t)$.

On the complexity side, a line of work started by Lokshtanov, Marx, and Saurabh \cite{DBLP:journals/talg/LokshtanovMS18} provides tight lower bounds for many problems where $c^t\cdot n^{\bigO(1)}$-time algorithms were known. These type of complexity results typically show the optimality  of the base $c$ of the exponent in the best known $c^t\cdot n^{\bigO(1)}$-time algorithm, by proving that the existence of a $(c-\epsilon)^t\cdot n^{\bigO(1)}$-algorithm for any $\epsilon>0$ would violate the Strong Exponential-Time Hypothesis (SETH)
\cite{DBLP:journals/siamcomp/OkrasaR21,DBLP:conf/esa/OkrasaPR20,DBLP:conf/stacs/EgriMR18,DBLP:conf/soda/CurticapeanLN18,
DBLP:conf/iwpec/BorradaileL16,DBLP:journals/dam/KatsikarelisLP19, DBLP:conf/icalp/MarxSS21, DBLP:conf/soda/CurticapeanM16,DBLP:conf/soda/FockeMR22,DBLP:conf/iwpec/MarxSS22,DBLP:conf/soda/FockeMINSSW23}.
The goal of this paper is to unify some of these lower bounds under the umbrella of \emph{list homomorphism with deletion} problems, obtaining tight lower bounds for an expressive family of optimization problems that include members that are essentially equivalent to fundamental problems such a \textsc{Vertex Cover}, \textsc{Max Cut}, \textsc{Odd Cycle Transversal}, and \textsc{Edge/Vertex Multiway Cut}.

\paragraph{Graph homomorphisms.} Given graphs $G$ and $H$, a {\em homorphism from $G$ to $H$} is a (not necessarily injective) mapping $f: V(G)\to V(H)$ that preserves the edges of $G$, that is, if $uv\in E(G)$, then $f(u)f(v)\in E(H)$. For example, if $H$ is the complete graph $K_c$ on $c$ vertices, then the homomorphisms from $G$ to $H$ correspond to the proper vertex $c$-colorings of $G$: adjacent vertices have to be mapped to distinct vertices of $H$. For a fixed graph $H$, the problem $\textsc{Hom}(H)$ asks if the given graph $G$ has a homomorphism to $H$. Motivated by the connection to $c$-coloring when $H=K_c$, the problem is also called \textsc{$H$-coloring} \cite{DBLP:journals/jct/HellN90,DBLP:journals/csr/HellN21,DBLP:journals/csr/HellN08,DBLP:books/daglib/0013017,DBLP:journals/dam/MaurerSW81,DBLP:journals/iandc/MaurerSW81b,Hahn1997}.

The list version of $\textsc{Hom}(H)$ is the generalization of the problem where the possible image of each $v\in V(G)$ is restricted \cite{FEDER1998236,DBLP:conf/dimacs/HellN01,DBLP:journals/jgt/FederHH03,DBLP:journals/combinatorica/FederHH99,DBLP:conf/soda/HellR11,DBLP:conf/soda/EgriHLR14,DBLP:conf/lics/DalmauEHLR15,DBLP:conf/stacs/EgriMR18,DBLP:conf/mfcs/BokBFHJ20,DBLP:conf/esa/OkrasaPR20,DBLP:conf/soda/FockeMR22}. This generalization allows us to express a wider range of problems and it makes complexity results more robust. Formally, for a fixed undirected graph $H$ (possibly with selfloops), the input of the $\textsc{LHom}(H)$ problem consists of a graph $G$ and a {\em list assignment} $L:V(G)\to 2^{V(H)}$, the task is to decide if there is a \emph{list homomorphism} $f$ from $(G,L)$ to $H$, that is, a homomorphism $f$ from $G$ to $H$ that satisfies $f(v)\in L(v)$ for every $v\in V(G)$. Note that $\textsc{Hom}(H)$ is trivial if $H$ has a vertex with a loop, but loops may have a non-trivial role in the $\textsc{LHom}(H)$ problem as not every list may contain the same looped vertex. In fact, it is already non-trivial to consider the special case where $H$ is \emph{reflexive} \cite{FEDER1998236}, that is, every vertex of $H$ has a loop.

The main topic of the current paper is a further generalization of $\textsc{LHom}(H)$ to an optimization problem where we are allowed to delete some edges/vertices of $G$. The edge-deletion variant $\LHomED(H)$ is defined the following way: given a graph $G$ and a list assignment $L:V(G)\to 2^{V(H)}$, the task is to find a minimum set $X\subseteq E(G)$ of edges such that $(G\setminus X,L)$ has a list homomorphism to $H$. In other words, we want to find a mapping $f:V(G)\to V(H)$ that satisfies $f(v)\in L(v)$ for every $v\in V(G)$ and satisfies $f(u)f(v)\in E(H)$ for the maximum number of edges $uv$ of $G$. The vertex-deletion variant $\LHomVD(H)$ is defined analogously: here the task is to find a minimum size set $X$ of vertices such that $(G-X,L)$ has a list homomorphism to $H$. The $\LHomVD(H)$ problem was considered from the viewpoint of FPT algorithms parameterized by the number of removed vertices \cite{DBLP:journals/algorithmica/ChitnisEM17,DBLP:conf/stoc/0002KPW22}.

While the $\textsc{Hom}(H)$ and $\textsc{LHom}(H)$ problems can be seen as generalizations of vertex coloring, the framework of deletion problems we consider here can express a wide range of fundamental optimization problems. We show below how certain problems can be reduced to $\LHomED(H)$ or $\LHomVD(H)$ for some fixed $H$. The reductions mostly work in the other direction as well (we elaborate on that in Section~\ref{sec:overview}), showing that this framework contains problems that are essentially equivalent to well-studied basic problems.
\begin{itemize}
\item \textsc{Vertex Cover}: Let $H=K_1$ be a single vertex $x$ without a loop. Then \textsc{Vertex Cover} can be expressed by $\LHomVD(K_1)$ with single-element lists: as vertex $x$ is not adjacent to itself, it follows that for every edge $uv$ of $G$, at least one of $u$ and $v$ has to be deleted.
\item \textsc{Independent Set}: As $G$ has an independent set of size $k$ if and only if it has a vertex cover of size $|V(G)|-k$, the same reduction can be used.
\item \textsc{Max Cut}: Let $H=K_2$ be two adjacent vertices without loops. Then \textsc{Max Cut} can be expressed by $\LHomED(H)$ with the list $V(H)$ at every vertex: the task is to delete the minimum number of edges to obtain a bipartition $(X,Y)$, that is, to maximize the number of edges between $X$ and $Y$.
\item \textsc{Odd Cycle Transversal}: Let $H=K_2$ be two adjacent vertices without loops. Then \textsc{Odd Cycle Transversal} can be expressed by $\LHomVD(H)$ with the list $V(H)$ at every vertex: the task is to delete the minimum number of vertices to obtain a bipartite graph.
\item \textsc{$s$-$t$ Min Cut}: Let $H$ contain two independent vertices $v_s$ and $v_t$ with selfloops. Then \textsc{$s$-$t$ Min Cut} can be expressed as $\LHomED(H)$ where $L(s)=\{v_s\}$, $L(t)=\{v_t\}$, and the list is $\{v_s,v_t\}$ for all remaining vertices. It is clear that $s$ and $t$ cannot be in the same component after removing the solution $X$ from $G$.

\item \textsc{Edge Multiway Cut} with $c$ terminals $t_1$, $\dots$, $t_c$: Let $H$ be $c$ independent vertices $v_1$, $\dots$, $v_c$ with selfloops. Then the problem can be expressed as $\LHomED(H)$ where $L(t_i)=\{v_i\}$ and non-terminals have list $\{v_1,\dots, v_c\}$. It is clear that if $X$ is a solution of $\LHomED(H)$, then each component of $G\setminus X$ can contain at most one terminal.
  
\item \textsc{Vertex Multiway Cut} with $c$ (undeletable) terminals $t_1$, $\dots$, $t_c$: Let $H$ be $c$ independent vertices $v_1$, $\dots$, $v_c$ selfloops. First we modify the graph: if terminal $t_i$ is adjacent to a vertex $w$, then we replace $t_i$ by $n=|V(G)|$ degree-1 copies adjacent to $w$. Then the problem can be expressed as $\LHomVD(H)$ where the list is $\{v_i\}$ for each copy of $t_i$ and $\{v_1,\dots, v_c\}$ for each non-terminal vertex. Observe that it does not make sense to delete a copy of any terminal. Therefore, the optimal solution to $\LHomVD(H)$ is a set of vertices, disjoint from the terminals, that separates the original terminals.
\end{itemize}

For every fixed $H$, our results give tight lower bounds for $\LHomVD(H)$ and $\LHomED(H)$, parameterized by the width of the given tree decomposition  problems. This comprehensive set of results reprove earlier lower bounds on basic problems in a uniform way, extend them to new problems that have not been considered before (e.g., \textsc{Multiway Cut}), and in fact fully investigate a large, well-defined family of problems. Earlier results in this area typically focused on specific problems or relatively minor variants of a specific problem. Compared to that, our results focus on a family of problems that include a diverse set of optimization problems interesting on their own right. The tight characterization also includes an algorithmic surprise: for some $\LHomED(H)$ problems, the obvious brute force algorithm is not optimal on its own, one needs to consider a specific form of decomposition into subproblems to achieve the best possible algorithm.

  \paragraph{Polynomial-time cases.} The seminal work of Ne\v{s}et\v{r}il and Hell \cite{DBLP:journals/jct/HellN90} characterized the polynomial-time solvable cases of $\Hom(H)$: it can be solved in polynomial time if $H$ is bipartite or has a loop, and it is \NP-hard for every other fixed $H$. For the more general list version, we need more restrictions: Feder, Hell, and Huang \cite{DBLP:journals/jgt/FederHH03} showed that the problem is polynomial-time solvable if $H$ is a bi-arc graph. Given the amount of attention this type of problems received in the literature \cite{DBLP:journals/jct/HellN90,FEDER1998236,DBLP:journals/combinatorica/FederHH99,DBLP:conf/dimacs/HellN01,DBLP:journals/jgt/FederHH03,FederHH07,DBLP:conf/soda/HellR11,DBLP:conf/soda/EgriHLR14,DBLP:conf/lics/DalmauEHLR15,DBLP:conf/mfcs/BokBFHJ20,DBLP:journals/dam/FederH20,DBLP:journals/siamdm/HellR12a,DBLP:conf/esa/HellMNR12,DBLP:journals/dam/FederHSS11}, it is somewhat surprising that the polynomial-time solvability of the deletion versions have not been systematically studied. Therefore, our first contribution is a polynomial-time versus \NP-hard dichotomy for $\LHomED(H)$ and  $\LHomVD(H)$. As expected, these more general problems remain polynomial-time solvable only for an even more restricted class of graphs. In particular, the reduction above from \textsc{Vertex Cover} shows that $\LHomVD(H)$ becomes \NP-hard already when there is a single loopless vertex in $H$ and hence we can expect polynomial-time algorithms only for reflexive graphs $H$.

\begin{restatable}{thm}{thmvddicho}
\label{thm:vd-dicho-intro}
The $\LHomVD(H)$ problem is polynomial-time solvable if $H$ is reflexive and does not contain any of the following:
\begin{myenumerate}
\item three pairwise non-adjacent vertices,
\item an induced four-cycle, or
\item an induced five-cycle.
\end{myenumerate}
Otherwise, $\LHomVD(H)$ is \NP-hard.
\end{restatable}

Edge-deletion problems are typically easier than their vertex-deletion counterparts, but the boundary line between the easier and hard cases is more difficult to characterize. This is also true in our case: for $\LHomED(H)$, the graph $H$ does not have to be reflexive to make the problem polynomial-time solvable, hence the proof of the classification result becomes significantly more complicated as graphs with both reflexive (i.e., looped) and irreflexive (i.e., non-looped) vertices must be handled as well. We need the following definition to state the dichotomy result. We say that the three vertices $v_1$, $v_2$, $v_3$ have \emph{private neighbors} if there are vertices $v'_1$, $v'_2$, $v'_3$ (not necessarily disjoint from $\{v_1,v_2,v_3\}$) such that $v_i$ and $v'_j$ are adjacent if and only if $i=j$. In particular, if $\{v_1,v_2,v_3\}$ are independent reflexive vertices then they have private neighbors. \emph{Co-private neighbors} are defined similarly, but $i=j$ is replaced by $i\neq j$. In particular, if $\{v_1,v_2,v_3\}$ are pairwise adjacent irreflexive vertices, then they have co-private neighbors. Finally, we say an edge is \emph{irreflexive} if both of its endpoints are irreflexive vertices. 

\begin{restatable}{thm}{thmeddicho}
\label{thm:ed-dicho-intro}
The $\LHomED(H)$ problem is polynomial time solvable if $H$ does not contain any of the following:
	\begin{myenumerate}
		\item an irreflexive edge,
		\item a $3$-vertex set $S$ with private neighbors, or
		\item a $3$-vertex set $S$ with co-private neighbors.
	\end{myenumerate}
Otherwise, $\LHomED(H)$ is \NP-hard.
\end{restatable}
The proof of Theorem~\ref{thm:ed-dicho-intro} exploits a delicate interplay between the geometric bi-arc representation (in the algorithm) and the characterization by forbidden subgraphs (for hardness). While the proofs of these dichotomy results are non-trivial, we do not consider them to be the main results of the paper.
Clearly, understanding the easy and hard cases of the problem is a necessary prerequisite for the lower bounds we are aiming at, hence we needed to prove these dichotomy results as they were not present in the literature in this form. We remark that $\LHomED(H)$ can be formulated as a Valued Constraint Satisfaction Problem (VCSP) with a single binary relation, hence the existence of a polynomial-time versus \NP-hard dichotomy should follow from known results on the complexity of VCSP \cite{DBLP:journals/siamcomp/KolmogorovKR17,DBLP:journals/jacm/ThapperZ16,DBLP:journals/jacm/KolmogorovZ13}. 
However, we obtain in a self-contained way a compact statement of an easily checkable classification property with purely graph-theoretic proofs and algorithms.

\paragraph{Bounded-treewidth graphs, vertex  deletion.}
Let us first consider the vertex-deletion version $\LHomVD(H)$ and determine how exactly the complexity of the problem depends on treewidth. We assume that the input contains a tree decomposition of $G$ having width $t$, we try to determine the smallest $c$ such that the problem can be solved in time $c^t\cdot n^{O(1)}$. This question has been investigated for $\Hom(H)$ \cite{DBLP:journals/siamcomp/OkrasaR21}, $\LHom(H)$ \cite{DBLP:conf/stacs/EgriMR18,DBLP:conf/esa/OkrasaPR20}, and the counting version of $\LHom(H)$ \cite{DBLP:conf/soda/FockeMR22}.

Standard dynamic programming techniques show that $\LHomVD(H)$ can be solved in time $(|V(H)|+1)^t\cdot n^{\bigO(1)}$ if a tree decomposition of width $t$ is given: each vertex has $|V(H)|$ possible ``states'' corresponding to where it is mapped to, plus one more state corresponding to deleting the vertex.
For some $H$, this naive algorithm can be improved the following way. First, if every vertex in $G$ has a list of size at most $\ell$, then $(|V(H)|+1)$ can be improved to $\ell+1$: each vertex has only $\ell$ states corresponding to the possible images, plus the state representing deletion. Second, we say that a set $S\subseteq V(H)$ is \emph{incomparable} if the neighborhoods of any two vertices in $S$ are incomparable, that is, for any $u,v\in S$, there is $u'\in \nh(u)\setminus \nh(v)$ and $v'\in \nh(v)\setminus \nh(u)$ (we denote by $\nh(v)$ the neighborhood of a vertex $v$, which includes $v$ itself if it has a loop). Let $i(H)$ be the size the of the largest incomparable set in $H$. The main observation (already made in \cite{DBLP:conf/stacs/EgriMR18,DBLP:conf/esa/OkrasaPR20}) is that we can assume that every list $L(v)$ is an incomparable set: if $\nh(v)\subseteq \nh(v')$ for $v,v'\in L(v)$, then we can always use $v'$ in place of $v$ in a solution. Therefore, we can assume that every list has size at most $i(H)$, resulting in running time $(i(H)+1)^t\cdot n^{\bigO(1)}$. Our main result for the vertex-deletion version shows the optimality of this running time.

  \begin{thm}\textbf{\rm \bf (Main result for treewidth, vertex deletion)}\label{thm:vd-main-intro-tw}
  Let $H$ be a fixed graph which contains either an irreflexive vertex or three pairwise non-adjacent reflexive vertices or an induced reflexive cycle on four or five vertices. Then $\LHomVD(H)$ on $n$-vertex instances given with a tree decomposition of width~$t$
  \begin{myenumerate}[(a)]
  \item can be solved in time $(i(H)+1)^t \cdot n^{\bigO(1)}$ and
    \item cannot be solved in time $(i(H)+1-\epsilon)^t \cdot n^{\bigO(1)}$ for any $\epsilon >0$, unless the SETH fails.
    \end{myenumerate}
  \end{thm}

  Theorem~\ref{thm:vd-main-intro} refines the \NP-hardness of Theorem~\ref{thm:vd-dicho-intro} by obtaining a lower bound that precisely matches the algorithm described above. This shows that for $\LHomVD(H)$, restricting the lists to incomparable sets is the \emph{only} algorithmic idea that can improve the running time of the naive algorithm. In particular, we cannot consider the connected components of $H$ separately (as was possible in the earlier results \cite{DBLP:journals/siamcomp/OkrasaR21,DBLP:conf/esa/OkrasaPR20,DBLP:conf/stacs/EgriMR18,DBLP:conf/soda/FockeMR22}). It is an essential feature of the deletion problem that hardness can stem from disconnected structures (see, for example, the reduction above from \textsc{Multiway Cut}).

\paragraph{Bounded-treewidth graphs, edge  deletion.}
For the edge-deletion version $\LHomED(H)$, the natural expectation is that $i(H)^t\cdot n^{\bigO(1)}$ is the best possible running time: as vertices cannot be deleted, each vertex $v$ has only $|L(v)|\le i(H)$ states in the dynamic programming. While this running time can be achieved using the idea of incomparable sets, it turns out that, somewhat surprisingly, this is \emph{not} the optimal running time for every $H$. There are graphs $H$ for which $\LHomED(H)$ can be solved significantly faster, thanks to a new algorithmic idea, the use of a specific form of decomposition. We need the following definition. 

\begin{restatable}[Decomposition]{defn}{defdecomp}
\label{def:decomp}
Given a graph $H$ with vertex set $V$ and a partition of $V$ into three possibly empty sets $A$, $B$, and $C$, we say that $(A,B,C)$ is a \emph{decomposition} of $H$ if the following hold:
	\begin{myitemize}
		\item $B$ is a reflexive clique with a full set of edges between $A$ and $B$,
		\item $C$ is an (irreflexive) independent set with no edge between $A$ and $C$,
		\item $A\neq \emptyset$ and $B\cup C\neq \emptyset$.
	\end{myitemize}
\end{restatable}
The crucial property of this definition is that if $S$ is an incomparable set, then it is fully contained in one of $A$, $B$, or $C$. Indeed, for any $a\in A$, $b\in B$, $c\in C$, we have $\nh(c)\subseteq \nh(a)\subseteq \nh(b)$. Therefore, if we assume that each list $L(v)$ is an incomparable set, then each $L(v)$ is a subset of one of these three sets. Let $V_A$, $V_B$, $V_C$  be the sets of vertices of $G$ whose lists are a subset of $A$, $B$, and $C$, respectively. Observe that if $u\in V_A$ and $v\in V_B$ are adjacent in $G$, then whenever assignment $f:V(G)\to V(H)$ respects the lists of $u$ and $v$, then $f(u)f(v)$ is \emph{always} an edge of $H$ (as $A$ and $B$ are fully connected). Therefore, the edge $uv$ of $G$ does not play any role in the problem. Similarly, if $u\in V_A$ and $v\in V_C$, then $f(u)f(v)$ is \emph{never} an edge of $H$ (as $A$ and $C$ are independent), hence $uv$ always has to be deleted in the solution. This means that the edges between $V_A$ and $V_B\cup V_C$ can be ignored and the problem falls apart into two independent instances $G[V_A]$ of $\LHomED(H[A])$ and $G[V_B\cup V_C]$ of $\LHomED(H[V_B\cup V_C])$.\label{page:ed-alg}

\begin{figure}
	\centering
	\includegraphics[scale=0.83,page=1]{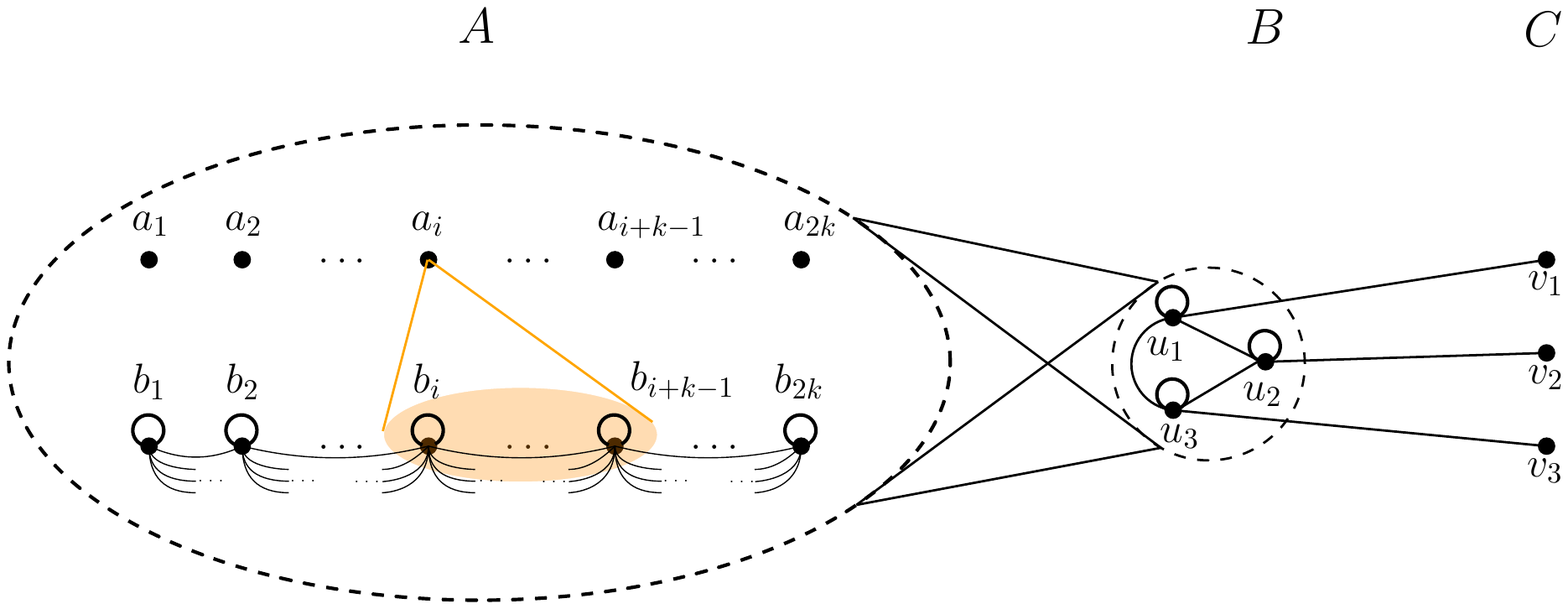}
	\caption{An example of a graph $H$ that has a decomposition $(A,B,C)$, with $i(H)=k$ and $\param(H)=3$. The $a_i$'s form an irreflexive independent set and the $b_i$'s form a reflexive clique. Every vertex $a_i$ is adjacent to $\{b_{i},\dots, b_{i+k-1}\}$, and $\{u_1,u_2,u_3\}$ is fully adjacent to every $a_i$ and $b_i$. Observe that $i(H)\ge i(H[A])\ge k$, as vertices $a_1$, $\dots$, $a_k$ have incomparable neighborhoods. There is no irreflexive edge in $H[A]$, and it can be checked that there is no 3-element set with private or co-private neighbors, implying that $\LHomED(H[A])$ is polynomial-time solvable. But $\{u_1,u_2,u_3\}$ has private neighbors, making $\LHomED(H)$ \NP-hard.} \label{fig:dec-intro}
\end{figure}

How does this observation help solving the problem more efficiently? As every incomparable set is a subset of one of the three sets, we have $i(H)=\max\{i(H[A]),i(H[B\cup C])\}$. Thus it seems that one of the two instances will be at least as hard as the original instance. The catch is that it could happen that one of the two instances is polynomial-time solvable and contains a large incomparable set, while the other is \NP-hard but contains only small incomparable sets. For example, it is possible that $i(H)=i(H[A])=k$, $i[H[B\cup C]]=3$,
but $\LHomED(H[A])$ is polynomial-time solvable. Then we can decompose the problem into an instance of $\LHomED(H[A])$ and an instance of $\LHomED(H[B\cup C])$, solve the former in polynomial time, and the latter in time $i(H[B\cup C])^t\cdot n^{\bigO(1)}=3^t\cdot n^{\bigO(1)}$. Figure~\ref{fig:dec-intro} shows an example where this situation occurs.

Our main result for the edge-deletion version is showing that there are \emph{precisely two} algorithmic ideas that can improve the running time for $\LHomED(H)$: restricting the lists to incomparable sets and exploiting decompositions. Formally, let $\param(H)$ be the maximum of $i(H')$ taken over all induced undecomposable subgraphs $H'$ of $H$ that is \emph{not} classified as polynomial-time solvable by Theorem~\ref{thm:ed-dicho-intro}, that is, do not contain an irreflexive edge and do not contain 3 vertices with private or co-private~neighbors.

\begin{thm}\textbf{\rm \bf (Main result for treewidth, edge deletion)}\label{thm:ed-main-intro-tw}
  Let $H$ be a fixed graph that contains either an irreflexive edge, three vertices with private neighbors, or three vertices with co-private neighbors. Then $\LHomED(H)$ on $n$-vertex instances given with a tree decomposition of width $t$
  \begin{myenumerate}[(a)]
  \item can be solved in time $\param(H)^t \cdot n^{\bigO(1)}$ and
    \item cannot be solved in time $(\param(H)-\epsilon)^t \cdot n^{\bigO(1)}$ for any $\epsilon >0$, unless the SETH fails.
    \end{myenumerate}
  \end{thm}

  For the lower bound of Theorem~\ref{thm:ed-main-intro}, it is sufficient to prove that $\param(H)$ is the correct base of the exponent if $H$ is undecomposable. The flavor of the proof is similar to the proof of Theorem~\ref{thm:vd-main-intro}, but more involved. One reason for the extra complication is that vertex-deletion problems typically give us more power when designing gadgets in a reduction than edge-deletion problems. But beyond that, an inherent difficulty in the proof of Theorem~\ref{thm:ed-main-intro} is that the proof needs to exploit somehow the fact that $H$ is undecomposable. Therefore, we need to find an appropriate certificate that the graph is undecomposable and use this certificate in the gadget construction throughout the proof.

  \paragraph{Parameterization by \coreword\ size.} 
  Esmer et al.~\cite{EsmerFMR24hub} presented a new perspective on lower bounds parameterized by the width of the tree decomposition given in the input. It was shown that many of these lower bonds hold even if we consider a larger parameter. These results showed that for many problems hard instances do not have to use the full power of tree decompositions (not even of path decompositions), the real source of hardness is instances consisting of a large central ``\coreword'' connected to an unbounded number of constant-sized components.

  Formally, we say that a set $Q$ of vertices is a \emph{\core
  {\sigma}{\delta}} of $G$ if every component of $G-Q$ has at most $\sigma$ vertices and each such component is adjacent to at most $\delta$ vertices of $Q$ in $G$. 
Observe that if a graph has a \core{\sigma}{\delta} core of size $p$, then this can be turned into a tree decomposition of width less than $p+\sigma$. This shows that if a problem can be solved in time $c^t\cdot n^{\bigO(1)}$ given a tree decomposition of width $t$ is given in the input, then for every fixed $\sigma$ and $\delta$, this problem can be solved in time $c^p\cdot n^{\bigO(1)}$ given a \core{\sigma}{\delta} of size $p$ is given in the input. Thus any lower bound ruling out the possibility of the latter type of algorithm for a given $c$ also rules out the possibility of the former type of algorithm. 
Esmer et al.~\cite{EsmerFMR24hub} showed that for many fundamental problems, the previously known lower bounds parameterized by the width of the tree decomposition can be strenghtened to parameterization by \coreword\ size. Following their work, we also present our lower bound results in such a stronger form\footnote{An astute reader might wonder if the statements below cannot be strengthened by making $\sigma$ and $\delta$ universal constants. These issues are discussed by Esmer et al.~\cite{EsmerFMR24hub}; we refer to their work for more details.}.

\begin{restatable}{thm}{thmvdcoloringmain}
\textbf{\rm \bf (Main result for \coreword\ size, vertex deletion)}\label{thm:vd-main-intro}
  Let $H$ be a fixed graph which contains either an irreflexive vertex or three pairwise non-adjacent reflexive vertices or an induced reflexive cycle on four or five vertices.
Then for every $\epsilon>0$, there are $\sigma,\delta>0$ such that $\LHomVD(H)$ with a \core{\sigma}{\delta} of size $p$ given in the input cannot be solved in time $(i(H)+1-\epsilon)^p \cdot n^{\bigO(1)}$, unless the SETH fails.
\end{restatable}

\begin{restatable}{thm}{thmedmainintro}\label{thm:ed-main-intro}
\textbf{\rm \bf (Main result for \coreword\ size, edge deletion)}
 Let $H$ be a fixed graph that contains either an irreflexive edge, three vertices with private neighbors, or three vertices with co-private neighbors.
Then for every $\epsilon>0$, there are $\sigma,\delta>0$ such that $\LHomED(H)$ on $n$-vertex instances with a \core{\sigma}{\delta} of size $p$ given in the input cannot be solved in time $(\param(H)-\epsilon)^p \cdot n^{\bigO(1)}$, unless the SETH fails.
\end{restatable}

  Let us observe that the lower bounds in Theorems~\ref{thm:vd-main-intro} and \ref{thm:ed-main-intro} immediately imply the lower bounds in Theorems~\ref{thm:vd-main-intro-tw} and \ref{thm:ed-main-intro-tw}, respectively. We present these strenghened results in this paper because obtaining them did not require any extra effort: as we shall see, we simply need to use a stronger known lower bound as a starting point.

  We prove all our lower bounds by reduction from two problems. In the \coloringVD{q} problem, given a graph $G$, the task is to remove the minimum number of vertices such that the resulting graph is $q$-colorable. The \coloringED{q} problem is similar, but here we need to remove the minimum number of edges instead. Tight lower bounds for these problems parameterized by the width of the tree decomposition are known \cite{DBLP:journals/talg/LokshtanovMS18,DBLP:conf/iwpec/HegerfeldK22}. Recently, Esmer et al.~\cite{EsmerFMR24hub} strenghtened these results to parameterization by \coreword\ size.
  
\begin{restatable}[\cite{EsmerFMR24hub}]{thm}{vdcoloring}
\label{thm:vd-coloring-intro}
For every $q\ge 1$ and $\epsilon > 0$, there exist integers $\sigma,\delta\ge 1$ such that if there is an algorithm solving in time $(q + 1 - \epsilon)^{\cpar} \cdot n^{\bigO(1)}$ every $n$-vertex instance of \coloringVD{q} given with a \core{\sigma}{\delta} of size at most $\cpar$, then SETH fails.
\end{restatable}
\begin{thm}[\cite{EsmerFMR24hub}]\label{thm:ed-coloring-intro}
	For every $q\ge 2$ and $\eps>0$, there are integers $\sigma$ and $\delta$ such that if an algorithm solves in time $(q-\eps)^p\cdot n^{\bigO(1)}$ every $n$-vertex instance of \coloringED{q} that is given with a \core{\sigma}{\delta} of size $p$, then the SETH fails.
      \end{thm}
Our reductions replace edges in a \coloringVD{q} or \coloringED{q} instance by constant-sized gadgets. One can observe that such a transformation has a small effect on treewidth and also on \coreword\ size (although might change $\sigma$ and $\delta$ slightly). Thus we can use Theorems~\ref{thm:vd-coloring-intro} and \ref{thm:ed-coloring-intro} in a tranparent way to obtain the lower bounds in Theorems~\ref{thm:vd-main-intro} and \ref{thm:ed-main-intro}.

\section{Technical Overview}\label{sec:overview}

In this section, we overview some of the most important technical ideas in our results. For clearity, we start with the discussion of the vertex-deletion version and then continue with the more complicated edge-deletion variant.

  \subsection{Vertex-deletion version}
  We start with the vertex-deletion version, where both the \PP~vs.~\NP-hard dichotomy and the complexity bounds for bounded-treewidth graphs are significantly easier to prove.

  \paragraph{Equivalence of $\LHomVD(H)$ with classic problems.}
  We have seen earlier how \textsc{Vertex Cover}, \textsc{Odd Cycle Transversal}, and \textsc{Vertex Multiway Cut} can be reduced to $\LHomVD(H)$ for various graphs $H$. Let us briefly discuss reductions in the reverse direction. It is clear that $\LHomVD(K_1)$ is actually equivalent to \textsc{Vertex Cover}: if we remove those vertices that have empty lists, then the problem is precisely finding a vertex cover of minimum size. However, $\LHomVD(K_2)$ seems to be more general than \textsc{Odd Cycle Transversal}: a list of size one can express that the vertex has to be on a certain side of the bipartition of $G-X$ (if the vertex is not removed). Therefore, $\LHomVD(K_2)$ is slightly more general than \textsc{Odd Cycle Transversal}, and equivalent to an annotated generalization, where given $G$ and two sets $L,R\subseteq V(G)$, the task is to find a set $X$ of vertices of minimum size such that $G-X$ has a bipartition with $R$ and $L$ on different sides.

  For \textsc{Vertex Multiway Cut} with undeletable terminals, we can reduce $\LHomVD(H)$ (where $H$ consists of $k$ independent reflexive vertices $w_1$, $\dots$, $w_k$) to a multiway cut instance $G'$ the following way. Given an instance $(G,L)$ of $\LHomVD(H)$, we obtain $G'$ by first extending it with $k$ terminals $t_1$, $\dots$, $t_k$. Then for every $v\in V(G)$, we introduce a clique of size $|L(v)|$ that is completely connected to $v$. We introduce a perfect matching between the vertices of this clique and the set of terminals that corresponds to the elements of $L(v)$. Therefore, in every solution of \textsc{Vertex Multiway Cut}, all but one vertex of each clique has to be deleted for sure. We can also assume that no more than $|L(v)|-1$ vertices of the clique are deleted: if every vertex of the clique were deleted, then we can modify the solution by removing $v$ instead. This means that if $v$ is not deleted, then it is in the component of a terminal from $L(v)$.  Therefore, it can be shown that there is a tight correspondence between the optimum cost of the $\LHomVD(H)$ instance and the optimum cost of the \textsc{Vertex Multiway Cut} instance. We can also note that this transformation increases treewidth at most by an additive constant and if the original graph has a \core{\sigma}{\delta} of size $p$, then the constructed graph has a \core{\sigma(k+1)}{\delta+k} of size $p+k$. Therefore, we can state the following lower bound:
\begin{thm}\label{thm:vmwaylb} For every $k\ge 3$ and $\epsilon>0$, there are $\sigma,\delta>0$ such that \textsc{Vertex Multiway Cut} with $k$ terminals with a \core{\sigma}{\delta} of size $p$ given in the input cannot be solved in time $(k+1-\epsilon)^p \cdot n^{\bigO(1)}$, unless the SETH fails.
\end{thm}

\paragraph{Dichotomy for vertex deletion.} We need to prove that $\LHomVD(H)$ is polynomial-time solvable if $H$ is reflexive and $i(H)\le 2$, and it is \NP-hard for every other $H$. If $H$ contains an irreflexive vertex, then we have seen that \textsc{Vertex Cover} can be reduced to $\LHomVD(H)$. For reflexive $H$, the \NP-hard cases of $\LHomVD(H)$ can be easily established using the following alternative characterizations of the tractability condition:
  
\begin{restatable}{lem}{vdalternative}  
\label{lem:vd-alternative-intro}
Let $H$ be a reflexive graph. The following conditions are equivalent.
\begin{myenumerate}
\item $i(H) \leq 2$, \label{it:i2}
\item $H$ does not contain three pairwise nonadjacent vertices, an induced four-cycle, nor an induced five-cycle, \label{it:obstructions}
 \\[-1.2em]
\item $H$ is an interval graph whose vertex set can be covered by two cliques. \label{it:covered}
\end{myenumerate}
\end{restatable}
If $H$ is reflexive and contains an induced four-cycle or an induced five-cycle, then already $\LHom(H)$ is \NP-hard \cite{FEDER1998236}. If $H$ contains three pairwise non-adjacent reflexive vertices, then we have seen that \textsc{Vertex Multiway Cut} with three (undeletable) terminals can be reduced to it.

For the polynomial cases, by Lemma~\ref{lem:vd-alternative-intro} we need to solve the problem only when $H$ is an interval graph that can be partitioned into two cliques $\mathcal{L}$ and $\mathcal{R}$. We can observe that in this case the neighborhoods of the vertices inside $\mathcal{L}$ and $\mathcal{R}$ form two chains. Thus if we assume that every list $L(v)$ is an incomparable set, then every list can contain at most two vertices: one from $\mathcal{L}$ and one from $\mathcal{R}$.

We reduce $\LHomVD(H)$ to a minimum $s\hyph t$ cut problem. Note that using some form of minimum cut techniques cannot be avoided, as $s\hyph t$ \textsc{Min Cut} can be reduced to the case when $H$ consists of two independent reflexive vertices. Let $V_L$ and $V_R$ be the sets of vertices $v$ where $L(v)\subseteq \mathcal{L}$ and $L(v)\subseteq \mathcal{R}$, respectively. If two vertices $u\not\in V_R$ and $v\not\in V_L$ are adjacent such that the vertex in $L(u)\cap\mathcal{L}$ is {\em not} adjacent to the vertex of $L(v)\cap \mathcal{R}$, then we add a directed edge from $u$ to $v$. 
After a solution to $\LHomVD(H)$ is deleted, the remaining vertices can be partitioned into a ``left'' and ``right'' part according to whether they were mapped to $\mathcal{L}$ or $\mathcal{R}$. The directed edge $\overrightarrow{uv}$
 represents the constraint that we cannot have $u$ on the left part and $v$ on the right part simultaneously. Then our problem is essentially reduced to deleting the minimum number of vertices such that there is no path from $V_L$ to $V_R$.

\paragraph{Reduction using gadgets.} To rule out algorithms with running time $(i(H)+1-\epsilon)^t\cdot n^{\bigO(1)}$, we reduce from \coloringVD{q} for $q=i(H)$ to $\LHomVD(H)$. For this purpose, we take an incomparable set $S$ of size $i(H)$ and construct gadgets that can express ``not equal on $S$.'' A gadget in this context means an instance of $\LHomVD(H)$ with a pair of distinguished vertices $(x,y)$. If neither of these vertices is removed, then they need to have different colors from $S$. Every solution has one of the $(|L(x)|+1)(|L(y)|+1)$ possible behaviors on $(x,y)$ (mapping to $V(H)$ or deleting the vertices). Each behavior on $(x,y)$ has some cost: the minimum number of vertex deletions we need to make inside the gadget to find a valid extension (note that this cost does not include the deletion of $x$ and/or $y$). Our goal is to construct a gadget where $L(x)=L(y)=S$ and every behavior on $(x,y)$ has the same cost $\alpha$, except that mapping $x$ and $y$ to the same vertex of $S$ extends only with cost strictly larger than $\alpha$. We call such gadgets \emph{$S$-prohibitors}. Then we can reduce \coloringVD{q} to $\LHomVD(H)$ by giving the list $S$ to every vertex of the original graph $G$, and by replacing each of the $m$ edges with a copy of the $S$-prohibitor gadget. Then it is easy to see that the original graph can be made $q$-colorable with $k$ deletions if and only if the constructed $\LHomVD(H)$ instance has a solution with $\alpha\cdot |E(G)| +k$ deletions.

\paragraph{Constructing the prohibitor gadgets.} A \emph{$(v,S)$-prohibitor} gadget has two portals $(x,y)$ with $L(x)=L(y)=S$, and every behavior has cost exactly $\alpha$, except that it has cost strictly more than $\alpha$ when both $x$ and $y$ are mapped to $v$. By joining together $(v,S)$-prohibitors for every $v\in S$, we obtain the $S$-prohibitor defined in the previous paragraph.

The construction of the $(v,S)$-prohibitors is the core technical part of the  proof of Theorem~\ref{thm:vd-main-intro}. The proof uses the fact that we are considering an \NP-hard case of $\LHomVD(H)$ and hence one of the obstructions listed in Theorem~\ref{thm:vd-dicho-intro} appears in the graph $H$ (irreflexive vertex, three non-adjacent vertices, induced four-cycle, induced five-cycle). Some case analysis is needed based on, e.g., the type of the obstruction that appears, but in all cases the construction is surprisingly compact. We need three additional types of gadgets, which are put together in the way shown in Figure~\ref{fig:intro-prohib}. We can interpret the two portals $x$ and $y$ of a gadget as input and output, respectively. Then setting a value on the input may ``force'' a single value on the output or ``allow'' some values on the output, meaning that these combinations on the input and the output can be extended with minimum cost. 
\begin{itemize}
\item \textit{splitter:} if the input is assigned $v$, then the output is forced to $v'$; if the input is from $S\setminus \{v\}$, then the output can be either $v'$ or $w'$.
\item \textit{translator:} if the input is assigned $v'$, then the output is forced to $a$; if the input is $w'$, then the output can be $b$.
\item \textit{matcher:} minimum cost can be achieved if one of the portals is assigned $a$ and the other is $b$, but cannot be achieved if both portals are assigned $a$.
\end{itemize}
Suppose that vertices $t_1$ and $t_6$ are connected with these gadgets as in Figure~\ref{fig:intro-prohib}. If both $t_1$ and $t_6$ are mapped to $v$, then the splitters force $t_2$ and $t_5$ to $v'$, the translators force $t_3$ and $t_4$ to $a$, which is incompatible with minimum cost of the matcher. On the other hand, if at least one of $t_1$ and $t_6$ is mapped to a vertex from $S\setminus \{v\}$, then the splitters allow us to map one of $t_2$ and $t_5$ to $w'$ and the other to $v'$. In this case, the translators allow us to map one of $t_3$ to $a$ and the other to $b$, which is now compatible with the minimum cost of the matcher.
\begin{figure}
\centering
\includegraphics[scale=1.3,page=1]{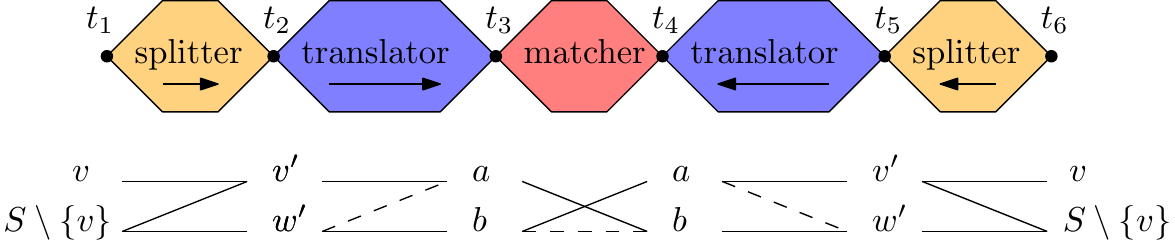}
  \caption{Construction of the $(v,S)$-prohibitor gadget. A dashed line means there is no edge between the two endpoints.}\label{fig:intro-prohib}
\end{figure}

The construction of the matcher is easy if we choose $a$ and $b$ to be non-adjacent vertices that are part of an obstruction. For example, if $a$, $x$, $b$, $y$ is an induced reflexive four-cycle, then a path of 5 vertices with lists $\{a,b\}-\{x,b\}-\{x,y\}-\{b,y\}-\{b,a\}$ is an appropriate matcher. Indeed, the minimum cost 0 cannot be achieved if both endpoints are mapped to $a$.

The splitter can be constructed the following way. Let us choose $w\in S\setminus \{v\}$. As $v$ and $w$ are incomparable, we can choose $v'\in \nh(v)\setminus \nh(w)$ and $w'\in \nh(w)\setminus \nh(v)$. Then the splitter is a four-vertex path with lists 
$S -(V(H)  \setminus \nh(v)) -\{v\} -\{v',w'\}$. The gadget has cost at least 1, as at least one of the two inner vertices has to be deleted. If the first vertex is assigned $v$ and the last vertex is assigned $w'$, then both inner vertices have to be deleted, making the cost 2.

Finally, a short case analysis gives a translator.
Recall from the previous paragraph that $v'\in \nh(v)\setminus \nh(w)$ and $w'\in \nh(w)\setminus \nh(v)$, and, as a case, suppose that $v'$ is not a neighbor of $b$. Then a six-vertex path with lists $\{v',w'\}-\{w\}-\{v'\}-\{b\}-\{a\}-\{a,b\}$ is a translator. At least two of the four inner vertices have to be deleted, meaning that the cost of this gadget is always at least 2. However, if we choose $v'$ on the first vertex and $b$ on the last vertex, then at least three of the inner vertices have to be deleted, raising the cost to 3.

\subsection{Edge-deletion version}
Let us turn our attention now to edge-deletion problems.
While the high-level goals are similar to the vertex-deletion version, the proofs are necessarily more involved: there are two concepts, \emph{bi-arc graphs} and \emph{decompositions} that are relevant only for the edge-deletion version.

\paragraph{Equivalence of $\LHomED(H)$ with classic problems.}
Earlier we have seen that \textsc{Max Cut} and \textsc{Edge Multiway Cut} can be reduced to $\LHomED(H)$ when $H$ is an irreflexive edge or $k$ independent reflexive vertices, respectively. Let us discuss reductions in the other direction. Similarly to the case of \textsc{Odd Cycle Transversal} for vertex deletion, $\LHomED(H)$ is actually equivalent to an annotated generalization of \textsc{Max Cut}, where the two given sets $L$ and $R$ should be on the two sides of the bipartition. However, this annotated generalization is easy to reduce to the original \textsc{Max Cut} problem. Introduce a new vertex $w$ and for every $v\in L$, we connect $w$ and $v$ with  $d(v)$ paths of length 2; for every $v\in R$, we connect $w$ and $v$ with $d(v)$ paths of length 3. We can verify that this extension forces every vertex of $L$ to be on the same side as $w$ and every vertex of $R$ to be on the other side. Furthermore, this extension increases treewidth only by a constant and if the original graph has a \core{\sigma}{\delta} of size $p$, then the constructed graph has a \core{\sigma'}{\delta'} of size $p+1$.

If $H$ consists of $k$ independent reflexive vertices $w_1$, $\dots$, $w_k$, then we can reduce an instance $(G,L)$ of $\LHomED(H)$ to \textsc{Edge Multiway Cut} the following way. Let us extend $G$ to a graph $G'$ by introducing $k$ terminal vertices $t_1$, $\dots$, $t_k$. For every vertex $v\in V(G)$, let us introduce $d(v)$ paths of length 2 between $v$ and $t_i$ if $L(v)$ {\em contains} $w_i$. Suppose now that, in a solution of the multiway cut instance, vertex $v$ is in the component of $t_i$. If $w_i$ is not in $L(v)$, then the solution has to cut all the $|L(v)|\cdot d(v)$ paths. But then we could obtain a solution of the same size by removing all the $d(v)$ original edges incident to $v$ and separating $v$ from all but one terminal by breaking $(|L(v)|-1)\cdot d(v)$ of the paths of length 2. Thus we can assume that vertex $v$ is in the component of some terminal from $L(v)$, showing that we have a reduction from $\LHomED(H)$ to \textsc{Edge Multiway Cut}. We can observe that this transformation increases treewidth at most by an additive constant. Therefore, we can obtain the following lower bound:
\begin{thm}\label{thm:emwaylb}
For every $k\ge 3$ and $\epsilon >0$, there are $\sigma, \delta>0$ such that \textsc{Edge Multiway Cut} with $k$ terminals on $n$-vertex instances given with a tree decomposition of width at most $t$ cannot be solved in time $(k-\epsilon)^t \cdot n^{\bigO(1)}$, unless the SETH fails. 
\end{thm}

\paragraph{Dichotomy for edge deletion.} Feder, Hell, and Huang~\cite{DBLP:journals/jgt/FederHH03} proved that $\LHom(H)$ is polynomial-time solvable for bi-arc graphs and \NP-hard otherwise. Bi-arc graphs are defined by a geometric representation with two arcs on a circle; the precise definition appears in Section~\ref{sec:prelims}.
We start with an alternate characterization of the tractability criterion, which can be obtained using the forbidden subgraph characterization of bi-arc graphs \cite{FederHH07,DBLP:journals/jgt/FederHH03}.  

\begin{restatable}{lem}{biarc}
\label{lem:biarc-intro}
  The following two are equivalent:
  \begin{myenumerate}
\item $H$ does not contain an irreflexive edge, 
		a $3$-vertex set $S$ with private neighbors, or
		a $3$-vertex set $S$ with co-private neighbors. \label{it:biarc1}
\item $H$ is a bi-arc graph that does not contain an irreflexive edge or
  a $3$-vertex set $S$ with private neighbors. \label{it:biarc2}
  \end{myenumerate}
\end{restatable}
With Lemma~\ref{lem:biarc-intro} in hand, the \NP-hardness part of Theorem~\ref{thm:ed-dicho-intro} follows easily. If $H$ is not a bi-arc graph, then already $\LHom(H)$ is \NP-hard; if $H$ contains an irreflexive edge or three vertices with private neighbors, then we can reduce from \textsc{Vertex Cover} or \textsc{Edge Multiway Cut} with 3 terminals, respectively.

Similarly to the proof of Theorem~\ref{thm:vd-dicho-intro}, the polynomial-time part of Theorem~\ref{thm:ed-dicho-intro} is based on a reduction to a flow problem. The fundamental difference is that in the edge-deletion case, there are graphs $H$ such that $i(H)>2$, but $\LHomED(H)$ is polynomial-time solvable (an example of such a graph is $H[A]$ from Figure~\ref{fig:dec-intro}). Thus, even if we assume that the list of a vertex is an incomparable set, it can have size larger than 2. Therefore, a simple reduction to $s\hyph t$ \textsc{Min Cut} where placing a vertex $v$ on one of two sides of the cut corresponds to the choice between the two elements of the list $L(v)$ cannot work. Instead, we represent each vertex $v$ with multiple vertices. Let $\ell=|L(v)|$. We represent vertex $v$ with a directed path on $\ell+1$ vertices, where we enforce (with edges of large cost) that the first and last vertices are always on the right and left side of the cut. We imagine the edges of the path to be undeletable, for example, each edge has large weight, implying that a minimum weight $s\hyph t$ cut would not remove any of them. This means that the path has $\ell$ possible states in a minimum $s\hyph t$ cut: the only possibility is that for some $i\in [\ell]$, the first $i$ vertices of the path are on the right side (the side corresponding to $t$), and the remaining $\ell+1-i$ vertices are on the left side (corresponding to $s$). Based on the geometric representation on the bi-arc graph $H$, we define an ordering $L(v)=\{a_1,\dots,a_\ell\}$ of each list. The idea is that assigning $a_i$ to $v$ corresponds to the state where the first $i$ vertices of the path are on the right side of the cut.

To enforce this interpretation, whenever $u$ and $v$ are adjacent vertices in $G$, we introduce some edges between the paths representing $u$ and $v$. These edges are introduced in a way that faithfully represents the \emph{interaction matrix} of $u$ and $v$, which is defined as follows. Let $L(u)=\{a_1,\dots, a_{\ell_u}\}$ and $L(v)=\{b_1,\dots, b_{\ell_v}\}$ in the ordering of the lists. The interaction matrix of $u$ and $v$ is a $|L(u)|\times |L(v)|$ matrix where the element in row $i$ and column $j$ is 1 if $a_ib_j\in E(H)$, and 0 otherwise.

Figure~\ref{fig:pathinteraction} (left) shows an example where $L(u)=\{a_1,\dots,a_5\}$, $L(v)=\{b_1,\dots, b_4\}$, and the interaction matrix is as show in the figure, i.e., the 0s form a rectangle in the top-right corner. Then we introduce an edge from the fourth vertex of the path of $u$ to the third vertex of the path of $v$. In the minimum $s\hyph t$ cut problem, this edge has to be removed whenever the tail of the edge is on the left side and the head of the edge is on the right side, which corresponds to assigning one of $\{a_1,a_2,a_3\}$ to $u$ and one of $\{b_3,b_4\}$ to $v$. Therefore, we need to remove this edge if and only if the states of the two paths correspond to a 0 entry in the interaction matrix, that is, when the edge $uv$ has to be removed since its image is not an edge of $H$. This means that this single edge indeed faithfully represents this particular interaction matrix.

\begin{figure}
\centering
\includegraphics[scale=1.1,page=2]{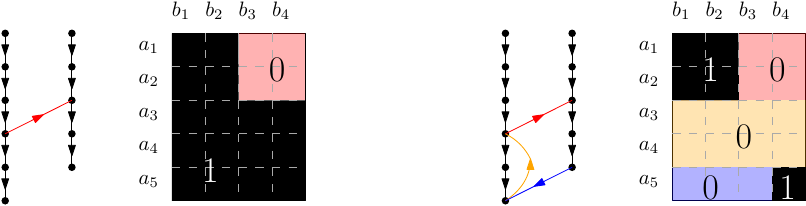}
  \caption{Representing the interaction of two vertices $u$ and $v$ with
    $L(v)=\{a_1,\dots,a_5\}$ and $L(u)=\{b_1,\dots, b_4\}$. Black areas denote ones in the interaction matrix.}\label{fig:pathinteraction}
\end{figure}

Through a detailed analysis of bi-arc graphs without irreflexive edges and 3-vertex sets with private neighbors, we determine how interaction matrices can look like. It turns out that the 0s in the matrix can be partitioned into at most three ``nice'' rectangles: appearing in the top-right corner, appearing in the lower-left corner, or having full width $|L(v)|$ (see Figure~\ref{fig:pathinteraction}, right). Each such nice rectangle can be represented by an edge, thus every interaction matrix can be represented by at most three edges such that in the solution we need to remove at most one of them.

\paragraph{Algorithms on bounded-treewidth graphs.} As discussed on \cpageref{page:ed-alg}, if $H$ has a decomposition as in Definition~\ref{def:decomp}, then $\LHomED(H)$ can be reduced to an instance of $\LHomED(H_1)$ and an instance of $\LHomED(H_2)$, where $H_1=H[A]$ and $H_2=H[B\cup C]$. It follows that if we use the $i(H)^t \cdot n^{\bigO(1)}$-time algorithm whenever $H$ is undecomposable, then we obtain an $\param(H)^t \cdot n^{\bigO(1)}$-time algorithm for every $H$. Furthermore, proving that there is no $(i(H)-\epsilon)^t \cdot n^{\bigO(1)}$-time algorithm for undecomposable $H$ proves that there is no $(\param(H)-\epsilon)^t \cdot n^{\bigO(1)}$-time algorithm for arbitrary $H$.

\paragraph{Reductions using gadgets.}
For the lower bound of Theorem~\ref{thm:ed-main-intro}, it is sufficient to prove the statement under the assumption that $H$ is undecomposable, hence $\param(H)=i(H)$.
For $q\ge 3$, we prove the lower bound by a reduction from  \coloringED{q}, whose hardness was established in \cref{thm:ed-coloring-intro}. As in the vertex-deletion case, we use gadgets that allow a straightforward reduction and the construction of these gadgets is the core technical part of the proof.

Here, a gadget is an instance with a set of distinguished vertices called portals.
Defining the intended behavior of gadgets is neater in the edge-deletion case as the possibility of deleting portals does not complicate matters. For every assignment of the portals, the \emph{cost} of the assignment is the minimum number of edges that needs to be deleted if we want to extend the assignment to the rest of the gadget. We can use the gadget to enforce that the assignment of the portals is one of minimum cost. Therefore, in order to reduce $\LHomED(K_q)$ to $\LHomED(H)$, we choose an incomparable set $S$ of size $q$ and design a gadget that has two portals $(p_1,p_2)$ with $L(p_1)=L(p_2)=S$, and every assignment $f$ with $f(p_1)\neq f(p_2)$ has cost exactly $\alpha$, while every assignment with $f(p_1)=f(p_2)$ has cost $\beta$ strictly more than $\alpha$. We replace every edge of the original graph $G$ with such a gadget. It is clear that the constructed $\LHomED(H)$ instance has a solution of cost $\alpha|E(G)|$ if and only if $G$ is $q$-colorable.

\paragraph{Realizing relations.} Let $c,d\in V(H)$ be two vertices and let $R\subseteq \{c,d\}^r$ be an arbitrary $r$-ary relation. We would like to prove a general statement saying that every such relation can be realized by some gadget: there is a gadget with $r$ portals such that
\begin{itemize}
\item the list of each portal vertex is $\{c,d\}$,
  \item an an assignment on the portal vertices has cost exactly $\alpha$ if it corresponds to a vector in $R$, and
  \item the cost of every other assignment is $\beta>\alpha$.
  \end{itemize}
  We show that if $c$ and $d$ are two vertices chosen from one of the obstructions appearing in Lemma~\ref{lem:biarc-intro}\,(1) (irreflexive edges, three vertices with private or co-private neighbors), then such a gadget representing $R\subseteq \{c,d\}^r$ can indeed be constructed. 
  Crucially, this requires to construct some gadget that realizes the ``Not Equals'' relation on $\{c,d\}$, i.e., $\NEQ=\{(c,d), (d,c)\}$. With $\NEQ$ in hand, we use an earlier result from~\cite{EsmerFMR24hub} for the list coloring problem, which shows that $\NEQ$ can be used to model arbitrary relations.
  Note that this is the point where we use the assumption that we are in the \NP-hard case of Theorem~\ref{thm:ed-dicho-intro} (which we definitively have to exploit at some point): We exploit the structure of an obstruction to model $\NEQ$ on two of its vertices.

  \paragraph{Indicators.} Our next goal is to construct  \emph{indicator gadgets,} defined as follows. The gadget has $\lambda+1$ portals for some constant $\lambda$. Portal $p$ has list $S$, and the remaining $\lambda$ portals have list $\{c,d\}$. Let $\alpha$ be the minimum number of edge deletions that are needed in the gadget. We can think of $p$ as the input and the rest of the portals as the outputs. If we are interested only in solutions where exactly $\alpha$ deletions are made inside the gadget, then assigning a value $a$ to the input is compatible with some set $I(a)\subseteq \{c,d\}^\lambda$ of assignments on the outputs. The indicator gadget has two properties: (1) $I(a)$ is non-empty for any $a\in S$ and (2) $I(a)\cap I(b)=\emptyset$ for any two distinct $a,b\in S$.

If we can construct indicator gadgets, then we can construct the gadget needed to reduce from  \textsc{$q$-Coloring} (that is, expressing $f(p_1)\neq f(p_2)$) in the following way. Let us introduce two copies of the indicator gadget on vertices $(p_1,u_1,\dots,u_\lambda)$ and on $(p_2,v_1,\dots,v_\lambda)$. We have $L(p_1)=L(p_2)=S$ and $L(u_i)=L(v_i)=\{c,d\}$ for $i\in [\lambda]$. Then we define an appropriate $2\lambda$-ary relation $R\subseteq \{c,d\}^{2\lambda}$, realize it with a gadget as discussed above, and then put this gadget on the vertices $\{u_1,\dots,u_p,v_1,\dots,v_p\}$. We define the relation $R$ such that it rules out for any $a\in S$ that the assignment on $(u_1,\dots,u_\lambda)$ is from $I(a)$ and the assignment on $(v_1,\dots,v_\lambda)$ is also from $I(a)$; as we can realize any relation $R$, we can certainly realize such a gadget. Then this gadgets enforces, for any $a$, that the value $a$ cannot appear on both $p_1$ and $p_2$ simultaneously, but allows every other combination. 

We construct indicator gadgets for $\lambda=|S|(|S|-1)$. For every pair $(a,b)$ of distinct vertices from $S$, we construct a subgadget with two portals $(q_1,q_2)$ with $L(q_1)=S$ and $L(q_2)=\{a',b'\}$ for some $a',b'\in V(H)$, and satisfying the following:
\begin{myenumerate}
\item assigning $a$ on $q_1$ forces $a'$ on $q_2$.
\item assigning $b$ on $q_1$ forces $b'$ on $q_2$.
  \item for any $e\in S\setminus \{a,b\}$, assigning $e$ on $q_1$ allows at least one of $a'$ or $b'$ on $q_2$. 
  \end{myenumerate}
  We construct $|S|(|S|-1)$ such subgadgets --- one for every distinct $(a,b)$. The construction of these subgadgets is fairly simple, but in general the pair $(a',b')$ can be different for every pair $(a,b)$. If we were so lucky that every pair $(a',b')$ is actually $(c,d)$, then we would be done with the construction of the indicator. In this case, we can simply join these $|S|(|S|-1)$ subgadgets at $q_1$ to obtain a gadget with input $q_1$ and $|S|(|S|-1)$ output vertices. Now it is clear that if we assign values $a$ and $b$ to the input, then they cannot be compatible with the same assignment on the output vertices: in the subgadget corresponding to pair $(a,b)$, value $a$ on the input forces $c$ on the output, while value $b$ forces $d$ on the output.

  In general, however, we cannot expect $(a',b')$ to be the same pair $(c,d)$ for every choice of $(a,b)$. Therefore, the final component is a gadget that ``moves'' an arbitrary pair $(a',b')$ to $(c,d)$.

  \paragraph{Moving pairs.} We say that there is an \emph{$(a,b)\to (c,d)$ move} if there is a gadget with two portals $(x,y)$ with $L(x)=\{a,b\}$, $L(y)=\{c,d\}$, and the following property: assigning $a$ (resp., $b$) to $x$ forces $c$ (resp., $d$) on $y$. In most cases, it is not very important to us which of $a$ and $b$ is mapped to $c$ or $d$, only the uniqueness of the mapping is important. Therefore, we introduce the notation \emph{$\{a,b\}\leadsto\{c,d\}$ move} to mean either an $(a,b)\to (c,d)$ move or an $(a,b)\to (d,c)$ move. The main result is that if the graph $H$ is undecomposable, then we can have such moves between any two pairs of incomparable vertices.
\begin{restatable}{lem}{lemmovingmain}
  \label{lem:movingMain}
	Let $H$ be an undecomposable graph.
	Let $\{a,b\}$ and $\{c,d\}$ be (not necessarily disjoint) $2$-vertex incomparable sets in $H$. Then $\{a,b\}\leadsto \{c,d\}$.
\end{restatable}
The assumption that $H$ is undecomposable is essential here: one can observe that if there is a decomposition $(A,B,C)$ and $a,b\in A$ and $c,d\in B$, then 
an $\{a,b\}\to \{c,d\}$ move cannot exist: intuitively, we cannot transmit information through the complete connection between $A$ and $B$.

The first step of the proof is to show that such a move exists if the 2-vertex incomparable sets intersect: that is, there is a $\{a,b\}\leadsto \{a,c\}$ move whenever $\{a,b\}$ and $\{a,c\}$ are incomparable sets. This suggests defining the following auxiliary graph $\aux(H)$: the vertices of $\aux(H)$ correspond to $2$-vertex incomparable sets, and two such vertices are connected if they represent pairs that intersect. Our main goal is showing that (a large part of) $\aux(H)$ is connected. As discussed above, the proof has to use the fact that $H$ is undecomposable. We consider two cases depending on whether $H$ is a \emph{strong split graph} or not, that is, whether it can be partitioned into a reflexive clique and an irreflexive independent set. The way we can exploit the non-existence of decompositions depends on whether $H$ is in this class or not.

\paragraph{Case I: strong split graphs.} In the case of a strong split graph, the following algorithm can be used to detect if there is a non-trivial decomposition. Let us assume that $H$ does not have universal or independent vertices. We say that a vertex is \emph{maximal} if its neighborhood is inclusionwise maximal, that is, there is no vertex that is adjacent to a proper superset of the neighborhood. The key observation is that every maximal vertex has to be in part $B$ of the decomposition. Therefore, we initially
\begin{itemize}
\item move every maximal vertex into $B$ and move every other vertex to $A$.
\end{itemize}
Then we repeatedly apply the following two steps as long as possible:
\begin{itemize}
\item If $v\in A$ is irreflexive and not adjacent to some vertex in $B$, then we
  move $v$ into $C$.
\item If $v\in A$ is reflexive and adjacent to $C$, then we move $v$ into $B$.
  \end{itemize}
  It can be checked that the algorithm is correct: if it stops with a non-empty set $A$, then $(A,B,C)$ is a valid decomposition. Thus the assumption that $H$ has no decomposition implies that the algorithm moves every vertex to $B\cup C$.

  Consider an incomparable pair $\{a,b\}$ that we want to move to $\{c,d\}$. It is sufficient to consider only the case where $a$ and $b$ are both reflexive. The algorithm eventually moves $a$ to $B$, and there is a sequence of moves that certify this. That is, there is a sequence $\ell_0,r_1,\ell_1,r_1,\dots,r_k,\ell_k$ such that $\ell_0=a$, $\ell_k$ is a maximal reflexive vertex, $\ell_i$ is a reflexive vertex adjacent to $r_{i+1}$, and $r_i$ is an irreflexive vertex not adjacent to $\ell_i$ (see Figure~\ref{fig:alternating}). If we choose this alternating path certificate to be of minimal length, then $\ell_i$ and $\ell_{i+1}$ are incomparable: $r_{i+1}$ and $r_{i+2}$ are adjacent to exactly one of them. Therefore, the pairs $\{\ell_i,\ell_{i+1}\}$ and $\{\ell_{i+1},\ell_{i+2}\}$ are adjacent in $\aux(H)$, implying that $\{a,b\}$ is in the same component of $\aux(H)$ as $\{\ell_{k-1},\ell_k\}$. If $q$ is some maximal vertex with a neighborhood distinct from $\ell_k$, then $\{\ell_k,q\}$ is also incomparable, and it is adjacent to $\{\ell_{k-1},\ell_k\}$. The conclusion is that every incomparable pair $\{a,b\}$ is in the same component as some pair $\{a',b'\}$ of incomparable \emph{maximal} vertices. Therefore, it is sufficient to show that whenever $\{a',b'\}$ and $\{c',d'\}$ are two pairs of incomparable vertices such that $a',b',c',d'$ are all maximal, then $\{a',b'\}$ and $\{c',d'\}$ are in the same component of $\aux(H)$. Then at least one of $\{a',d'\}$ or $\{a',c'\}$ is incomparable (depending on whether $\nh(a')=\nh(c')$ or not). Either of these pairs is adjacent to both $\{a',b'\}$ and $\{c',d'\}$.
  \begin{figure}
\centering
\includegraphics[scale=1.1,page=1]{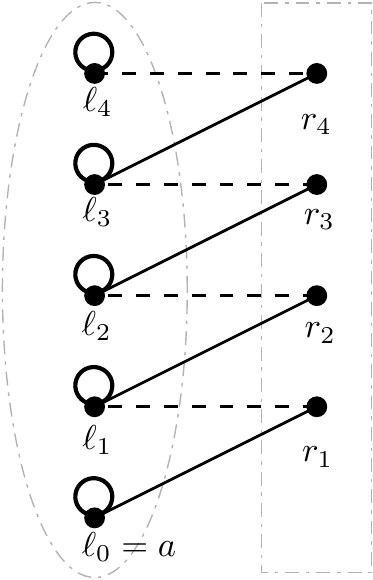}	  
\caption{An alternating path certifying that $a$ is moved to $B$. Vertex $\ell_4$ is maximal.}\label{fig:alternating}
\end{figure}    
  \paragraph{Case II: graphs that are not strong split graphs.}
  If $H$ is not a strong split graph, then either it contains two adjacent irreflexive vertices, or two non-adjacent reflexive vertices. We can find a decomposition the following way. Initially, we
  \begin{itemize}
  \item put into $A$ every reflexive vertex that is not adjacent to some other reflexive vertex, and
  \item put into $A$ every irreflexive vertex that is adjacent to some other irreflexive vertex.
  \end{itemize}
  Then we repeat the following two steps as long as possible:
  \begin{itemize}
\item If $v\not\in A$ is irreflexive and adjacent to $A$, then we move $v$ into $A$.
\item If $v\not\in A$ is reflexive and not adjacent to some vertex in $A$, then we move $v$ into $A$.
\end{itemize}
Again, we can verify that if the algorithm stops without moving every vertex to $A$, then we have a non-trivial decomposition. Therefore, for every vertex $a$, the algorithm provides a sequence of moves that certifies that $a$ has to be in part $A$ of any decomposition. Similarly to the previous case, we can use such a (minimal) certificate to show that every $(a,b)$ is in the same component of $\aux(H)$ as some $(a',b')$, where $a'$ and $b'$ are either adjacent irreflexive vertices or non-adjacent reflexive vertices. Therefore, all that is left to show is that if both $(a',b')$ and $(c',d')$ have this property, then they are in the same component of $\aux(H)$. This can be proved with a short case analysis.

\section{Preliminaries}\label{sec:prelims}
For an integer $k$, by $[k]$ we denote $\{1,\ldots,k\}$. For a set $X$, by $2^X$ we denote the family of all subsets of $X$.

\paragraph{Graph theory.}
Let $H$ be a graph. By $V(H)$ and $E(H)$ we denote, respectively, the vertex set and the edge set of $H$.
A vertex of $H$ is irreflexive (resp. reflexive) if it does not have a loop (resp. has a loop).
A graph is irreflexive (resp. reflexive) if its every vertex is irreflexive (resp. reflexive).

For a graph $H=(V,E)$ and a set $X \subseteq V$, by $H[X]$ we denote the graph induced by $X$, 
i.e., $(X, \{ e \in E ~|~ e \subseteq X\}$. By $G-X$ we denote the graph obtained form $H$ by removing all vertices in $X$ along with incident edges, i.e., $H[V \setminus X]$. For a set $X \subseteq E$, by $H \setminus X$ we denote the graph obtained by removing all edges in $X$, i.e., $(V,E \setminus X)$.

Given a graph $H=(V,E)$ and a vertex $v\in V$, by $\nh(v)$ we denote the neighborhood of $v$ in $H$.
Note that $v \in \nh(v)$ if and only if $v$ is reflexive. A vertex is \emph{isolated} if its neighborhood is empty.
We say that $v$ \emph{dominates} $u\in V$ if $\nh(u) \subseteq \nh(v)$. 
If equality does not hold then $v$ \emph{strictly} dominates $u$.
The vertex $v$ is \emph{maximal} if there is no vertex in $V$ that strictly dominates $v$. Finally, a vertex is \emph{universal} if it is adjacent to every vertex in $H$ (including itself).

If $u$ dominates $v$ or $v$ dominates $u$ then $u$ and $v$ are \emph{comparable}. Otherwise, they are \emph{incomparable}. A set $S\subseteq V$ is \emph{incomparable} if its vertices are pairwise incomparable. Conversely, $S$ is \emph{comparable} if the vertices of $S$ are pairwise comparable.
Recall that by $i(H)$ we denote the size of a largest incomparable set in $H$.

\paragraph{\boldmath Treewidth and \core{\sigma}{\delta}s.}
Consider a graph with a \core{\sigma}{\delta} $Q$ of size $p$. Introducing a bag that contains $Q$ that is the center of a star whose leaves are $Q\cup C_i$ for each connected component $C_i$ of $G-Q$. Then this is a tree decomposition of $G$ of width at most $p+\sigma-1$. We state this observation formally.
\begin{obs}\label{obs:coretotw}
	For some $\sigma, \delta\ge 1$, let $G$ be a graph given with a \core{\sigma}{\delta} of size $p$. One can obtain a tree decomposition of width less than $p+\sigma$ in time polynomial in the size of $G$.
\end{obs}

\paragraph{Private neighbors, co-private neighbors.}
Let $U$ be a set of vertices and let $u,v\in U$. We say that $u$ has a \emph{private neighbor with respect to $U$} if there is $u'\in \nh(u)\setminus \bigcup_{v\in U, v\neq u} \nh(v)$. We say that $u,v$ have a \emph{co-private neighbor with respect to $U$} if there is a vertex $u'\in \left(\nh(u)\cap \nh(v)\right)\setminus \bigcup_{w\in U\setminus\{u,v\}} \nh(w)$. If the set $U$ is clear from context we may not explicitly specify with respect to which set the vertices have private/co-private neighbors.
The set $U$ has \emph{private neighbors} if each $u\in U$ has a private neighbor (with respect to $U$). Similarly, $U$ has \emph{co-private neighbors} if each pair of distinct vertices $u,v\in U$ has a co-private neighbor.

\paragraph{Variants of the list homomorphism problem.}
In this paper we consider three computational problems denoted by \LHom($H$), \LHomVD($H$), and \LHomED($H$).
In each of the problems $H$ is a fixed graph, and the instance is a pair $(G,L)$, where $G$ is a graph and $L : V(G) \to  2^{V(H)}$ is a \emph{list} function. 

In the \LHom($H$) problem we ask whether there exists a homomorphism $\phi: V(G) \to V(H)$ which respects lists $L$, i.e., $\phi(v) \in L(v)$ for all $v \in V(G)$.
In the \LHomVD($H$) (resp. \LHomED($H$)) we ask for a smallest set $X$ of \emph{vertices} (resp. \emph{edges}) such that $G-X$ (resp. $G \setminus X$) admits a homomorphism to $H$ that respects lists $L$ (``VD'' in the name of the problem stands for ``Vertex Deletion'' and ``ED'' stands for ``Edge Deletion'').

Note that \LHomVD($H$) and \LHomED($H$) are optimization problems. Sometimes it will be convenient to consider their corresponding decisions versions, when we are additionally given an integer $k$ and we ask whether the instance graph can be modified into a yes-instance of \LHom($H$) by removing at most $k$ vertices/edges.

Let $(G,L)$ be an instance of one of the problems defined above.
Suppose that there is a vertex $x \in V(G)$ such that $L(x)$ contains two vertices $u,v$ with $\nh(u) \subseteq \nh(v)$.
We claim that we can safely remove $u$ from $L(x)$.
Indeed, note that in any homomorphism that maps $x$ to $u$ we can safely remap $x$ to $v$, without changing the images of other vertices.
Thus without loss of generality we can assume that in each of the three problems, each list is an incomparable set.
In particular, the size of each list is at most $i(H)$.

\paragraph{Gadgets.}
Let $J$ be a graph together with a list assignment $L\from V(J)\to 2^{V(H)}$ and $r$ distinguished vertices $\boldx=(x_1,\ldots, x_r)$ from $J$. Then we refer to the tuple $\calJ=(J, L, \boldx)$ as a \emph{$r$-ary $H$-gadget}.
We might not specify $r$ nor $H$ in case they are clear from the context.
The vertices $x_1,\ldots,x_r$ are called \emph{portals}.

In most cases our gadgets will be just paths with portals in endvertices.
In such a case we introduce the following abbreviated notation.
Let $\calJ=(J,L,(x,y))$ be a binary gadget of arity $2$ such that $J$ is a path with consecutive vertices $x,p_1,\ldots,p_\ell,y$. We refer to $\calJ$ just by specifying the lists of consecutive vertices, i.e., $L(x) - L(p_1) -\ldots -L(p_\ell) - L(y)$.

\section{\boldmath Complexity dichotomy for \LHomVD($H$)}
\label{sec:dicho-vd}
The aim of this section is to prove the following result.

\thmvddicho*

Before we proceed to the proof, let us show alternative characterizations of the graphs considered in \cref{thm:vd-dicho-intro}.

\vdalternative*
\begin{proof}
\noindent\textbf{(\ref{it:i2}.$\to$\ref{it:obstructions}.)}
This implication is trivial, as each of the structures listed in statement \eqref{it:obstructions} contains an incomparable set of size 3.

\noindent\textbf{(\ref{it:obstructions}.$\to$\ref{it:covered}.)}
In what follows we will use some facts from graph theory,
and formally we apply them to the graph obtained from $H$ by removing all loops.

Note that every induced cycle with at least six vertices contains an independent set of size 3,
thus the only induced cycles in $H$ are triangles.
Furthermore $H$ does not contains an \emph{asteroidal triple}: three independent vertices, so that each pair is joined with a path avoiding the neighborhood of the third one.
It is well known that graphs with no induced cycles of length at least 4 and no asteroidal triples are precisely interval graphs~\cite{Lekkeikerker1962}. Thus $H$ is an interval graph with maximum independent set of size at most 2.

Interval graphs are perfect, so their complements are perfect as well.
As the complement of $H$ is triangle-free, it is bipartite.
This means that the vertex set of $H$ can be covered with two cliques.

\noindent\textbf{(\ref{it:covered}.$\to$\ref{it:i2}.)}
Suppose that the vertex set of $H$ can be covered with two cliques $\calL$ and $\calR$.
Consider some intersection representation of $H$ by segments on a line.
As segments on a line satisfy the Helly property, there is a point $\ell$ contained in all segments from $\calL$ 
and a point $r$ contained in all segments from $\calR$. Without loss of generality assume that $\ell$ is to the left of $r$.
Observe that we can trim the segments from $\calL$ so that their left endpoint is $\ell$,
and the segments from $\calR$ so that their right endpoint is $r$, obtaining another intersection representation of $H$ by segments on a line.

Let $v_1,v_2,\ldots,v_k$ be the vertices from $\calL$, ordered increasingly with respect to the length of their corresponding segments (ties are resolved arbitrarily). As all these segments share their left endpoint, we observe that for $i < j$ the segment representing $v_i$ is contained in the segment representing $v_j$. Consequently, we obtain that $\nh(v_1) \subseteq \nh(v_2) \subseteq \ldots \subseteq \nh(v_k)$ (here we use the fact that all vertices are reflexive, so we always have $v_i \in \nh(v_i)$).
By analogous reasoning for $\calR$ we observe that each of cliques $\calL,\calR$ is a comparable set of vertices.
Consequently $i(H) \leq 2$.

This completes the proof.
\end{proof}

Now we can proceed to the proof of the complexity dichotomy for \LHomVD($H$).

\begin{proof}[Proof of \cref{thm:vd-dicho-intro}.]
Let us begin with the hard cases.

\paragraph{Hardness.}
If $H$ contains an irreflexive vertex $z$, then solving \LHomVD($H$) on a graph where all lists are set to $\{z\}$ is equivalent to solving \textsc{Max Independent Set}.

So let us assume that $H$ is reflexive. By \cref{lem:vd-alternative-intro}. we can assume that is contains one of the three subgraphs given in item \ref{it:obstructions} of the lemma.

If $H$ contains three pairwise non-adjacent reflexive vertices $a,b,c$, then there is a straightforward reduction from \textsc{Vertex Multiway Cut} with three terminals.
We take the same graph and set the list of every terminal to $\{a\}, \{b\}$, and $\{c\}$, respectively.
The lists of all non-terminal vertices are $\{a,b,c\}$.

If $H$ contains a four- or five-vertex reflexive cycle, then it is already \NP-hard to decide whether an input graph admits a list homomorphism to $H$ without deleting any vertices~\cite{FEDER1998236}. 

\paragraph{Algorithm.}
So from now on let us suppose that $H$ is reflexive and $i(H) \leq 2$.
By~\cref{lem:vd-alternative-intro} we observe that the vertex set of $H$ can be covered by two cliques $\calL$ and $\calR$.
As it was already observed in the proof of \cref{lem:vd-alternative-intro}, if the vertices of a reflexive interval graph can be covered by two cliques, then each of these cliques forms a comparable set.

Let $(G,L)$ be an instance of \LHomVD($H$), where $G$ has $n$ vertices, and consider $x \in V(G)$.
For a homomorphism $\phi$ from $G$ to $H$, we say that $x$ is mapped to the \emph{left} (resp. \emph{right}) if $\phi(x) \in \calL$ (resp. $\phi(x) \in \calR$.
Note that we can safely assume that $L(x) \neq \emptyset$, as we always need to delete such a vertex $x$.
Thus $|L(x)| \in \{1,2\}$.
If $|L(x)|=1$, we say that $x$ is \emph{decided}: we know in advance whether $x$ will be mapped to the left or to the right (if $x$ is not deleted).
If $|L(x)|=2$, we say that $x$ is \emph{undecided}. Recall the list of each undecided vertex contains one element from $\calL \setminus \calR$ and one element from $\calR \setminus \calL$. Let $m$ be the number of undecided vertices in $G$.

Now the problem boils down to deciding, for each undecided vertex, whether it will be mapped to the left or to the right.
Let $xy$ be an edge of $G$. Note that mapping both $x$ and $y$ to the same side always satisfies the constraint given by the edge $xy$.
However, sometimes we cannot map $x$ and $y$ to different sides. This depends on the existence of the corresponding edge in $H$ and is entirely determined in advance by the lists of $x$ and $y$.
We say that the ordered pair $(x,y)$ of vertices in $G$ is \emph{left-right-incompatible} if the following conditions are satisfied:
\begin{myenumerate}
\item $x$ is adjacent to $y$ in $G$,
\item $\calL \cap L(x) \neq \emptyset$ and $\calR \cap L(y) \neq \emptyset$,
\item denoting $\{a\} = \calL \cap L(x)$ and $\{b\} = \calR \cap L(y)$, we have $ab \notin E(H)$.
\end{myenumerate}
Intuitively, a pair $(x,y)$ is left-right-incompatible if we cannot simultaneously map $x$ to the left and $y$ to the right,
but looking at the list of each vertex separately this possibility is not ruled out.

We aim to reduce the problem to solving the minimum vertex separator problem:
given a digraph $D$ with \emph{source} $s$ and \emph{sink} $t$, find a smallest set $X \subseteq V(D) \setminus \{s,t\}$ such that there is no $s$-$t$-path in $D - X$. This problem can be solved using one of algorithms for finding a maximum flow.

The problem can be equivalently stated as follows.
Given a digraph $D$ with specified vertices $s,t$, partition $V(D)$ intro three sets $S,X,T$, such that
\begin{myitemize}
\item $s \in S$ and $t \in T$,
\item there is no arc beginning in $S$ and ending in $T$.
\item $X$ is of minimum possible size.
\end{myitemize}
This formulation will be more useful to us.
We will create an instance $(D,s,t)$ of the minimum vertex separator problem where the minimum $s$-$t$-separator has size at most $k$ if and only if $(G,L)$ can be modified into a yes-instance of \LHom($H$) by removing at most $k$ vertices.

We start the construction of $D$ with introducing the set $V(G)$ and two additional vertices $s$ and $t$.
In our intended solution, the vertices of $G$ mapped to the left (resp. right) will be in $S$ (resp. $T$), while the deleted vertices will be in $X$.

Consider a vertex $x \in V(G)$. If $x$ is decided and $L(x) \subseteq \calL$, we add the arc $\overrightarrow{sx}$.
This way we make sure that $x$ can either be mapped to the left (i.e., $x \in S$) or deleted (i.e., $x \in X$).
Similarly, if  $L(x) \subseteq \calR$, we add the arc $\overrightarrow{xt}$.

The only thing left is to take care of incompatible pairs.
Let $(x,y)$ be a left-right-incompatible pair in $G$.
Then we add the arc $\overrightarrow{xy}$.
Note that this ensures that we cannot have $x \in S$ and $y \in T$,
i.e., we cannot simultaneously map $x$ to the left and map $y$ to the right.
This completes the construction of $(D,s,t)$. The correctness of the reduction follows from the description above.
\end{proof}

\section{\boldmath Tight results for \LHomVD($H$)}
\label{sec:vd-tight}
In this section we prove \cref{thm:vd-main-intro-tw} and \cref{thm:vd-main-intro}.

\subsection{\boldmath  Algorithm for \LHomVD($H$)}\label{sec:vd-algo}

First we show the algorithmic part of \cref{thm:vd-main-intro-tw}.

\begin{customthm}{\ref{thm:vd-main-intro-tw}~(a)}\label{thm:vd-algo}
	Let $H$ be a fixed graph. Then every $n$-vertex instance of $\LHomVD(H)$ given along with a tree decomposition of width $t$ can be solved in time $(i(H)+1)^t \cdot n^{\bigO(1)}$.
\end{customthm}
\begin{proof}
	Consider an instance $(G,L)$ of \LHomVD($H$), given along with a tree decomposition of width $t$.
	Recall that without loss of generality we may assume that each list is an incomparable set, i.e., the size of each list is at most $i(H)$. Now the algorithm is a straightforward dynamic programming on the given tree decomposition: For each vertex of $G$, we have at most $i(H)+1$ possible states --- mapping it to one of at most $i(H)$ vertices from its list, or deleting it.
\end{proof}

\subsection{\boldmath  Hardness for \LHomVD($H$)}\label{sec:vd-lb}

In this section we show \cref{thm:vd-main-intro}; recall that it implies the hardness counterpart of \cref{thm:vd-main-intro-tw}.
Recall that a special case of  \cref{thm:vd-main-intro} where $H$ is an irreflexive clique, even in non-list variant, is given in \cref{thm:vd-coloring-intro}.

\thmvdcoloringmain*

Let $H$ be a fixed graph.
Recall that a binary $\calJ = (J,L,(x,y))$ is a graph with lists $L : V(J) \to 2^{V(H)}$ and two distinguished vertices $x,y$ called portals.
We treat $(J,L)$ as an instance of \LHomVD($H$), in particular each portal can be mapped to a vertex from its list or be deleted.
To simplify the notation, we introduce a special symbol $\del$: if we say that a vertex $v$ is mapped to $\del$, we mean that $v$ is deleted.
A useful way of thinking of it is to imagine that we are considering homomorphism to the graph obtained from $H$ by adding a new universal vertex $\del$.

Let $(J,L,(x,y))$ be a binary gadget. For $(a,b) \in (L(x) \cup \{\del\}) \times (L(y) \cup \{\del\})$,
by $\vdcount(\calJ \to H, (a,b))$ we denote the size of a smallest set $X \subseteq V(J) \setminus \{x,y\}$
for which there exists a list homomorphism $\phi$ from $J-X$ to $H$ such that $\phi(x)=a$ and $\phi(y)=b$
(recall that mapping a portal to $\del$ actually denotes deleting this portal).
The \emph{base cost} of $\calJ$ is $\vdcount(\calJ \to J, (\del,\del))$.
Note that if the base cost of $\calJ$ is $\alpha$, then for every $(a,b) \in (L(x) \cup \{\del\}) \times (L(y) \cup \{\del\})$ we have $\vdcount(\calJ \to H, (a,b)) \geq \alpha$.
Intuitively, every pair $(a,b) \in L(x) \times L(y)$ such that $\vdcount(\calJ \to H, (a,b))=\alpha$ is ``free'', and for others we need to ``pay'' by deleting some vertices of $\calJ$.

The crucial gadget used in the proof of \cref{thm:vd-main-intro} is called a \emph{prohibitor}.

\begin{defn}[$(v,S)$-prohibitor]\label{def:prohibitor}
Let $H$ be a graph, $S$ be an incomparable subset of $V(H)$, and $v $ be a vertex in $S$.
A binary gadget $\calF_v=(F_v,L,(x,y))$ with base cost $\alpha$, where $f_v$ is a path with endvertices $x \neq y$, is a \emph{$(v,S)$-prohibitor} if the following properties hold
\begin{myenumerate} 
\item $L(x)=L(y) = S$,
\item for all  $(a,b) \in (S \cup \{\del\})^2 \setminus (v,v)$ it holds that $\vdcount(\calF_v \to H, (a,b)) = \alpha$,
\item $\vdcount(\calF_v \to H, (v,v)) > \alpha$.
\end{myenumerate}
\end{defn}

The following result is the main technical lemma used in the proof of \cref{thm:vd-main-intro}.

\begin{restatable}{lem}{lemprohibitor}
\label{lem:vd-gadget}
Let $H$ be a fixed graph which contains either an irreflexive vertex or a three pairwise non-adjacent reflexive vertices or an induced reflexive cycle on four or five vertices.
Let $S$ be an incomparable set of size at least 2 and let $v \in S$.
Then there exists a $(v,S)$-prohibitor.
\end{restatable}

Let us postpone the proof of \cref{lem:vd-gadget} and first prove \cref{thm:vd-main-intro}, assuming \cref{lem:vd-gadget}.
Suppose $S$ is an incomparable set in $H$ of size at least 2, and for every $v \in S$ we are given a $(v,S)$-prohibitor $\calF_v=(F_v,L,(x_v,y_v))$ with base cost $\alpha_v$.
An \emph{$S$-prohibitor} is a binary gadget $\calF=(F,L,(x,y))$, where graph $F$ obtained by identifying all vertices $x_v$ into a vertex $x$ and all vertices $y_v$ into a vertex $y$. Note that portals of each $\calF_v$ are non-adjacent, so this identification does not introduce any multiple edges. Furthermore, the lists of all portals are the same, i.e., $S$.
The base cost of $\calF$ is $\alpha := \sum_{v \in S} \alpha_v$.
The following properties of the $S$-prohibitor follow directly from the properties of $(v,S)$-prohibitors.
\begin{myenumerate}[(P1)]
\item $L(x)=L(y) = S$,
\item for all  $(a,b) \in (S \cup \{\del\})^2 \setminus \bigcup_{v \in S}(v,v)$ it holds that $\vdcount(\calF \to H, (a,b)) = \alpha$, \label{prop:proh1}
\item for all $v \in S$ it holds that $\vdcount(\calF \to H, (v,v)) > \alpha$.\label{prop:proh2}
\end{myenumerate}

Now we are ready to prove \cref{thm:vd-main-intro}.

\begin{proof}[Proof of \cref{thm:vd-main-intro}.]
Again, we will prove the lower bound for the decision variant, where we are given an integer $k$ and we ask whether the optimum solution is of size at most $k$.
Let $q := i(H)$ and let $S=\{v_1,v_2,\ldots,v_q\}$ be an incomparable set in $H$ of size $q$.

We reduce from \coloringVD{q}, whose hardness was shown in \cref{thm:vd-coloring-intro}.
Note that if $q=1$, then $H$ cannot contain a three-element set with private or co-private neighbors, so $H$ contains an irreflexive vertex $z$  and the result follows immediately from \cref{thm:vd-coloring-intro}.
So from now on assume that $q \geq 2$ and consider an instance $(G,k)$ of \coloringVD{q} given with a \core{\sigma}{\delta} $Q$.

\paragraph{Construction of $(G',L',k')$.}
We build an equivalent instance $(G',L',k')$ of \LHomVD($H$) as follows.
We start the construction by introducing all vertices of $G$ to $G'$.
The list $L'$ of $x \in V(G)$ is set to $L'(x) = S$.
Let $\calF=(F,L',(x',y'))$ be the $S$-prohibitor obtained by combining $(v,S)$-prohibitors given by \cref{lem:vd-gadget} for all $v \in S$.
For each edge $xy$ of $G$, we introduce a copy $\calF_{xy}$ of $\calF$ and identify $x$ with $x'$ and $y$ with $y'$.
This completes the construction of $G'$.
Finally, we set $k' := k + \alpha \cdot |E(G)|$, where $\alpha$ is the base cost of the $S$-prohibitor.

\paragraph{Equivalence of instances.}
Now let us argue that $(G',L',k')$ is a yes-instance of \LHomVD($H$) if and only if $(G,k)$ is a yes-instance of \coloringVD{q}.
First, suppose that $(G,k)$ is a yes-instance of \coloringVD{q}, i.e., there exists a set $X \subseteq V(G)$ of size at most $k$ and a proper coloring $\phi$ of $G - X$.
Each vertex $x \in V(G') \setminus X$ is mapped to $v_i$, where $i = \phi(x)$.
Now consider an edge $xy$ of $G$ and the corresponding copy $\calF_{xy}$ of the $S$-prohibitor $\bigl( F_{xy},L',(x',y') \bigr)$ introduced to $G'$.
We observe that either at least one of $x'=x,y'=y$ is in $X$, or $x',y'$ are mapped to distinct vertices from $S$.
Thus, by property (P\ref{prop:proh1}), there exists a set $X_{xy} \subseteq V(F_{xy})$ of size $\alpha$, such that the mapping of $\{x',y'\} \setminus X$ can be extended to a list homomorphism from $F_{xy} - X_{xy}$ to $H$.
Summing up, by deleting the set $X' :=  X + \sum_{xy \in E(G)} X_{xy}$, we obtain a graph that admits a list homomorphism to $H$. The total number of deleted vertices is $\abs{X'} = |X| + \sum_{xy \in E(G)} |X_{xy}| \leq k + \alpha \cdot |E(G)| = k'$.

For the other direction, suppose there exists a set $X' \subseteq V(G')$ of size at most $k'$,
such that the graph $G' - X'$ admits a homomorphism $\psi$ to $H$, respecting the lists $L'$.
Consider  $xy \in E(G)$ and let $X'_{xy} := X' \cap (V(F_{xy}) \setminus \{x',y'\})$, where $\calF_{xy}=(F_{xy},L',(x',y'))$ is the $S$-prohibitor introduced for $xy$.
We claim that without loss of generality we can assume that $|X'_{xy}| = \alpha$.
Indeed, the properties (P\ref{prop:proh1}), (P\ref{prop:proh2}) imply that $|X'_{xy}| \geq \alpha$.
Suppose now that $X'_{xy} > \alpha$. If $\{x',y'\} \cap X' \neq \emptyset$ or $\psi(x') \neq \psi(y')$,
then by property (P\ref{prop:proh1}) we can substitute $X'_{xy}$ with a smaller subset of $V(F_{xy}) \setminus \{x',y'\}$ with the same properties.
Thus we can assume that $x',y' \notin X'$ and $\psi(x') = \psi(y')$.
Define $\tilde{X}'$ to be the set obtained from $X'$ by including $x'$ and substituting the vertices of $X'_{xy}$ with the subset of $V(F_{xy}) \setminus \{x',y'\}$ given by the property (P\ref{prop:proh1}).
Note that $|\tilde{X}'| \leq |X'| \leq k'$ and by (P\ref{prop:proh1}), $\psi$ can be modified into a homomorphism from $G' - \tilde{X}'$ to $H$, respecting the list $L'$.

Now define $X := X' \cap V(G)$. We note that $|X| = |X'| - \alpha \cdot |E(G)| \leq k$.
It is straightforward to verify that $\phi : V(G) \setminus X \to [q]$ defined as $\phi(x) = i$ where $\psi(x) = v_i$ is a proper coloring of $G - X$.

\paragraph{Structure of $G'$.}
Let $h$ be the number of vertices  in the $S$-prohibitor $\mathcal{F} = \bigl(F,L',(x',y')\bigr)$; note that its number of vertices is upper-bounded by a function of $\abs{H}$.
Consider a \core{\sigma}{\delta} $Q$ of $G$; recall that the vertices of $Q$ are also vertices of $G'$.
We observe that $Q$ is a \core{\sigma'}{\delta} of $G'$ for some $\sigma'$ depending only of $\sigma,\delta$, and $H$.

Indeed, consider a component $C'$ of $G'-Q$.
Note that $C'$ is either equal to a copy $\mathcal{F}_{xy}$ of $\mathcal{F}$ without the portal vertices, for an edge $xy$ inside $Q$; or it corresponds to some component $C$ of $G - Q$, where all edges (including those from $C$ to $Q$) are replaced by copies of $\mathcal{F}$. 
In the first case, $C'$ has at most $h$ vertices.
In the second case, $C'$ has at most $\sigma$ vertices from $G$, at most $\binom{\sigma}{2}\cdot h$ vertices from the gadgets with both portals in $C'$, and at most $\delta \cdot h$ vertices from the gadgets with one portal in $Q$ and the other in $C'$.
Thus, $C'$ has at most $\sigma' := \sigma + \binom{\sigma}{2} \cdot h + \delta \cdot h$ vertices.
In both cases, a component of $G'-Q$ is adjacent to at most $\max(2,\delta)$ vertices of $Q$.

\medskip
Furthermore, note that $G'$ has $\abs{V(G)} + |E(G)| \cdot h = |V(G)|^{\bigO(1)}$ vertices.
Now the lower bound follows directly from \cref{thm:vd-coloring-intro}.
\end{proof}

\subsubsection{Constructing prohibitors}

Now we are left with proving \cref{lem:vd-gadget}.

\lemprohibitor*
\begin{proof}
Let $H$ be a graph, $S$ be an incomparable set of size at least two, and let $v \in S$.
Note that since $S$ is incomparable, every $u \in S \setminus \{v\}$ has a neighbor in $V(H) \setminus \nh(v)$.
Pick arbitrary $w \in S \setminus \{v\}$. Let $v' \in \nh(v) \setminus \nh(w)$ and $w' \in \nh(w) \setminus \nh(v)$.

We will construct a $(v,S)$-prohibitor in several steps.
The intermediate gadgets depend on the choice of $S,v,v',w,w'$.
However, as these sets and vertices are fixed throughout the proof,
we will not indicate this in the notation in order to keep it simple.

\paragraph{Step 1. Constructing a splitter.} 
The first building block is the gadget called a \emph{splitter},
which is a four-vertex path $\calA = S -V(H)  \setminus \nh(v) -\{v\} -\{v',w'\}$.
It is straightforward to verify the following properties.
\begin{myenumerate}
\item The base cost of $\calA$ is 1.
\item $\vdcount(\calA \to H, (a,b))=1$ if and only if $(a,b) \in \{ (v,v') \} \cup \left \{ (u,v'), (u,w') ~|~ u \in S \setminus \{v\} \right \}$.
\end{myenumerate}

\medskip
The next gadget is called a \emph{matcher}.
For two vertices $p,q \in V(H)$, a \emph{$(p,q)$-matcher} is a binary gadget $\calM_{p,q}=(M_{p,q},L,(y_1,y_2))$ with base cost $\alpha$,
where $M_{p,q}$ is a path with endvertices $y_1,y_2$ such that $L(y_1)=L(y_2)=\{p,q\}$, and the following properties are satisfied.
\begin{myenumerate}
\item If $(a,b) \notin \{ (p,p), (q,q) \}$, then $\vdcount(\calM_{p,q} \to H, (a,b)) = \alpha$,
\item $\vdcount(\calM_{p,q} \to H, (p,p)) > \alpha$.
\end{myenumerate}
Note that the value of $\vdcount(\calM_{p,q} \to H, (q,q))$ might be either equal to or larger than $\alpha$.

\paragraph{\boldmath  Step 2. Constructing a $(v',w')$-matcher.}
We aim to construct a $(v',w')$-matcher.
Its construction of the matcher depends on the type of ``hard substructure'' present in $H$.

\medskip
\paragraph{\boldmath  Case I: $H$ contains an irreflexive vertex.}
First, let us consider the case that $H$ has an irreflexive vertex $i$.
We distinguish three subcases.
\begin{description}
\item[Case (a): $i \in \nh(w') \cap \nh(v')$.] Then the matcher is $\{v',w'\} - \{w,i\}-\{i\}-$ $\{i\}-\{w,i\}-\{v',w'\}$ and its base cost is 1.
\item[Case (b): $i \in \nh(w') \setminus \nh(v')$.] Then the matcher is $\{v',w'\}-$ $\{i\}-\{i\}-$ $\{v',w'\}$ and its base cost is 1.
\item[Case (c): $i \notin \nh(w')$.] Then the matcher is $\{v',w'\}-\{v,w\}-\{w'\}-$ $\{i\}-\{i\}-$ $\{w'\}-\{v,w\}-\{v',w'\}$ and its base cost is 2.
\end{description}

\medskip
\paragraph{\boldmath  Case II: $H$ is a reflexive graph.}
Let $a,b$ be two non-adjacent vertices of $H$. Note that such vertices always exist by our assumption on $H$.
Later we will explain how exactly we choose $a$ and $b$.

\subparagraph{\boldmath  Step 2.1. Constructing an $(a,b)$-matcher}
Now we show that we can appropriately choose non-adjacent $a,b$ so that we can construct an $(a,b)$-matcher $\calM_{a,b}$. 
The construction depends on the type of the ``hard'' substructure contained in $H$.
\begin{description}
\item[Case (a): $H$ has three pairwise non-adjacent vertices.]
Let these vertices be $a,b,c$.
Then we define $\calM_{a,b}=\{a,b\}-\{b\}-$ $\{c\}-\{a\}- \{b\}-\{a,b\}$ and its base cost is 2.
\item[Case (b): $H$ has an induced cycle $C$ with four vertices.]
Let the consecutive vertices of $C$ be $a,x,b,y$.
Then we define $\calM_{a,b}=\{a,b\}-\{x,b\}-$ $\{x,y\}-\{b,y\}-\{a,b\}$ and its base cost is 0.
\item[Case (c): $H$ has an induced cycle $C$ with five vertices.]
Let the consecutive vertices of $C$ be $a,x,b,y,z$.
Then we define $\calM_{a,b}=\{a,b\}-\{x,y\}-\{b,z\}-\{a,b\}$ and its base cost is 0.
\end{description}
The properties of the gadgets can again be easily verified.

\subparagraph{Step 2.2. Constructing a translator.} 
The next gadget is called a \emph{$(v',w')\to(a,b)$-translator} $\calB_{(v',w')\to(a,b)}=(B_{(v',w')\to(a,b)},L,(z_1,z_2))$.
The graph $B_{(v',w')\to(a,b)}$ is a path whose endvertices are $z_1,z_2$ and have lists $L(z_1)=\{v',w'\}$ and $L(z_2)=\{a,b\}$.
Denoting the base cost of $\calB$ by $\alpha$, the intended behavior of the translator is as follows.
\begin{myenumerate}
\item If $(x,y) \notin \{ (v',b), (w',a) \}$, then $\vdcount(\calB_{(v',w')\to(a,b)} \to H, (x,y)) = \alpha$,
\item $\vdcount(\calB_{(v',w')\to(a,b)} \to H, (v',b)) > \alpha$.
\end{myenumerate}
Again,  we do not restrict the value of $\vdcount(\calB \to H, (w',a))$, we only know that it is at least $\alpha$.

We claim that we can always construct either a $(v',w')\to(a,b)$-translator or a $(v',w')\to(b,a)$-translator.
The construction is split into four subcases.
\begin{description}
\item[Case (a): $v' \in \nh(a)$ and $w \in \nh(b)$.] Then we define $\calB_{(v',w')\to(a,b)}= \{v',w'\}-\{a,w\}-\{a,b\}$ and its base cost is 0.
\item[Case (b): $v' \notin \nh(a)$ and $v' \in \nh(b)$.] Then we define $\calB_{(v',w')\to(b,a)}= \{v',w'\}-\{w\}-\{v'\}-\{a,b\}$ and its base cost is 1.
\item[Case (c): $v' \notin \nh(a)$ and $v' \notin \nh(b)$.] Then we define $\calB_{(v',w')\to(a,b)}= \{v',w'\}-\{w\}-\{v'\}-\{b\}-\{a\}-\{a,b\}$ and its base cost is 2.
\item[Case (d): $w \notin \nh(b)$.] Then we define $\calB_{(v',w')\to(b,a)}= \{v',w'\}-\{w\}-\{b\}-\{a,b\}$ and its base cost is 1.
\end{description}

\subparagraph{\boldmath  Step 2.3. Combining translators and an $(a,b)$-matcher into a $(v',w')$-matcher.}
Now we are ready to construct a $(v',w')$-matcher in Case II, see also \cref{fig:vd-matcher}.
Suppose that in Step 2.2. we obtain a $(v',w')\to(a,b)$-translator; in the other case swap the roles of $a$ and $b$ in the description below.
We create two copies $\calB_{(v',w')\to(a,b)}=(B_{(v',w')\to(a,b)},L,(z_1,z_2))$ and $\calB_{(v',w')\to(a,b)}'= (B_{(v',w')\to(a,b)}',L,(z_1',z_2'))$ of the $(v',w')\to(a,b)$-translator and the $(a,b)$-matcher $\calM_{a,b}=(M_{a,b},L,(y,y'))$.
We identify the vertices as follows: $z_2=y$ and $z_2' = y'$.
Note that we always identify vertices with equal lists.
The portals of the constructed gadget are $z_1,z_1'$ and its base cost is the sum of base costs of $\calB_{(v',w')\to(a,b)}, \calB_{(v',w')\to(a,b)}'$, and $\calM_{a,b}$.
The properties of $\calB_{(v',w')\to(a,b)},\calB_{(v',w')\to(a,b)}'$, and $\calM_{a,b}$ imply that the constructed gadget is indeed a $(v',w')$-matcher.

\paragraph{\boldmath  Step 3. Construction a $(v,S)$-prohibitor.}
In the final step we use splitters and a $(v',w')$-matcher to construct a $(v,S)$-prohibitor.
The construction is analogous as the one is Step 2.3, see also \cref{fig:vd-prohibitor}.
We create two copies $\calA=(A,L,(x_1,x_2))$ and $\calA'=(A',L,(x_1',x_2'))$ of the splitter and the $(v',w')$-matcher $\calM_{v',w'}=(M_{v',w'},L,(y,y'))$.
We identify the vertices as follows: $x_2=y$ and $x_2' = y'$. The portals of the constructed gadget are $x_1,x_1'$ and its base cost is the sum of base costs of $\calA, \calA'$, and $\calM_{v',w'}$.
The correctness of the construction follows directly from the properties of the splitter and the $(v',w')$-matcher.
\end{proof}

\begin{figure}
\begin{subfigure}[b]{0.5\textwidth}
\centering
\includegraphics[scale=1,page=4]{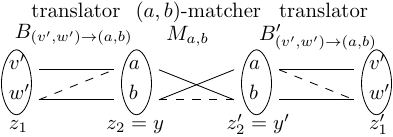}
\caption{A $(v',w')$-matcher constructed in Case II.}
\label{fig:vd-matcher}
\end{subfigure}
\begin{subfigure}[b]{0.5\textwidth}
\centering
\includegraphics[scale=1,page=5]{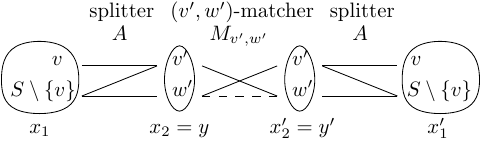}
\caption{The construction of a $(v,S)$-prohibitor.}
\label{fig:vd-prohibitor}
\end{subfigure}
\caption{High-level construction of the gadgets in \cref{lem:vd-gadget}.
Ovals denote the portals of the gadgets with their lists given inside.
Solid lines show which transitions between colors of portals are always possible,
and dashed lines show transitions that might or might not be possible.}
\end{figure}

\section{\boldmath Complexity dichotomy for \LHomED($H$)}
\label{sec:dicho-ed}
In this section we prove the complexity dichotomy for \LHomED($H$).

\thmeddicho* 

The proof is more involved than its counterpart for \LHomVD($H$), i.e., \cref{thm:vd-dicho-intro}.
Again, we will start with providing an alternative characterization of graphs considered in \cref{thm:ed-dicho-intro}.

Recall that \LHom($H$) is polynomial-time-solvable if $H$ is a \emph{bi-arc graph} and \NP-hard otherwise~\cite{DBLP:journals/jgt/FederHH03}.
Usually bi-arc graphs are defined in terms of certain geometric representation. We introduce it later, in \cref{sec:geometric}, and now show an equivalent characterization.

For a graph $H=(V,E)$, by $H^*$ we denote its \emph{associate bipartite graph}, i.e.,
the graph with vertex set $\bigcup_{v \in V} \{ v', v'' \}$ and edge set $\bigcup_{uv \in E} \{ u'v'', v',u''\}$.
We denote the bipartition of $H^*$ by $(V',V'')$, where $V' = \{v' ~|~ v \in V\}$ and $V'' = \{v'' ~|~ v \in V\}$.

Let $H$ be a bipartite graph with bipartition $(U,V)$ and let $k \geq 1$.
An \emph{special edge asteroid} is a sequence $u_0v_0,u_1v_1,\ldots,u_{2k}v_{2k}$ of edges in $H$,
where for all $i \in \{0,\ldots,2k\}$ it holds that $u_i \in U$ and $v_i \in V$,
and for each $i \in \{0,\ldots,2k\}$ (subscripts are computed modulo $2k$) there exists a path $P_{i,i+1}$ with endvertices $u_i$ and $u_{i+1}$,
such that
\begin{myenumerate}[({SEA}1.)]
\item for each $i \in \{0,\ldots,2k\}$ there is no edge between $\{u_{i+k},v_{i+k}\}$ and $\{v_i,v_{i+1}\} \cup V(P_{i,i+1})$,\label{prop:sea1}
\item there is no edge between $\{ u_0,v_0\}$ and $\{v_1,v_2,\ldots,v_{2k}\} \cup V(P_{1,2}) \cup \ldots, V(P_{2k-1,2k})$. \label{prop:sea2}
\end{myenumerate}

\begin{thm}[Feder, Hell, Huang~\cite{DBLP:journals/combinatorica/FederHH99,DBLP:journals/jgt/FederHH03}]\label{thm:biarc-characterizations}
Let $H$ be a graph. The following statements are equivalent.
\begin{myenumerate}
\item $H$ is a bi-arc graph.
\item $H^*$ is the complement of a circular-arc graph.
\item $H^*$ does not contain an induced cycle with at least six vertices or a special edge asteroid.
\end{myenumerate}
\end{thm}

Now let us proceed to the alternative characterization of graphs considered in \cref{thm:ed-dicho-intro}.

\biarc*
\begin{proof}
Let $H=(V,E)$ be a graph with no irreflexive edge nor a three-vertex set with private neighbors.

\noindent(\ref{it:biarc1}.$\to$\ref{it:biarc2}.)
For contradiction suppose that $H$ does not contain a three-vertex set with co-private neighbors and is not a bi-arc graph. By \cref{thm:biarc-characterizations} this means that $H^*$ contains either an induced cycle with at least six vertices or a special edge asteroid. Notice that any induced cycle with at least 10 vertices contains a special edge asteroid,
this leaves us with three cases to consider.

First, suppose that $H^*$ contains an induced cycle with consecutive vertices $v_0',v_1'',v_2',v_3'',v_4',v_5''$,
where $v_0,\ldots,v_5$ are some (non-necessarily distinct) vertices of $H$.
From the definition of $H^*$ it follows that the set $\{v_0,v_2,v_4\}$ has co-private neighbors $\{v_1,v_3,v_5\}$.
Indeed, each $v_i$ for $i \in \{1,3,5\}$ is adjacent to $v_{i-1}$ and $v_{i+1}$ but not to $v_{i+3}$ (subscripts computed modulo 6). This contradicts the assumption on $H$.

Now suppose that $H^*$ contains an induced cycle with consecutive vertices $v_0',v_1'',\ldots,v_7''$,
where $v_0,\ldots,v_7$ are some (non-necessarily distinct) vertices of $H$.
As $H$ does not contain an irreflexive edge and $v_3v_4 \in E(H)$, at least one of $v_3,v_4$ is a reflexive vertex.
By symmetry suppose that $v_3$ is reflexive. We observe that the set $\{v_0,v_3,v_6\}$ has private neighbors:
$v_1$ is a private neighbor of $v_0$, $v_5$ is a private neighbor of $v_6$, and $v_3$ is a private neighbor of $v_3$.
This contradicts the assumption on $H$.

Finally, suppose that $H^*$ contains a special edge asteroid $u_0'v_0'',u_1'v_1'',\ldots,u_{2k}'v_{2k}''$,
where $k \geq 1$ and each $u_i,v_i$ for $i \in \{0,\ldots,2k\}$ is a vertex of $H$.
The definition of $H^*$ asserts that for each $i \in \{0,1,k\}$ we have $u_iv_i \in E(H)$
Furthermore, by property (SEA\ref{prop:sea1}.) we know that $u_0v_k, u_1v_k, v_0u_k,v_1u_k \notin E(H)$.
Finally, by property (SEA\ref{prop:sea2}.) we know that $u_0v_1,v_0u_1 \notin E(H)$.
Thus the set $\{u_0,u_1,u_k\}$ has private neighbors $v_0,v_1,v_k$, a contradiction.

\smallskip
\noindent(\ref{it:biarc2}.$\to$\ref{it:biarc1}.) For the other direction, suppose that $H$ is bi-arc and contains a three element set $\{v_1,v_2,v_3\}$  with co-private neighbors. Let $\bar v_1$ be the co-private neighbor of $v_2,v_3$, and analogously define $\bar v_2$ and $\bar v_3$. It is straightforward to observe that the consecutive vertices $v_1',\bar v_3'', v_2', \bar v_1'', v_3', \bar v_2''$ induce a six-cycle in $H^*$. By \cref{thm:biarc-characterizations} this contradicts $H$ being a bi-arc graph.
\end{proof}

\subsection{\NP-hardness}\label{sec:EDdichoHardness}

Now we are ready to prove the hardness part of \cref{thm:ed-dicho-intro}.
By \cref{lem:biarc-intro} it is sufficient to show the following.

\begin{thm}
Let $H$ be a fixed graph, such that one of the following holds
\begin{myenumerate}
\item $H$ contains an irreflexive edge,
\item $H$ contains a three-element set with private neighbors,
\item $H$ is not a bi-arc graph.
\end{myenumerate}
Then \LHomED($H$) is \NP-hard.
\end{thm}
\begin{proof}
Consider the cases.
\begin{enumerate}
\item If $H$ contains an irreflexive edge $xy$, then solving \LHomED($H$) on an instance where all lists are $\{x,y\}$ is equivalent to solving \textsc{Max Cut}.
\item Suppose that $H$ contains a set $\{x_1,x_2,x_3\}$ with private neighbors. For each $i \in [3]$, the private neighbor of $x_i$ is denoted by $x_i'$. 
We reduce from \textsc{Edge Multiway Cut} with three terminals.
Consider an instance $G$ with terminals $t_1,t_2,t_3$.
Notice that by subdividing each edge of $G$ once we obtain a bipartite graph $G'$ which is an equivalent instance of 3-\textsc{Edge Multiway Cut} (still with terminals $t_1,t_2,t_3$). Let the bipartition of $G'$ be $(U,V)$, where $t_1,t_2,t_3 \in U$. For each vertex $w$ of $G'$ we define $L(w) \in V(H)$ as follows
\[
L(w) =
\begin{cases}
\{x_i\} & \text{ if } w = t_i \text{ for } i \in [3],\\
\{x_1,x_2,x_3\} & \text{ if } w \in U \setminus \{t_1,t_2,t_3\},\\
\{x_1',x_2',x_3'\} & \text{ if } w \in V.
\end{cases}
\]
It is straightforward to verify that solving $(G',L)$ as an instance of \LHomED($H$) is equivalent to solving $G'$ (and thus $G$) as an instance of \textsc{Edge Multiway Cut} with three terminals.
\item If $H$ it not a bi-arc graph, then already solving \LHom($H$) (or equivalently, solving \LHomED($H$) with deletion budget 0) is \NP-hard~\cite{DBLP:journals/jgt/FederHH03}.\qedhere
\end{enumerate}
\end{proof}

\subsection{Polynomial-time algorithm}
Before we proceed to the proof of the algorithmic statement in \cref{thm:ed-dicho-intro},
let us carefully analyze the structure of the graphs $H$ that are considered here.
We use the formulation based on \cref{lem:biarc-intro}, i.e., throughout this section we assume that $H$ is a bi-arc graph
that does not contain an irreflexive edge nor a three-element set with private neighbors.
For brevity, we do not repeat this in assumption in the lemmas.

\subsubsection{\boldmath Analyzing a geometric representation of $H$}\label{sec:geometric}

As we mentioned before, bi-arc graphs are typically defined in terms of a certain geometric representation.
Let $C$ be a circle, the top-most point being $p$, the bottom-most point being $q$.
A \emph{bi-arc} $(N, S)$ is an ordered pair of arcs on $C$ such that the \emph{north arc} $N$ contains $p$ but not $q$, and the \emph{south arc} $S$ contains $q$ but not $p$.
A graph $H$ is a \emph{bi-arc graph} if there exists a bijection mapping every vertex $x \in V(H)$ to a bi-arc $\{(N_x,S_x) \mid x\in V(H)\}$
such that, for any $x,y\in V(H)$, we have
	\begin{myitemize}
		\item If $xy\in E(H)$ then neither $N_x$ intersects $S_y$ nor $N_y$ intersects $S_x$.
		\item If $xy\notin E(H)$ then both $N_x$ intersects $S_y$ and $N_y$ intersects $S_x$.
	\end{myitemize}
Let $\calN=\{N_x \mid x\in V(H)\}$ and $\calS=\{S_x \mid x\in V(H)\}$. We refer to the tuple $(\calN, \calS)$ as a \emph{geometric representation} of $H$.	
For a subset $X$ of the vertices of $H$, we set $\calN_X\coloneqq \{N_x \mid x\in X\}$ and $\calS_X\coloneqq \{S_x \mid x\in X\}$.

In what follows we assume that $H=(V,E)$ is a bi-arc graph with geometric representation $(\calN, \calS)$ defined on the circle $C$ with top-most point $p$ and the bottom-most point $q$. Let us introduce some more notation, see also \cref{fig:geometric}.

Let $C_L$ and $C_R$ be the left and the right half of the circle $C$ (both including $p$ and $q$), respectively.
We consider the points in $C_L$ and $C_R$ ordered clockwise starting from $q$ and $p$, respectively.
Thus, e.g., if $x,y \in C_L$ and $x$ is closer to $q$ than $y$ (along $C_L$), we write $x < y$.
For $N\in \calN$, the \emph{lower bound} $\ell(N)$ of $N$ is the smallest of the points in $N\cap C_L$, i.e., it is the endpoint of $N$ that is on the left half of $C$. Correspondingly, its \emph{upper bound} $u(N)$ is the endpoint of $N$ in $C_R$.
Similarly, for $S\in \calS$, $\ell(S)$ is the endpoint of $S$ in $C_R$, and $u(S)$ is its endpoint in $C_L$.
The orders on the two halves of $C$ allow us to compare the same bounds for arcs from the same set of the geometric representation, say $\ell(S_1)$ and $\ell(S_2)$ for $S_1, S_2\in \calS$. They also allow us to compare opposed bounds for arcs from different sets. For example, for $S\in \calS, N\in \calN$, both $\ell(N)$ and $u(S)$ are in $C_L$, and if, say, $\ell(N)<u(S)$ then $N$ and $S$ intersect on the left half of $C$.

\begin{figure}
\centering
\includegraphics[scale=1]{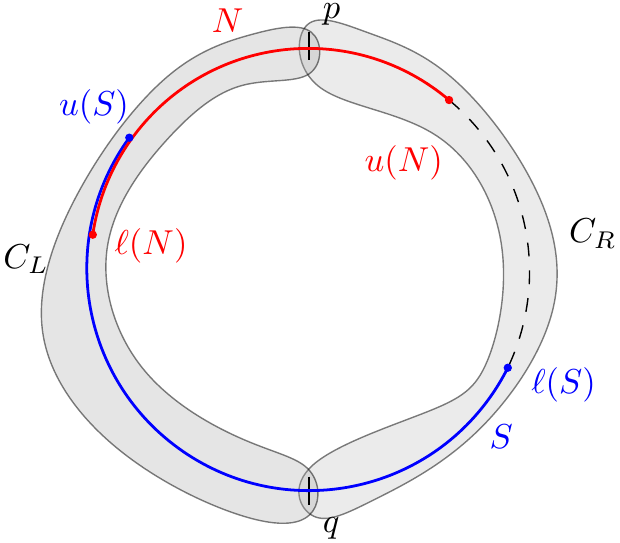}
\caption{Geometric representation of $H$. Here $\ell(N) < u(S)$ and $u(N) < \ell(S)$.}
\label{fig:geometric}
\end{figure}

\paragraph{Orderings of arcs and vertices.}
First, we observe that from $(\calN,\calS)$ we can deduce certain natural ordering of incomparable sets in $H$.

Let $x,y\in V(H)$. We write $N_x  \precarc N_y$ if $\ell(N_x)< \ell(N_y)$ and $u(N_x)< u(N_y)$. Analogously, $S_x \precarc S_y$ if $\ell(S_x)<\ell(S_y)$ and $u(S_x)< u(S_y)$.
From this we define a binary relation $\prec$ on $V(H)$  as follows:
\begin{myitemize}
	\item If both $x$ and $y$ are irreflexive then $x \prec y$ if and only if $S_x \precarc S_y$.
	\item Otherwise, $x \prec y$ if and only if $N_x \precarc N_y$.
\end{myitemize}

Observe that it is possible that two arcs $N_x,N_y$ (resp. $S_x,S_y$) are incomparable with respect to $\precarc$.
However, this means that the vertices $x$ and $y$ are comparable.

\begin{lem}\label{lem:linorderedArcs}
If $X$ is an incomparable set in $H$, then $\precarc$ induces a strict total order on $\calN_X$ and on $\calS_X$.
\end{lem}
\begin{proof}
Let us consider $(\calN_X,\precarc)$, as the argument for $\calS_X$ is analogous.
It is straightforward to see that $\precarc$ is irreflexive and transitive.
It remains to show that it every pair of elements of $\calN_X$ is comparable in $\precarc$.
For contradiction, suppose this is not the case that there are $x,y\in X$ such that neither $N_x \precarc N_y$ nor $N_y \precarc N_x$.
As $p \in N_x \cap N_y$, we observe that one of the arcs $N_x,N_y$ must be contained in the other.
By symmetry suppose $N_y\subseteq N_x$.
Then every arc in $\calS_X$ that is disjoint with $N_x$ is also disjoint from $N_y$.
Consequently, every neighbor of $x$ is also a neighbor of $y$.
This means that $y$ dominates $x$, a contradiction to the fact that $X$ is incomparable.
\end{proof}

We  say that a set $X \subseteq V$ is \emph{reflexive} (resp. \emph{irreflexive}) if every vertex is $X$ is reflexive (resp. irreflexive).
If $X$ is neither reflexive nor irreflexive, it is \emph{mixed}.
Let us point out that if $\{x,y\}$ is incomparable and reflexive, or incomparable and irreflexive, then $N_x \precarc N_y$ does not necessarily imply $S_x \precarc S_y$ (or vice versa). 
However, if $\{x,y\}$ is incomparable and mixed then the orders coincide.

\begin{lem}\label{lem:mixedpo}
Let $\{x,y\}$ be a mixed incomparable set in $H$.
Then $N_x \precarc N_y$ if and only if $S_x \precarc S_y$.
\end{lem}
\begin{proof}
Suppose $S_x \precarc S_y$ (the other direction is analogous).
Consider the case that $x$ is reflexive and $y$ is irreflexive.
Since $\{x,y\}$ is incomparable there is a vertex $y' \in \nh(y)\setminus \nh(x)$. Thus, $N_{y'}$ intersects $S_x$ but not $S_y$,
and $S_{y'}$ intersects $N_x$ but not $N_y$.
Then $N_{y'}$ intersects $S_x$ in $C_R$ since $S_x  \precarc S_y$.
Since $x$ is reflexive and hence $N_x$ and $S_x$ do not intersect it follows that $u(N_{x}) < u(N_{y'})$.
Since $H$ contains no irreflexive edge it follows that $y'$ is also reflexive and consequently $N_{y'}$ and $S_{y'}$ do not intersect,
which means $\ell(S_{y'})>u(N_{y'})>u(N_x)$.
Thus, $S_{y'}$ and $N_x$ intersect in $C_L$, i.e.~$u(S_{y'})>\ell(N_x)$.
Since $N_y$ does not intersect $S_{y'}$ it follows that $\ell(N_y)>u(S_{y'})> \ell(N_x)$,
which implies $N_x \precarc N_y$ by \cref{lem:linorderedArcs}, as required.  
	
The case where $x$ is irreflexive and $y$ is reflexive can be treated analogously using a neighbor $x' \in \nh(x)\setminus \nh(y)$.
\end{proof}

The next lemma is an analog of \cref{lem:linorderedArcs} for $\prec$.

\begin{lem}\label{lem:linorderedVertices}
If $X$ is an incomparable set in $H$, then $\prec$ is a strict total order on $X$.
\end{lem}
\begin{proof}
It is straightforward to see that $\prec$ is irreflexive.
By \cref{lem:linorderedArcs} we observe that every pair of elements is comparable by $\prec$.
It remains to show transitivity. Let $x,y,z\in X$ with $x \prec y$ and $y \prec z$. If all of $x,y,z$ are reflexive or if all of them are irreflexive then $x \prec z$ follows directly. 
Otherwise, $x \prec z$ follows in a straightforward way from \cref{lem:mixedpo}.
\end{proof}

As a consequence of \cref{lem:linorderedVertices} we observe that there are few possible ways how a vertex may interact with an incomparable set.

\paragraph{Characterizing types of interactions.}
Let $X$ be an incomparable set in $H$. As $H$ has no irreflexive edges, by \cref{lem:linorderedVertices} we know that $\prec$ is a strict total order on $X$. Let the elements of $X$ be $x_1 \prec x_2 \prec \dots \prec x_{|X|}$.
A \emph{prefix} of $X$ is any subset of $X$ of the form $\{x_1,x_2,\ldots,x_i\}$ for $i \in [|X|-1]$.
Similarly, a \emph{suffix} of $X$ is any subset of the form $\{x_i,x_{i+1},\ldots,x_t\}$ for $i \in \{2,\ldots,|X|\}$.
Note that a prefix or suffix is always a non-empty proper subset of $X$.
A \emph{prefix} (resp. \emph{suffix}) of $\calN_X$ or $\calS_X$ is the set of appropriate arcs of a prefix (resp. suffix) of $X$.

\begin{defn}[Interaction types]\label{def:types}
Let $y$ be any vertex of $H$. We define some \emph{interaction types} of $y$ with $X$ as follows.
\begin{description}
\item[p-interaction] $y$ has a \emph{p-interaction} (or \emph{p-interacts}) with $X$ if $\nh(y)\cap X$ is a prefix of $X$.
In terms of the geometric representation this means that the arcs $N_y$ and $S_y$ intersect a suffix of $\calS_X$ and $\calN_X$, respectively.
\item[s-interaction] $y$ has an \emph{s-interaction} (or \emph{s-interacts}) with $X$ if $\nh(y)\cap X$ is a suffix of $X$.
This means that the arcs $N_y$ and $S_y$ intersect a prefix of $\calS_X$ and $\calN_X$, respectively.
\item[0-interaction] $y$ has a \emph{0-interaction} (or \emph{0-interacts}) with $X$ if $\nh(y)\cap X=\emptyset$.
This means that the arcs $N_y$ and $S_y$ intersect all arcs in $\calS_X$ and $\calN_X$, respectively.
\item[1-interaction] $y$ has a \emph{1-interaction} (or \emph{1-interacts}) with $X$ if $X \subseteq \nh(y)$.
This means that the arcs $N_y$ and $S_y$ intersect no arcs in $\calS_X$ and $\calN_X$, respectively.
\end{description}
\end{defn}

The following observation is a simple consequence of the properties of the arcs in $(\calN,\calS)$.

\begin{lem}\label{lem:simpleNonAdjacency}
Let $X$ be an incomparable set in $H$ and let $x_1, x_2, x_3\in X$ with $x_1 \prec x_2 \prec x_3$.
Then $\nh(x_1) \cap \nh(x_3) \subseteq \nh(x_2)$.
\end{lem}
\begin{proof}
Let $y \in \nh(x_1) \cap \nh(x_3)$.
First suppose that at least two vertices in $\{x_1, x_2, x_3\}$ are reflexive.
Then we have $N_{x_1} \precarc N_{x_2} \precarc N_{x_3}$ and consequently $N_{x_2}\subseteq (N_{x_1}\cup N_{x_3})$.
Since $S_y$ intersects neither $N_{x_1}$ nor $N_{x_3}$, it cannot intersect $N_{x_2}$.
Hence $y$ and $x_2$ are adjacent.
	
Now suppose that at most one vertex in $\{x_1, x_2, x_3\}$ is reflexive.
This case can be treated analogously using the fact that now \cref{lem:mixedpo} ensures $S_{x_1} \precarc S_{x_2} \precarc S_{x_3}$.
\end{proof}

Now we show that \cref{def:types} describes all possible interactions of a vertex and an incomarable set.

\begin{lem}\label{lem:hastobetype}
Let $X$ be an incomparable set and let $y$ be a vertex in $H$. Then $y$ has of one of the four interaction types with $X$ defined in \cref{def:types}.
\end{lem}
\begin{proof}
For contradiction, suppose the opposite.
Then there are $x_1, x_2, x_3\in X$ with $x_1 \prec  x_2 \prec x_3$ such that either i) $yx_1, yx_3 \in E(H)$ and $yx_2 \notin E(H)$, or ii) $yx_2 \in E(H)$ and $yx_1, yx_3 \notin E(H)$.
By \cref{lem:simpleNonAdjacency} we know that the first case is impossible, to let us assume that the second one occurs.  
Note that $y$ is a private neighbor of $x_2$ with respect to $\{x_1, x_2, x_3\}$.
Since $X$ is incomparable, there are vertices $z\in \nh(x_1)\setminus \nh(x_2)$ and $z'\in \nh(x_3)\setminus \nh(x_2)$.
By \cref{lem:simpleNonAdjacency}, $z \notin \nh(x_3)$ and $z' \notin \nh(x_1)$.
Thus, $z$ and $z'$ are, respectively, private neighbors of $x_1$ and $x_3$ with respect to $\{x_1, x_2, x_3\}$.
This contradicts the assumption on $H$.
\end{proof}

In the next lemmas we analyze pairwise interactions of incomparable sets.

\begin{lem}\label{lem:sameTypesameInteraction}
Let $X$ and $\{y_1, y_2\}$ be incomparable sets in $H$ such that both $y_1$ and $y_2$ have a p-interaction (resp. s-interaction) with $X$.
Then $y_1$ and $y_2$ have the same neighborhood in $X$.
\end{lem}
\begin{proof}
Suppose both $y_1$ and $y_2$ have a p-interaction with $X$. (The case where they both have an s-interaction is analogous.) For contradiction, suppose that $\nh(y_2) \cap X \setminus \nh(y_1) \cap X \neq \emptyset$.
Since $y_1$ and $y_2$ both have a p-interaction with $X$,
there are $x_1,x_2,x_3\in X$ with $x_1 \prec x_2 \prec x_3$ such that
$\nh(y_1)\cap \{x_1,x_2,x_3\}=\{x_1\}$ and $\nh(y_2)\cap \{x_1,x_2,x_3\}=\{x_1,x_2\}$.
	
Suppose that at least two of $x_1, x_2, x_3$ are reflexive and consequently $N_{x_1} \precarc N_{x_2} \precarc N_{x_3}$. Then $S_{y_1}$ intersects $N_{x_2}, N_{x_3}$ in $C_R$, and $S_{y_2}$ intersects $N_{x_3}$ also in $C_R$.
In particular, $S_{y_1} \precarc S_{y_2}$.
	
Since $\{y_1,y_2\}$ is incomparable there is some $y_1' \in \nh(y_1)\setminus \nh(y_2)$.
Note that $y_1' \notin \{x_1,x_2,x_3\}$.
Since $S_{y_1} \precarc S_{y_2}$ it follows that $N_{y_1'}$ intersects $S_{y_2}$ in $C_L$, i.e., $\ell(N_{y_1'}) < u(S_{y_2})$.

Since $X$ is incomparable there is some $x_3' \in \nh(x_3)\setminus \nh(x_2)$.
Note that $x_3' \notin \{y_1,y_2\}$.
Since $N_{x_2} \precarc N_{x_3}$ it follows that $S_{x_3'}$ intersects $N_{x_2}$ in $C_L$, i.e., $\ell(N_{x_2})< u(S_{x_3'})$. 
	
Summing up, we have $\ell(N_{y_1'})<u(S_{y_2})<\ell(N_{x_2})< u(S_{x_3'})$, where the second inequality follows from the fact that $S_{y_2}$ and $N_{x_2}$ are disjoint $y_2x_2 \in E(H)$.
As $\ell(N_{y_1'}) M i(S_{x_3'})$, we observe that $N_{y_1'}$ intersects $S_{x_3'}$ and this $y_1'$ and $x_3'$ are non-adjacent.
Now observe that the set $\{y_1',x_2,x_3\}$ has private neighbors $y_1,y_2$, and $x_3'$, a contradiction with the assumption on $H$.
	
The case where at most one of $x_1,x_2,x_3$ is reflexive can be treated analogously using the fact that $S_{x_1}\precarc S_{x_2} \precarc S_{x_3}$ (recall \cref{lem:mixedpo}) and considering arcs from $\calN$ instead of ones from $\calS$ and vice versa.
\end{proof}

\begin{lem}\label{lem:orderofTypes}
Let $X$ and $Y$ are incomparable sets in $H$.
Let $y_1, y_2, y_3 \in Y$.
\begin{myenumerate}[(i)]
\item If $y_1$ p-interacts with $X$ and $y_2$ 1-interacts with $X$, then $y_1 \prec y_2$. \label{item:pb41}
\item If $y_1$ p-interacts with $X$ and $y_2$ 0-interacts with $X$, then $y_1 \prec y_2$. \label{item:pb40}
\item If $y_1$ 1-interacts with $X$ and $y_2$ s-interacts with $X$, then $y_1 \prec y_2$. \label{item:1b4s}
\item If $y_1$ 0-interacts with $X$ and $y_2$ s-interacts with $X$, then $y_1 \prec y_2$. \label{item:0b4s}
\item If $y_1$ p-interacts with $X$ and $y_2$ s-interacts with $X$, then $y_1 \prec y_2$. \label{item:pb4s}
\item If $y_1$ p-interacts and $y_2$ 1-interacts with $X$, then $y_3$ cannot 0-interact with $X$. \label{item:ifpthenno10}
\item If $y_1$ s-interacts and $y_2$ 1-interacts with $X$, then $y_3$ cannot 0-interact with $X$. \label{item:ifsthenno01}
	\end{myenumerate}
\end{lem}
\begin{proof}
Let $x_1$ and $x_2$ be, respectively, the minimum and the maximum elements of $X$ (according to $\prec$).

\paragraph{Proof of \eqref{item:pb41}.}
Note that according to their interaction types, $y_1$ is adjacent to $x_1$ but not $x_2$,
and $y_2$ is adjacent to both $x_1$ and $x_2$.

\smallskip
Suppose that $\{x_1, x_2\}$ is reflexive or mixed. Then $N_{x_1} \precarc N_{x_2}$ and consequently $S_{y_1}$ intersects $N_{x_2}$ in $C_R$, i.e., $\ell(S_{y_1})< u(N_{x_2})$
As $S(y_2)$ is disjoint with $N(x_2)$, we have $\ell(S(y_1)) < \ell(S(y_2))$ and thus $S_{y_1} \precarc S_{y_2}$.
If $\{y_1, y_2\}$ is irreflexive or mixed, this already implies $y_1 \prec y_2$ (recall \cref{lem:mixedpo}).

So suppose that $y_1$ and $y_2$ are reflexive.
If $N_{y_1}$ intersects $S_{x_2}$ in $C_L$, then, as $x_2y_2 \in E(H)$,
we have $N_{y_1} \precarc N_{y_2}$ and therefore $y_1 \prec y_2$ as required.
If otherwise $N_{y_1}$ intersects $S_{x_2}$ in $C_R$,
then $\ell(S_{x_2})<u(N_{y_1})<\ell(S_{y_1}) <u(N_{x_2})$,
where the second inequality uses that $y_1$ is reflexive and thus $N_{y_1}$ and $S_{y_1}$ are disjoint.
In particular, we established that $S_{x_2}$ and $N_{x_2}$ intersect, which means that $x_2$ is irreflexive.
Consequently, $\{x_1,x_2\}$ is mixed and so $S_{x_1} \precarc S_{x_2}$ according to \cref{lem:mixedpo}.
However, this gives us $\ell(S_{x_1}) < \ell(S_{x_2})<u(N_{y_1})$, i.e.,that $S_{x_1}$ intersects $N_{y_1}$,
which is a contradiction to the fact that $y_1$ and $x_1$ are adjacent.

\smallskip			%
Now suppose that $\{x_1, x_2\}$ is irreflexive.
Note that in this case, since $H$ does not contain an irreflexive edge, both $y_1$ and $y_2$ are reflexive.

Similar to the previous case, $S_{x_1} \precarc S_{x_2}$ and consequently $N_{y_1}$ intersects $S_{x_2}$ in $C_L$, i.e.,$\ell(N_{y_1})< u(S_{x_2})$.
As $N_{y_2}$ is disjoint with $S_{x_2}$, we obtain $\ell(N_{y_1})< u(S_{x_2}) < \ell(N_{y_2})$.
Thus, $N_{y_1} \precarc N_{y_2}$, which implies $y_1 \prec y_2$ since $\{y_1, y_2\}$ is reflexive.

\paragraph{Proof of \eqref{item:pb40}.}
Note that according to their interaction types, $y_1$ is adjacent to $x_1$ but not $x_2$, and $y_2$ is adjacent to none of $x_1$ and $x_2$.

\smallskip
Suppose that $\{x_1, x_2\}$ is reflexive or mixed. Then $N_{x_1} \precarc N_{x_2}$ and consequently $S_{y_1}$ intersects $N_{x_2}$ in $C_R$, i.e., $\ell(S_{y_1})< u(N_{x_2})$. 

We claim that $S_{y_1} \precarc S_{y_2}$.
For contradiction suppose that $S_{y_2} \precarc S_{y_1}$.
We will show that then $H$ contains three-element set with private neighbors.
Since each of $X$ and $Y$ is incomparable there are $y_2' \in \nh(y_2)\setminus \nh(y_1)$ and $x_2' \in \nh(x_2)\setminus \nh(x_1)$.
It follows that $S_{x_2'}$ intersects $N_{x_1}$ in $C_L$, i.e.,~$\ell(N_{x_1})<u(S_{x_2'})$.
Also, $N_{y_2'}$ intersects $S_{y_1}$ in $C_L$ (because $S_{y_2} \precarc S_{y_1}$), i.e.,~$\ell(N_{y_2'})<u(S_{y_1})$.
Since $x_1$ and $y_1$ are adjacent we have $u(S_{y_1})< \ell(N_{x_1})$ and consequently
$\ell(N_{y_2'})<u(S_{y_1})<\ell(N_{x_1})<u(S_{x_2'})$.
This means that which means that $y_2'$ and $x_2'$ are not adjacent.
Therefore, the set $\{y_1,x_2',y_2\}$ has private neighbors $x_1,x_2,y_2'$, a contradiction.

Thus we established that $S_{y_1} \precarc S_{y_2}$.
If $\{y_1, y_2\}$ is irreflexive or mixed, or if $N_{y_1} \precarc N_{y_2}$, then $y_1 \prec y_2$, as required.
So let us assume for contradiction that $y_1$ and $y_2$ are reflexive and $N_{y_2} \precarc N_{y_1}$.
Since $S_{y_1} \precarc S_{y_2}$ it follows that $S_{y_2}$ intersects $N_{x_1}$ in $C_L$.
Similarly, $S_{y_1}$ intersects $N_{x_2}$ in $C_R$.
Using the fact the $y_1$ and $y_2$ are reflexive, we conclude that both $N_{y_1}$ and $N_{y_2}$ are strictly contained in $N_{x_1}\cup N_{x_2}$.
Since $N_{y_2} \precarc N_{y_1}$ it follows that $S_{x_2}$ intersects $N_{y_1}$ in $C_R$ and therefore $x_2$ is irreflexive (as $N_{y_1}$ is strictly contained in $N_{x_1}\cup N_{x_2}$).
Thus, $X$ is mixed and so $S_{x_1} \precarc S_{x_2}$ by \cref{lem:mixedpo}.
However, this implies that $N_{y_1}$ intersects $S_{x_1}$, a contradiction.
		
\smallskip
Now suppose that $\{x_1, x_2\}$ is irreflexive. 
Then $S_{x_1} \precarc S_{x_2}$ and consequently $N_{y_1}$ intersects $S_{x_2}$ in $C_L$.
If $N_{y_2} \precarc N_{y_1}$ we obtain a three-element set with private neighbors as in the previous case.
So we can assume $N_{y_1} \precarc N_{y_2}$.
Since $H$ does not contain an irreflexive edge, we observe that $y_1$ is reflexive.
Therefore $\{y_1,y_2\}$ is reflexive of mixed and thus $y_1 \prec y_2$.

\paragraph{Proof of \eqref{item:1b4s} and \eqref{item:0b4s}.} The proofs are analogous to the proofs of \eqref{item:pb41} and \eqref{item:pb40}, respectively.

\paragraph{Proof of \eqref{item:pb4s}.}  According to their interaction types, $y_1$ is adjacent to $x_1$ but not to $x_2$,
and $y_2$ is adjacent to $x_2$ but not to $x_1$. Note that the roles of $\{x_1,x_2\}$ and $\{y_1,y_2\}$ are symmetric.
Since $H$ does not contain an irreflexive edge,
it is not the case that both $\{x_1,x_2\}$ and $\{y_1,y_2\}$ are irreflexive.

We claim that $N_{x_1} \precarc N_{x_2}$ if and only if $S_{y_1} \precarc S_{y_2}$,
and $S_{x_1} \precarc S_{x_2}$ if and only if $N_{y_1} \precarc N_{y_2}$.
Let us prove the first statement, as the other one is symmetric (by swapping the roles of $x_1,x_2$ and $y_1,y_2$).

Suppose that $N_{x_1} \precarc N_{x_2}$ and $S_{y_2} \precarc S_{y_1}$.
As $S_{y_1}$ is disjoint with $N_{x_1}$ and $S_{y_2}$ is disjoint with $N_{x_2}$,
we notice that $S_{y_1}$ cannot intersect $N_{x_2}$, a contradiction.
Now suppose that $S_{y_1} \precarc S_{y_2}$ and $N_{x_2} \precarc N_{x_1}$.
The proof is symmetric to the previous case by swapping $x_1,x_2$ with $y_1,y_2$
and arcs from $\calS$ with arc from $\calN$.

The properties above imply that if $\{x_1,x_2\}$ is mixed or reflexive (i.e., $N_{x_1} \precarc N_{x_2}$),
and $\{y_1,y_2\}$ is mixed of irreflexive (i.e., $\prec$ orders $y_1,y_2$ with respect to arcs in $\calS$),
we obtain $y_1 \prec y_2$.
Similarly, if $\{x_1,x_2\}$ is mixed or irreflexive and $\{y_1,y_2\}$ is mixed or reflexive,
then $y_1 \prec y_2$.

We are left with the case that both $\{x_1,x_2\}$ and $\{y_1,y_2\}$ are reflexive and the following properties hold:
\begin{myitemize}
\item $N_{x_1} \precarc N_{x_2}$ (because $x_1 \prec x_2$),
\item $N_{y_2} \precarc N_{y_1}$ (otherwise $y_1 \prec y_2$ and we are done),
\item $S_{y_1} \precarc N_{y_1}$ and $S_{x_2} \precarc S_{x_1}$  (by the implications shown above).
\end{myitemize}
For contradiction, suppose that such a case occurs.
As $N_{y_2}$ is disjoint with $S_{x_2}$, it intersects $S_{x_1}$ in $C_L$.
On the other hand, $N_{y_2}$ does not intersect $S_{y_2}$, thus $u(S_{y_2}) < \ell(N_{y_2}) < u(S_{x_1})$.
But now $N_{x_1}$ is supposed to intersect $S_{y_2}$ in $C_L$ which is impossible without intersecting $S_{x_1}$, a contradiction.

\paragraph{Proof of \eqref{item:ifpthenno10}.}
Suppose the converse holds, i.e. there exist $y_1,y_2,y_3 \in Y$ such that they p-interact, 0-interact and 1-interact with $X$, respectively.
In particular $\nh(x_1) \cap \{y_1,y_2,y_3\} = \{y_1,y_2\}$ and $\nh(x_2) \cap \{y_1,y_2,y_3\} = \{y_2\}$.

Statements \eqref{item:pb41} and \eqref{item:pb40} imply that $y_1 \prec y_2$ and $y_1 \prec y_3$.

If $y_1 \prec y_2 \prec y_3$, then the interaction of $x_2$ with $Y$ is not of any type specified in \cref{def:types}.
If $y_1 \prec y_3 \prec y_2$, then the interaction of $x_1$ with $Y$ is not of any type specified in \cref{def:types}.
In both cases this contradicts \ref{lem:hastobetype}.

\paragraph{Proof of \eqref{item:ifsthenno01}.} The proof is analogous to the proof of \eqref{item:ifpthenno10}.
\end{proof}

For an incomparable set $Y$ enumerated as $y_1 \prec y_2 \prec \ldots \prec y_{|Y|}$,
a \emph{segment} is a contiguous subsequence of $y_1,\ldots,y_{|Y|}$,
i.e., $\{y_i,y_{i+1},\ldots,y_j\}$ for some $i \leq j$.
A segment $Y'$ \emph{precedes} a segment $Y''$ if if $y' \prec y''$ for every $y' \in Y'$ and $y'' \in Y''$.
A segment $Y'$ is a \emph{p-segment} with respect to an incomparable set $X$ if every $y' \in Y'$ has a p-interaction with $X$.
Analogously we define s-segments, 1-segments, and 0-segments.

As a consequence of \cref{lem:orderofTypes} we obtain the following corollary.
\begin{cor}\label{cor:orderofTypes}
Suppose that $X$ and $Y$ are incomparable sets in $H$.
Then $Y$ can be partitioned into (possibly empty) segments with the following interactions with $X$:
\begin{myitemize}
\item a p-segment preceding a 1-segment preceding an s-segment, or
\item a p-segment preceding a 0-segment preceding an s-segment, or
\item a 0-segment preceding a 1-segment, or
\item a 1-segment preceding a 0-segment.
\end{myitemize}
Furthermore, the vertices of each segment have the same neighborhood in $X$.
\end{cor}
\begin{proof}
The possible ways of partitioning $Y$ follow immediately from \cref{lem:orderofTypes}.
The equality of neighborhoods is trivial in case of 0- and 1-segments, and for p- and s-segments it follows from \cref{lem:sameTypesameInteraction}.
\end{proof}

\paragraph{Interaction matrices.}
A useful way to think of interactions of incomparable set is via their \emph{interaction matrices}.
Let $X$ and $Y$ be two incomparable sets in $H$. Recall that by \cref{lem:linorderedVertices} each of these sets is totally ordered by $\prec$.
We enumerate the elements of $X$ as $x_1 \prec x_2 \ldots x_{|X|}$ and the elements of $Y$ as $y_1 \prec y_2 \prec \ldots \prec y_{|Y|}$.
The interaction matrix of $(X,Y)$ is the boolean matrix with $|X|$ rows and $|Y|$ columns such that for $i \in [|X|]$ and $j \in [|Y|]$ we have
\[
M_{i,j} = \begin{cases}
1 & \text { if } x_iy_j \in E(H),\\
0 & \text { if } x_iy_j \notin E(H).
\end{cases}
\]
Note that the transpose of $M$ is the interaction matrix of $(Y,X)$.
\cref{cor:orderofTypes} can be restated in the following way in terms of interaction matrices.

\begin{cor}\label{cor:verynicematrix}
Let $X$ and $Y$ be incomparable sets and let $M$ be the interaction matrix of $(X,Y)$.
Then there are indices $(i_1,j_1), (i_2,j_2) \in \{0,\ldots,|X|+1\} \times \{0,\ldots,|Y|+1\}$
such that $M_{i,j} = 1$ if and only if either $i \leq i_1$ and $j \leq j_1$
or $i \geq i_2$ and $j \geq j_2$.
Furthemore, either $i_1 < i_2$ and $j_1 < j_2$,
or $i_1 \geq i_2$ and $j_1 \geq j_2$.
\end{cor}
\begin{proof}
The proof follows in a straightforward way from \cref{cor:orderofTypes} and the fact that $M$ transposed is also an interaction matrix
and thus has to satisfy the same conditions.
\end{proof}

\begin{figure}
\centering
\includegraphics[scale=1,page=1]{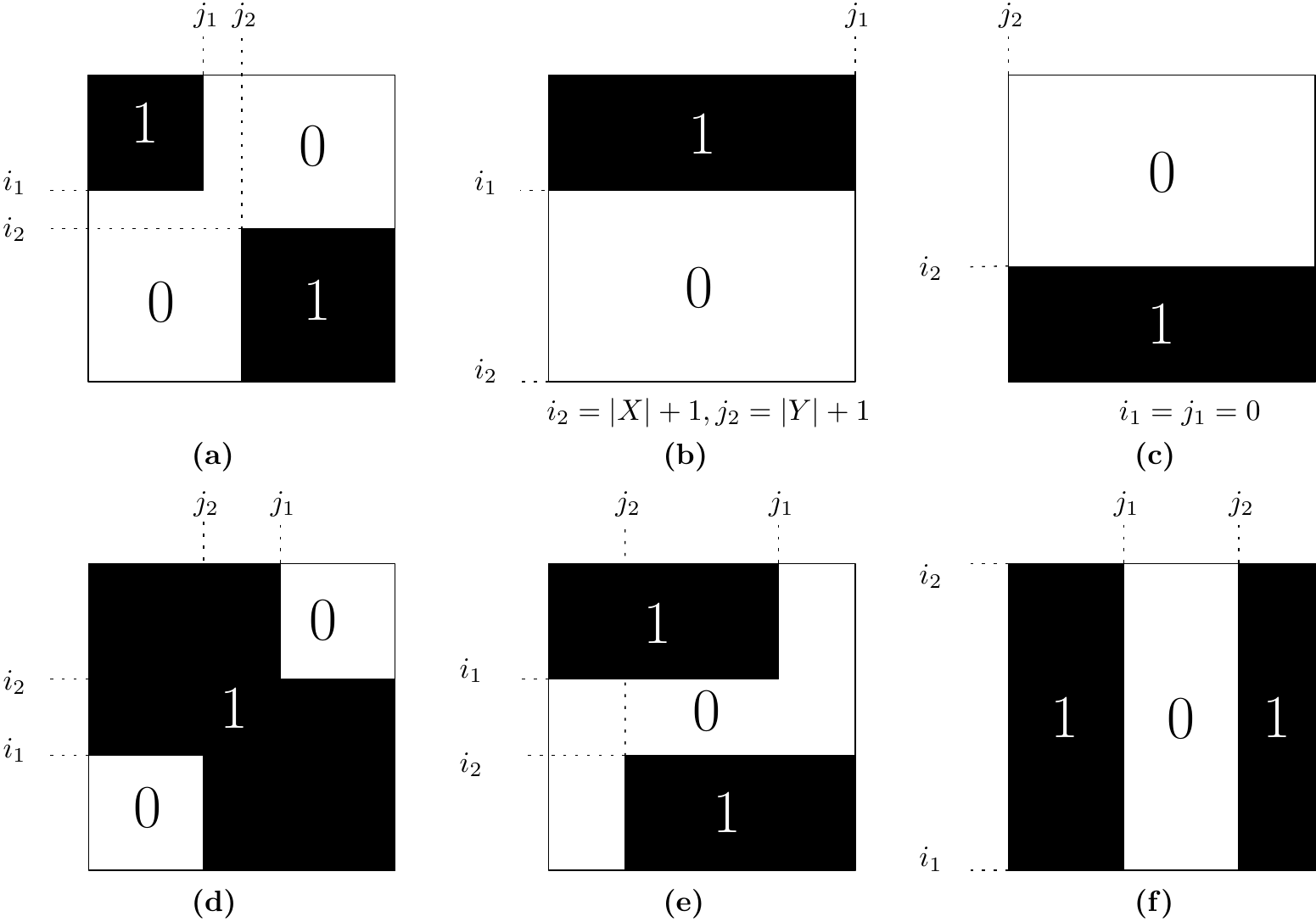}
\caption{Interaction matrices (a) -- (d) are allowed by \cref{cor:verynicematrix}.
Matrices (e) and (f) do not satisfy the last condition in \cref{cor:verynicematrix}.}
\label{fig:matrices2}
\end{figure}

Note that \cref{cor:verynicematrix} implies that the non-zero entries of $M$ form two (possibly overlapping or empty) rectangles,
one containing the top-left corner and the other containing the bottom-right corner, see \cref{fig:matrices2}.
However, in our algorithm we are more interested in covering zero entries of $M$.
It is straightforward to observe that they can also be covered by two rectangles, one containing the bottom-left corner,
and the other containing the top-right corner of $M$.
However, the next corollary we consider a slightly different partition of $M$ into rectangles, this time pairwise disjoint.

\begin{cor}\label{cor:nicematrix}
Let $X$ and $Y$ be incomparable sets and let $M$ be the interaction matrix of $(X,Y)$.
Then the zero entries can be partitioned into three (possibly empty) pairwise disjoint rectangles $R_1, R_2, R_3$,
such that:
\begin{myitemize}
\item $R_1$ contains the top-right corner of $M$ or is empty,
\item $R_2$ consists of a number of consecutive rows of $M$ or is empty,
\item $R_3$ contains the bottom-left corner of $M$ or is empty.
\end{myitemize}
\end{cor}
\begin{proof}
	For $(i,j) \in \{0,\ldots,|X|+1\} \times \{0,\ldots,|Y|+1\}$, by $M_{\leq i, \geq j}$ we mean the submatrix consisting of all entries $M_{i',j'}$ for $(i', j') \in [|X|] \times [|Y|]$ such that $i' \leq i$ and $j' \geq j$. Note that if $i=0$ or $j=|Y|+1$, then $M_{\leq i, \geq j}$ is empty.
Similarly we define  $M_{\geq i, \leq j}$.

Let $(i_1,j_1)$ and $(i_2,j_2)$ be as in \cref{cor:verynicematrix}. We have two cases, see also \cref{fig:matrices3}

If $i_1 < i_2$ and $j_1 < j_2$, then we define $R_1 = M_{\leq i_1, \geq j_1+1}$ and $R_3 = M_{\geq i_2, \leq j_2-1}$.
The rectangle $R_2$ consists of all entries $M_{i,j}$ for $i_1 < i < i_2$ and $1 \leq j \leq |Y|$.

If $i_1 \geq i_2$ and $j_1 \geq j_2$, then we define $R_1 = M_{\leq i_2-1, \geq j_1+1}$ and $R_3 = M_{\geq i_1+1, \leq j_2-1}$.
The rectangle $R_2$ is empty.
\end{proof}

\begin{figure}
\centering
\includegraphics[scale=1,page=2]{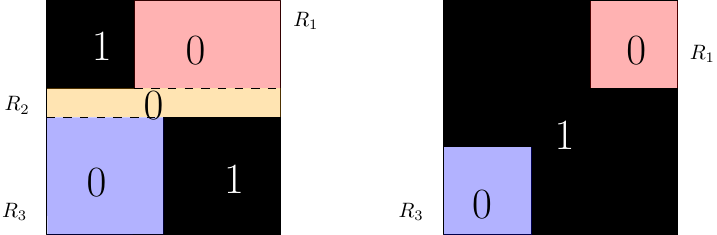}
\caption{\cref{cor:nicematrix}: partitioning zero entries of an interaction matrix into at most three pairwise disjoint rectangles.}
\label{fig:matrices3}
\end{figure}

Note that in \cref{cor:nicematrix} we lost some properties of $M$ compared to \cref{cor:verynicematrix}.
In particular, the matrix from \cref{fig:matrices2}~(d) is allowed by \cref{cor:nicematrix} but not by \cref{cor:verynicematrix}.
Nevertheless, the weaker properties will be sufficient for our algorithm.

\subsubsection{\boldmath  Solving \LHomED($H$)}

Equipped with \cref{cor:nicematrix}, we can finally complete the proof of \cref{thm:ed-dicho-intro}.
Recall that by  \cref{lem:biarc-intro} it is sufficient to prove the following theorem.

\begin{thm}\label{thm:easiness}
Let $H$ be a fixed bi-arc graph which contains no irreflexive edge and no set of 3 vertices with private neighbors.
Then \LHomED($H$) can be solved in polynomial time.
\end{thm}
\begin{proof}
We will solve the problem by a reduction to the min-cut problem.
In this problem we are given a digraph $D$ with source $s$ and sink $t$, multiple arcs, and arc weights, we aim to find a minimum $s$-$t$-cut:
a set $A$ of arc with minimum total weight such that $D \setminus A$ contains no $s$-$t$-paths.
This problem can be solved in polynomial time using network flow algorithms.

Note that the problem can be equivalently stated as follows.
We want to partition the vertex set of $D$ into two sets, $S$ and $T$,
such that
\begin{myitemize}
\item $s \in S$ and $t \in T$,
\item the set $A$ of arcs beginning in $S$ and ending in $T$ is as small as possible.
\end{myitemize}
We will use this interpretation when describing the intended behavior of $D$.
We will refer to the vertices in $S$ as \emph{mapped to the left},
and the vertices in $T$ as \emph{mapped to the right}.

In our construction of $D$ we will have two types of arcs: with unit weight and arcs whose weight is equal to a large number $\mathfrak{w}$,
to be specified later. We will choose $\mathfrak{w}$ in a way that in any optimum solution no edge with weight $\mathfrak{w}$ is deleted.
Thus these edges will be called \emph{unbreakable}.

Let $(G,L)$ be an instance of \LHomED($H$), where $H$ is as in the statement of the theorem.
Fix some geometric representation of $H$ as a bi-arc graph and let $\prec$ be the relation on vertices defined in the previous section.
We start the construction of $D$ with introducing the source $s$ and the sink $t$.

Now consider a vertex $v \in V(G)$ with list $X := L(v)$.
Note that we can assume that $L(v) \neq \emptyset$, as otherwise we immediately report a no-instance.
Recall that $X$ is an incomparable set and by \cref{lem:linorderedVertices} the relation $\prec$ is a strict total order on $X$.
Let us enumerate $X$ as $x_1 \prec x_2 \prec \ldots \prec x_{|X|}$.

We introduce to $D$ a directed path $P^v$ with $|X|+1$ vertices denoted by $p^v_0,p^v_1,\ldots,p^v_{|X|}$.
The edges of $P^v$ are unbreakable. Furthermore we add unbreakable arcs $\overrightarrow{p^v_0t}$ and $\overrightarrow{sp^v_{|X|}}$.
Note that this means that $p^v_0$ will always be mapped to the right and $p^v_{|X|}$ will always be mapped to the left.
Furthermore, as there can be no edges from vertices mapped to the left to vertices mapped to the right,
and all edges of $P^v$ are unbreakable, we conclude that in any optimum solution there is $i \in [|X|]$
such that $p^v_0,\ldots,p^v_{i-1}$ are mapped to the right and $p^v_{i},\ldots,p^v_{|X|}$ are mapped to the left.
The choice of $i$ will be interpreted as mapping $v$ to $x_i$. Let us denote this value of $i$ by $\trans(v)$.

Now let us take care of edges of $G$. Fix some arbitrary ordering of vertices of $G$.
Let $vw$ be an edge such that $v$ precedes $w$ in the ordering.
Let $X := L(v)$ and $Y := L(w)$, where the elements of $X$ are $x_1 \prec \ldots \prec x_{|X|}$ and the elements of $Y$ are $y_1 \prec \ldots \prec y_{|Y|}$.
Let $M$ be the interaction matrix of $(X,Y)$.
Let $R_1,R_2,R_3$ be the (possibly empty) rectangles given by \cref{cor:nicematrix}.

Note that if $v$ is mapped to $x_i$ and $w$ is mapped to $y_j$, then the edge $vw$ has to be deleted if and only if $M_{i,j}=0$,
i.e, it $M_{i,j}$ is contained in exactly one of $R_1,R_2,R_3$.

Suppose that $R_1$ is non-empty and its bottom-left corner is $M_{i',j'}$. In particular $j' \geq 1$.
This means that if $v$ is mapped to $x_i$ for some $i \leq i'$ and $w$ is mapped to $y_j$ for some $i \geq j'$,
the edge $vw$ needs to be deleted.
We add the arc $\overrightarrow{p^v_{i'}p^w_{j'-1}}$.
Note that is $v$ is mapped to $x_i$ for $i \leq i'$, then $\trans(v) \leq i'$ and thus $p^v_{i'}$ is mapped to the left.
If $w$ is mapped to $y_j$ for $j \geq j'$, then $\trans(v) \geq j'$ and thus $p^w_{j'-1}$ is mapped to the right.
Thus then we need to remove the arc $\overrightarrow{p^v_{i'}p^w_{j'-1}}$.

If $R_3$ is non-empty, we proceed in an analogous way.
Let the top-right corner of $R_3$ be $M_{i',j'}$: we add the arc $\overrightarrow{p^w_{j'}p^v_{i'-1}}$.

Finally, suppose that $R_2$ is non-empty and let is consist of rows $i',i'+1,\ldots,i''$.
This means that if $v$ is mapped to $x_i$ for $i' \leq i \leq i''$, then we have to remove the edge $vw$, regardless of the image of $w$.
We add the arc $\overrightarrow{p^v_{i''}p^v_{i'-1}}$. We point out that such an arc might have already been introduced; in this case we add another copy.
Note that if $v$ is mapped to  is mapped to $x_i$ for $i' \leq i \leq i''$, i.e., $i' \leq \trans(v) \leq i''$,
then $p^v_{i'}$ is mapped to the left and $p^v_{i''}$ is mapped to the left. Thus in this case the arc $\overrightarrow{p^v_{i''}p^v_{i'-1}}$ needs to be deleted.

This completes the construction of $D$. We set $\mathfrak{w}$ to be the number of unit arcs plus one.
Note that removing all unit arcs destroys all $s$-$t$-paths in $D$, i.e., we never need to remove any unbreakable arc.

From the discussion above it follows that a minimum $s$-$t$-cut in $D$ has weight $k$ if and only if we can modify $(G,L)$ into a yes-instance of \LHom($H$) by removing $k$ edges.

Indeed, consider an edge $vw$ of $G$.
If $v$ and $w$ are mapped to an edge of $V(H)$, then we do not need to remove any arc introduced for $vw$.
Now consider the case that $v$ and $w$ are mapped to a non-edge $x_iy_j$ of $V(H)$, where $x_i \in L(v)$ and $y_j \in L(w)$.
Then the edge $vw$ needs to be removed from $G$. On the other hand, we always need to remove from $D$ at least one unit arc introduced for the edge $vw$.
Furthermore, removing one arc is always sufficient.
This follows from the fact that rectangles $R_1,R_2,R_3$ are pairwise disjoint and we only need to remove the arc corresponding to the rectangle containing $M_{i,j}$ (where $M$ is the interaction matrix of $(L(v), L(w))$.
This completes the proof.
\end{proof}

\section{\boldmath Tight results for \LHomED($H$)}
\label{sec:ed-tight}
In this section we prove \cref{thm:ed-main-intro-tw} and \cref{thm:ed-main-intro}.

Let us define the central parameter of this section.
Recall that by $i(H)$ we denote the size of a largest incomparable set in $H$.
Also, recall the definition of the decomposition.

\defdecomp*

If a graph $H$ admits no decomposition, we say that it is \emph{undecomposable}.

\begin{defn}[Obstruction]\label{def:obstruction}
	An \emph{obstruction} in $H$ is a set $O \subseteq V(H)$, such that at least one of the following holds
	\begin{myenumerate}
		\item $|O|=2$ and $O$ induces in $H$ an irreflexive edge,
		\item $|O|=3$ and $O$ has private neighbors,
		\item $|O|=3$ and $O$ has co-private neighbors.
	\end{myenumerate}
\end{defn}

Recall that by \cref{thm:ed-dicho-intro} obstructions are precisely the minimal substructures in $H$ that make the \LHomED($H$) problem \NP-hard (assuming \PP $\neq$ \NP).
	
For a graph $H$ that contains an obstruction, we define 
\begin{align*}
\param(H) :=  \max \{ i(H') ~|~  & H' \text{ is an undecomposable induced subgraph} \\ & \text{of $H$ that contains an obstruction}  \}
\end{align*}
For sake of completeness, if $H$ does not contain an obstruction, we define $\param(H):=1$.

Clearly $\param(H) \leq i(H)$ for every graph $H$.
However, we observe that $i(H)$ cannot be upper-bounded by any function of $\param(H)$.

Let $H_1 = K_2$ and call its vertex set $A$. Clearly, $H_1$ contains an obstruction and $i(H_1)=2$.
Let $k \geq 3$ be an integer. Let $H_2$ be a graph constructed as follows.
We introduce a reflexive clique $B$ consisting of vertices $v_0,v_1,\ldots,v_{k+1},w_0,w_1,\ldots,w_{k+1}$.
Furthermore, we add a set $C$ of irreflexive vertices $u_1,u_2,\ldots,u_k$.
For each $i \in [k]$, the vertex $u_i$ is adjacent to $v_j$ for all $j \geq i$ and all $w_j$ for all $j \leq i$.
We point out that the vertices $v_0,\ldots,v_{k+1}$ are pairwise comparable, and so are the vertices $w_0,\ldots,w_{k+1}$.
Furthermore, for each $i,j$ it holds that $\nh(u_j) \subseteq \nh(v_i)$ and $\nh(u_j) \subseteq \nh(w_i)$.
Finally, vertices $u_1,\ldots,u_k$ form an incomparable set. Thus $i(H_1)=k$.
We observe that $H_1$ does not contain any obstruction. Indeed, irreflexive vertices form an independent set.
Suppose now that $H_1$ contains a three-element set $O$ with (co-)private neighbors. In particular, $O$ is incomparable and of size 3,
thus it must consist of irreflexive vertices, say $u_i,u_j,u_\ell$ for $i < j < \ell$.
But note that $\nh(u_i) \cap \nh(u_\ell) \subseteq \nh(u_j) \subseteq \nh(u_i) \cup \nh(u_\ell)$.
Consequently, $u_j$ cannot have a private neighbor, and $u_i,u_\ell$ cannot have a co-private neighbor, a contradiction.

Now let $H$ be a graph obtained by adding all edges between $A$ and $B$.
Note that $(A,B,C)$ is a decomposition of $H$. Consequently, we have $i(H) = i(H_2)= k$ and $\param(H) = i(H_1)=2$.
As $k$ can be arbitrarily large, we conclude that $i(H)$ cannot be bounded in terms of $\param(H)$.

\subsection{\boldmath  Algorithm for \LHomED($H$)}
First we show the algorithmic part of \cref{thm:ed-main-intro-tw}.

\begin{customthm}{\ref{thm:ed-main-intro-tw}~(a)}\label{thm:ed-algo}
	Let $H$ be a fixed graph and let $(G,L)$ be an $n$-vertex instance of $\ED(H)$, given along with a tree decomposition of width $t$.
	Then $(G,L)$ can be solved in time $\param(H)^t \cdot n^{\bigO(1)}$.
\end{customthm}

Before we proceed to the proof, let us show how a decomposition of $H$ can be used algorithmically.

\begin{lem}\label{lem:ed-decomp}
	Let $H$ be a graph that admits a decomposition $(A,B,C)$.
	Let $H_1 := H[A]$ and $H_2:= H[B \cup C]$.
	Then every instance $(G,L)$ of $\ED(H)$ can be solved by solving an instance $(G_1,L)$ of $\ED(H_1)$ and  an instance $(G_2,L)$ of $\ED(H_2)$,
	where $G_1,G_2$ are vertex-disjoint induced subgraphs of $G$.
	All additional computation is performed in time polynomial in $|V(G)|$.
\end{lem}
\begin{proof}
	Notice that for every $a \in A$, $b \in B$, and $c \in C$,
	we have $\nh(c) \subseteq \nh(a) \subseteq \nh(b)$.
	Since for every $x \in V(G)$ we can assume that $L(x)$ is an incomparable set,
	we conclude that each list is fully contained in $A,B$, or $C$.
	
	For $X \in \{A,B,C\}$, let $V_X$ be the set of vertices $x \in V(G)$ such that $L(x) \subseteq X$. By the previous observation, sets $V_A,V_B,V_C$ form a partition of $V(G)$.
	
	Consider $x_a \in V_A$ and $x_b \in V_B$ such that $x_ax_b \in E(G)$.
	Note that for any $a \in L(x_a)$ and any $b \in L(x_b)$ we have $ab \in E(H)$.
	Thus the edge $x_ax_b$ can be safely deleted without changing the solution to the instance.
	
	Now consider $x_a \in V_A$ and $x_c \in V_C$ such that $x_ax_c \in E(G)$.
	For any $a \in L(x_a)$ and any $c \in L(x_c)$ we have $ac \notin E(H)$.
	Thus in any solution the edge $x_ax_c$ has to be deleted.
	Consequently, we can focus on solving the instance $(G \setminus \{x_ax_c\},L)$,
	and afterwards add $x_ax_c$ to the set of deleted edges.
	
	Repeating the above two steps exhaustively, we obtain an equivalent instance where there are no edges between $V_A$ and $V_B$,
	and no edges between $V_A$ and $V_C$.
	Thus, the instances given by graphs $G_1 := G[V_A]$ and $G_2 := G[V_B \cup V_C]$ can be solved independently.
	Finally, recall that for each $i \in [2]$, the lists of all vertices of $G_i$ are contained in $V(H_i)$, so $(G_i,L)$ is an instance of (the minimization version of) $\ED(H_i)$.
\end{proof}

We show that it is straight-forward that $\ED(H)$ can be solved in time $i(H)^t \cdot n^{\bigO(1)}$.

\begin{lem}\label{lem:ed-algo-simple}
	Let $H$ be a fixed graph and let $(G,L)$ be an $n$-vertex instance of $\ED(H)$, given along with a tree decomposition of width $t$.
	Then $(G,L)$ can be solved in time $i(H)^t \cdot n^{\bigO(1)}$.
\end{lem}
\begin{proof} 
	Suppose that for $v,v'\in V(H)$ we have $\nh(v)\subseteq \nh(v')$. Then in every solution $v'$ can be used instead of $v$. Therefore, we can assume that every list in $L$ is an incomparable set of $H$, and consequently that each list is of size at most $i(H')$.
	Now we can easily solve the problem in time $i(H')^t \cdot n^{\bigO(1)}$ by a straightforward dynamic programming on a tree decomposition:
	For each mapping of the current bag $X$ we remember the minimum number of edges 
	that need to be deleted to obtain a list homomorphism of the subgraph of $G$ induced by the bags in the subtree rooted in $X$,
	whose projection on $X$ is as prescribed.
\end{proof}

Now we are ready to prove \cref{thm:ed-algo}.

\begin{proof}[Proof of \cref{thm:ed-algo}.]
	We will prove the following statement by induction on $|V(H')|$:
	\begin{itemize}
		\item[$(\star)$] For every induced subgraph of $H'$ of $H$,
		every $n$-vertex instance $(G,L)$ of $\ED(H')$ 
		given along with a tree decomposition with width $t$ can be solved in time $\param(H)^t \cdot n^{\bigO(1)}$.
	\end{itemize}
	First, suppose that $H'$ does not have an obstruction. Then the problem can be solved in polynomial time by \cref{thm:ed-dicho-intro}.
	
	Second, suppose that $H'$ is undecomposable. Then $i(H') = \param(H)$ and the statement follows from \cref{lem:ed-algo-simple}.
	
	Note that the case that $|V(H')|=1$ is covered by the cases discussed above.
	So suppose that $V(H')$ contains an obstruction and admits a decomposition $(A,B,C)$, and that $(\star)$ holds for all subgraphs with fewer than $|V(H')|$ vertices.
	Let $H'_1:=H'[A]$ and $H'_2:=H'[B \cup C]$. As $A \neq \emptyset$ and $B \cup C \neq \emptyset$, the inductive assumption applies to $H_1'$ and $H'_2$. Note that given a tree decomposition of $G$ with width $t$ we can easily obtain a tree decomposition of every induced subgraph of $G$ with width at most $t$.
	Thus, by \cref{lem:ed-decomp}, $(\star)$ holds also for $H'$.
	Now the theorem follows from $(\star)$ for $H'=H$.
\end{proof}

\subsection{\boldmath  Hardness for \LHomED($H$)}
In this section we prove the \cref{thm:ed-main-intro}.

\thmedmainintro*

While the high-level idea of the proof is very similar to the one of \cref{thm:ed-main-intro}, the construction of the gadgets is much more involved in this case. The main difficulty stems from the fact that we have an additional algorithmic trick, i.e., the decomposition of $H$, and we need to take this into consideration while building gadgets.

\subsubsection{Realizing relations}
Recall the definition of gadgets from \cref{sec:prelims}. Note that a gadget may use lists.
We start by defining what it means for a gadget to realize a relation.

\begin{defn}[Realizing a relation]\label{def:realizing}
	Let $H$ be a graph. For some positive integer $r$, let $R\subseteq V(H)^r$. Consider some $r$-ary $H$-gadget $\calJ=(J,L,\boldx)$ , and some $\boldd\in V(H)^r$.
	By $\edcount(\calJ \to H, \boldd)$ we denote the size of a minimum set of edges $X$ that ensure that there is a list homomorphism $h$ from $(J\setminus X,L)$ to $H$ with $\phi(\boldx)=\boldd$.
	Then $\calJ$ \emph{realizes} the relation $R$ if there is an integer $k$ such that, for each $\boldd\in V(H)^r$, $\edcount(\calJ \to H, \boldd)=k$ if $\boldd\in R$, and $\edcount(\calJ \to H, \boldd)>k$, otherwise.
	We say that $\calJ$ \emph{$\omega$-realizes} $R$ for some integer $\omega\ge 1$ if additionally, for each $\boldd\notin R$ we have $\edcount(\calJ \to H, \boldd)=k+\omega$.
\end{defn}

Our definitions of $H$-gadgets and realizing a relation generalize the definitions of a $q$-gadget and the corresponding definition of realizing a relation from \cite{EsmerFMR24hub}. Consider the  graph $K_2$ with vertices $a$ and $b$. Then  the problem \LHomED($K_2$) is a reformulation of list $2$-coloring with edge deletions, where mapping a vertex to $a$ or $b$ is interpreted as coloring it using the two respective colors.
Hence, we can translate the following result from \cite{EsmerFMR24hub} to our notation using $H$-gadgets.

\begin{lem}[{\cite[Corollary 6.12]{EsmerFMR24hub}}]\label{lem:2colrealizing}
	Consider the graph $K_2$ with vertices $a$ and $b$. 
	For each $r\ge 1$, and $R\subseteq \{a,b\}^r$, there is an $r$-ary $K_2$-gadget that $1$-realizes $R$.
\end{lem}

We observe that in the world of colorings, an edge essentially takes the role of a ``Not Equals'' gadget. So we can deduce the following generalization of \cref{lem:2colrealizing}.
\begin{cor}\label{lem:NEQtoRelstrong}
	Let $H$ be some graph, and let $a$, $b$ be two vertices in $H$. Suppose there is an $H$-gadget that $1$-realizes $\NEQ(a,b)\coloneqq\{(a,b),(b,a)\}$.
	Then, for each $r\ge 1$, and $R\subseteq \{a,b\}^r$, there is an $r$-ary $H$-gadget that $1$-realizes $R$.
\end{cor}
\begin{proof}
	By \cref{lem:2colrealizing}, there is an $K_2$-gadget $\calJ$ that $1$realizes $R$, assuming the vertices of $K_2$ are labeled $a$ and $b$.
	A single edge in $\calJ$ is nothing more or less than a $K_2$-gadget $1$-realizing $\NEQ(a,b)$. Therefore, if we modify $\calJ$ by replacing each of its edges by the $H$-gadget that $1$-realizes $\NEQ(a,b)$, the resulting $H$-gadget $1$-realizes $R$.
\end{proof}

\subsubsection{Main gadgets}\label{sec:maingadgets}
The proof of \cref{thm:ed-main-intro} is based on a series of intermediate results, 
in which we show that certain relations can be realized over $H$ (see \cref{def:realizing}).

Given some binary relation $R\subseteq X \times Y$ and an element $x\in X$, we define $R(x)\coloneqq\{y\in Y \mid (x,y)\in R\}$.
An important role is played by \emph{indicators}. Let $H$ be a graph.
For a set $S \subseteq V(H)$ and two vertices $a,b \in V(H)$, an \emph{indicator of $S$ over $\{a,b\}$} is any relation $I\subseteq S\times \{a,b\}^{|S|(|S|-1)}$ such that
\begin{myitemize}
	\item $I(x)$ is non-empty for each $x\in S$, and
	\item $I(x)$ and $I(x')$ are disjoint for distinct $x,x'\in S$.
\end{myitemize}
Intuitively, $I$ ``translates'' elements of $S$ to binary sequences over $\{a,b\}$, so that distinct elements of $S$ are mapped to distinct (but not necessarily unique) codewords.	
	
The following lemma shows that the existence of an obstruction allows us to construct indicators of incomparable sets. 
\begin{restatable}{lem}{lemindicator} 
\label{lem:indicator}
Let $H$ be an undecomposable graph that contains an obstruction $O$.
Let $a,b\in O$ be distinct, and let $S$ be an incomparable set in $H$ of size at least 2.
Then there is an indicator $I$ of $S$ over $\{a,b\}$ and a gadget that realizes $I$. 
\end{restatable}

Let us postpone the proof of \cref{lem:indicator} to \cref{sec:indicator}. It uses \cref{lem:movingMain} as key ingredient. Proving \cref{lem:movingMain} is the most technically involved part of our lower bound, and this is done in \cref{sec:2vertexmoves}.

The next lemma shows that the structure of obstructions is rich enough to express all binary relations.

\begin{restatable}{lem}{lemallrelations}
\label{lem:ed-allrelations}
Let $H$ be a graph and let $O$ be an obstruction in $H$. Let $a,b \in O$ be distinct.
Then for each $\omega\ge 1$, $r\ge 1$, and $R\subseteq \{a,b\}^r$, there is a gadget that $\omega$-realizes $R$.
\end{restatable}
\begin{proof}
By \cref{lem:NEQtoRelstrong}, it suffices to show that there is a gadget that $1$-realizes $\NEQ(a,b)=\{(a,b),(b,a)\}$.
	We consider three cases, as in the definition of an obstruction.
	As the obstructions are symmetric, the choice of $a,b \in O$ does not matter.
	\begin{description}          
	\item[Case (1): $O$ induces an irreflexive edge.]
	Let $O = \{a,b\}$.
	A single edge between the portals $1$-realizes $\NEQ(a,b)$. %
	
	\item[Case (2): $|O|=3$ and $O$ has private neighbors.]
	Let $O = \{a,b,c\}$. Let $a',b',c'$ denote private neighbors of $a,b,c$, respectively.
	A gadget $1$-realizing $\NEQ(a,b)$ is depicted in \cref{fig:neq2}.
	
	\item[Case (3): $|O|=3$ and $O$ has co-private neighbors.]
	Let $O = \{a,b,c\}$. Let $\bar a$ be a co-private neighbor for $b,c$,
	let $\bar b$ be a co-private neighbor for $a,c$,
	and let $\bar c$ be a co-private neighbor for $a,b$.
	A gadget $1$-realizing $\NEQ(a,b)$ is depicted in \cref{fig:neq3}.
      \end{description}
      This completes the proof.\end{proof}
	
	\begin{figure}
	\centering
		\begin{subfigure}[b]{0.32\textwidth}
			\centering
			\includegraphics[scale=1,page=4]{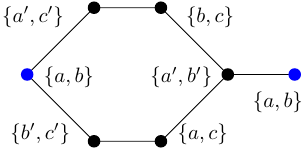}
			\caption{Case (2).}
			\label{fig:neq2}
		\end{subfigure}
		\begin{subfigure}[b]{0.42\textwidth}
			\centering
			\includegraphics[scale=1,page=5]{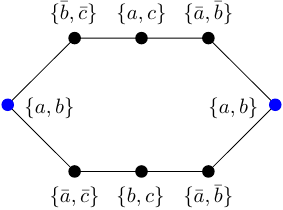}
			\caption{Case (3).}
			\label{fig:neq3}
		\end{subfigure}
		\caption{The construction of $\calJ_{\NEQ}$ in the proof of \cref{lem:ed-allrelations}.
			Portals of the gadgets are marked blue and the sets denote lists.
		}
	\end{figure}

Combining \cref{lem:indicator} and \cref{lem:ed-allrelations} we obtain the final gadget used in the hardness proof.

\begin{lem}\label{lem:ed-neq}
Let $H$ be an undecomposable graph that contains an obstruction $O$. Let $a,b\in O$ be distinct, and let $S$ be an incomparable set in $H$ of size at least 2. For some integers $r,p\ge 1$, let $R\subseteq S^r\times \{a,b\}^p$.
Then there is a gadget that realizes $R$.
\end{lem}
\begin{proof}
Let $I \subseteq S \times \{a,b\}^{|S|(|S|-1)}$ be an indicator of $S$ over $\{a,b\}$ as in \cref{lem:indicator}.
For each element $x\in S$ let $\id(x)$ be some fixed element from $I(x)$. This way, by definition of an indicator, each $x\in  S$ has a unique \emph{identifier} $\id(x)$ over $\{a,b\}$.
We ``translate'' $R$ to the following relation in $\{a,b\}^{r\abs{S}(\abs{S}+1)+p}$
\[
	R_I\coloneqq \{(\id(x_1), \ldots, \id(x_r), y_1,\ldots y_p) \mid (x_1,\ldots, x_r, y_1, \ldots, y_p)\in R\}.
\]

We obtain a gadget $\calJ$ that realizes $R$ by combining $r$ copies of the indicator gadget for $I$ (exists by \cref{lem:indicator}) with a gadget $\calJ_I$ realizing $R_I$ (exists by \cref{lem:ed-allrelations}), and identifying the respective portals, i.e., for each $i$ the portals of $\calJ_I$ that correspond to $\id(x_i)$ are identified with the respective portals of one copy of the indicator gadget.
\end{proof}

We refine the previous result by balancing out the cost of violating the relation. In particular, we show the existence of a $1$-realization, i.e., where the violation cost exceeds the deletion cost for states that are in $R$ by $1$, see \cref{def:realizing}.

\begin{lem}\label{lem:ed-allrelomega}
	Let $H$ be an undecomposable graph that contains an obstruction, and let $S$ be an incomparable set in $H$ of size at least 2. For some integers $r,p\ge 1$, let $R\subseteq S^r$.
	Then there is a gadget that $1$-realizes $R$.
\end{lem}
\begin{proof}
	Let $\bar R=S^r\setminus R$ and let $\bar \boldd_1, \ldots, \bar \boldd_{\abs{\bar R}}$ be an enumeration of $\bar R$.
	Let $O$ be an obstruction in $H$ and let $a,b\in O$ be distinct.
	We define a relation $R'$ in $S^r\times \{a,b\}^{\abs{\bar R}}$ as follows:
		A tuple $(x_1, \ldots, x_r, y_1, \ldots, y_{\abs{\bar R}})$ is in $R'$ if and only if one of the following holds:
			\begin{enumerate}
				\item $(x_1, \ldots, x_r)\in R$ and $y_i=a$ for all $i\in [\abs{\bar R}]$, or
				\item for some $i\in [\abs{\bar R}]$, $(x_1, \ldots, x_r)=\bar \boldd_i$, $y_i=b$, and $y_j=a$ for all $j\neq i$.
			\end{enumerate}
	According to \cref{lem:ed-neq}, there is a gadget $\calJ$ that realizes $R'$.
	Let $Z_1, \ldots, Z_r$ be the portals of $\calJ$ that belong to the entries $x_1\ldots, x_r$, and let $z_1, \ldots, z_{\abs{\bar R}}$ be those that belong to $y_1, \ldots, y_{\abs{\bar R}}$.
	By definition of a realization, there is some integer $\alpha$ such that for each $\boldd\in R'$, we have $\edcount(\calJ \to H,\boldd)=\alpha$.
	We modify $\calJ$ by attaching to each portal $z_i$ a vertex $v_i$ with a list $L(v_i)$. Recall that any two distinct vertices of an obstruction are incomparable. Consequently, $a$ has a neighbor $A$ in $H$ that is not a neighbor of $b$. For each $v_i$ we set $L(v_i)=\{A\}$.
	Now interpret this modified gadget $\calJ'$ as a gadget with portals $Z_1, \ldots, Z_r$. We claim that $\calJ'$ $1$-realizes $R$:
	Whenever $Z_1, \ldots, Z_r$ are mapped to a state in $R$ this requires $\alpha$ edge deletions in $\calJ$, and these also suffice as all of the $z_i$s can be mapped to $a$, which is adjacent to the image $A$ of $v_i$s.
	However, if $Z_1, \ldots, Z_r$ are in a state $\bar \boldd_i\in \bar R$ then $\alpha$ edge deletions are required in $\calJ$, but one more edge deletion is required between $z_i$ that is mapped to $b$ and $v_i$, which is mapped to $A$. At the same time $\alpha +1$ edge deletions are sufficient to extend such a state $\bar \boldd_i$, so all violations of $R$ have the same cost $\alpha +1$.
\end{proof}

\subsubsection{\boldmath Proof of \texorpdfstring{\cref{thm:ed-main-intro}}{Theorem \ref{thm:ed-hardness-final}}: Tight lower bound under the SETH}
Let us start with proving \cref{thm:ed-main-intro} in the special case that $H$ is undecomposable.

\begin{thm}\label{thm:ed-hardness}
Let $H$ be an undecomposable graph that contains an obstruction.
For every $\epsilon >0$, there are $\sigma, \delta>0$ such that \LHomED($H$) on $n$-vertex instances given with a \core{\sigma}{\delta} of size $p$ cannot be solved in time $(i(H)-\epsilon)^p \cdot n^{\bigO(1)}$, unless the SETH fails.
\end{thm}
\begin{proof}
Note that we have $i(H) \geq 2$, as every obstruction contains two incomparable vertices.
First consider a special case that $i(H) = 2$.
We note that then the graph must contain an irreflexive edge $ab$, as all other types of obstructions contain three pairwise incomparable vertices. But then, if all lists are set to $\{a,b\}$, the problem is equivalent to $\coloringED{2}$ (or, equivalently, \textsc{Max Cut}). Then the statement follows from \cref{thm:ed-coloring-intro} for $q=2$ (even in the stronger non-list variant).

So from now on assume that $q := i(H) \geq 3$. With \cref{thm:ed-coloring-intro} in mind, we give a reduction from $\coloringED{q}$. Let $(G,z)$ be an $N$-vertex instance of $\coloringED{q}$, where $z$ is the deletion budget amd the graph $G$ is given with a \core{\sigma}{\delta} $Q$.

Let $S$ be a maximum incomparable set in $H$, i.e., $|S|=i(H)=q$.
Let $\calJ = (J,L',(x,y))$ be a gadget that is $1$-realizing $\NEQ(S) := \{ (u,v) \in S \times S ~|~ u \neq v \}$ as given by \cref{lem:ed-allrelomega}.
Let $\alpha := \edcount(\calJ \to H, (a,b))$, where $a,b$ are arbitrarily chosen distinct elements of $S$.
Since $\calJ$ realizes $\NEQ(S)$, the value of $\alpha$ does not depend on the choice of $a,b$.

\paragraph{An instance of $\LHomED(H)$.}
We define an instance $(G',L',z')$ of $\LHomED(H)$ with deletion budget $z'$.
Let $G'$ be obtained from $G$ by substituting each edge $uv$ by a copy of $\calJ$,
where $x$ is identified with $u$, and $y$ is identified with $v$.
The lists $L'$ are inherited from $\calJ$; note that they are well-defined 
as each inner (non-portal) vertex is in one copy of $\calJ$ and the lists of portal vertices, i.e., the original vertices of $G$, are all equal to $S$.
Set $z' := \alpha \cdot |E(G)| +z$.

\paragraph{Equivalence of instances.}
Fix some bijection $f$ between $S$ and $[q]$.
Suppose that $G$ is $q$-colorable after (at most) $z$ edges are deleted.
This, combined with $f$, gives a mapping of the vertices of $G$ to $S$ so that, in $G'$, the portals of all but $z$ copies of $\calJ$ are mapped to different elements of $S$.
We can extend this mapping to a list homomorphism by removing $\alpha$ edges from each gadget with dichromatic portals, and $\alpha+1$ edges from each gadget with monochromatic portals. (Here we use the fact that $\calJ$ is a $1$-realizer of $\NEQ(S)$ and therefore the violation cost is always $\alpha+1$.) This yields a total of (at most) $z'$ edge deletions. 

Vice versa, a minimum size set $X$ of at most $z'$ edge deletions from $G'$ that ensures a list homomorphism $h$ from $G'\setminus X$ to $H$ means that $h$ maps the portals of at least $\abs{E(G)}-z$ copies of $\calJ$ to distinct vertices of $S$. Consequently, the restriction of $h$ to the vertices of $G$ gives a list $q$-coloring of $G$ that requires at most $z$ edge deletions.

\paragraph{Structure of the constructed instance.}
The size of $\calJ$ depends only on $H$, and we introduce $E(G)$ copies of $\calJ$. Consequently, the size of $G'$ is $N^{\bigO(1)}$.
Recall that $Q$ is the given \core{\sigma}{\delta} of $G$. Note that the connected components of $G'-Q$ are subsets of some copy of $\calJ$ or otherwise are connected components of $G-Q$ in which all edges are replaced by copies of $\calJ$. As $\calJ$ only has two portals, each such component has at most $\max\{\delta, 2\}$ neighbors in $Q$.
The size of a connected component of $G'-Q$ depends only on $\sigma$ and $H$. Consequently, $Q$ is a \core{\sigma'}{\delta'} of $G'$, where $\sigma'$ depends only on $\sigma$ and $H$, and $\delta'=\max\{\delta, 2\}$.

\paragraph{Running time.}
As the size of $G'$ is polynomial in $N$, the hypothetical algorithm for $\LHomED(H)$ would require $(i(H)-\epsilon)^p \cdot N^{\bigO(1)}$ time to solve  $\coloringED{q}$. This contradicts the SETH for some $\sigma$ and $\delta$ depending on $\eps$ (and consequently for some $\sigma'$ and $\delta'$ depending on $\eps$), which shows the theorem.
\end{proof}

Before we proceed to the proof of \cref{thm:ed-main-intro}, we need one more lemma.

\begin{lem}\label{lem:ed-obstructionssurvive}
If $H$ contains an obstruction, then it contains an induced undecomposable subgraph $H'$ that has an obstruction.
\end{lem}
\begin{proof}
If $H$ is undecomposable, then the statement is trivial.
Suppose then that $H$ admits a decomposition $(A,B,C)$. Note it is sufficient to show that either $H[A]$ or $H[B \cup C]$ contains an obstruction.
Indeed, then the statement will follow by applying the argument to $H[A]$ or $H[B \cup C]$ inductively, until we reach an undecomposable graph.

If $H$ contains an irreflexive edge, then this edge must be contained in $A$.
So now suppose that $H$ contains a three-element set $O$ with private or co-private neighbors. Note that any two vertices from $O$ are incomparable.
Since for all $a \in A$, $b \in B$, $c \in C$ we have $\nh(c) \subseteq \nh(a) \subseteq \nh(b)$, we observe that $O$ must be fully contained in one of $A,B,C$.

Let $O'$ be the set of (co-)private neighbors of $O$.
Each vertex from $O'$ is adjacent to some vertex in $O$ and non-adjacent to some vertex in $O$.
Thus if $O \subseteq A$, then $O' \subseteq A$.
If $O \subseteq B$, then $O' \subseteq C$.
If $O \subseteq C$, then $O' \subseteq B$.
Summing up, one of $H[A]$ or $H[B \cup C]$ contains an obstruction.
\end{proof}

Now let us show that \cref{thm:ed-hardness} implies \cref{thm:ed-main-intro}.

\begin{proof}[Proof of \cref{thm:ed-main-intro}.]
By \cref{lem:ed-obstructionssurvive} we know that there exists an undecomposable subgraph of $H$ that contains an obstruction.
Let $H'$ be such a subgraph which maximizes $i(H')$. Recall that we have $\param(H) = i(H')$.
Now the theorem follows from applying \cref{thm:ed-hardness} to $H'$.
\end{proof}

\subsubsection{Moving between $2$-vertex incomparable sets}\label{sec:2vertexmoves}

Towards proving \cref{lem:indicator} in \cref{sec:indicator}, the current section deals with the use of binary gadgets in order to ``move'' (formally defined in \cref{def:moves}) from one pair of vertices to another while maintaining certain properties. The goal is to prove \cref{lem:movingMain}, the statement of which uses the notation from \cref{def:movesforpairs}. 
We first introduce some helpful notation and preliminary results about binary gadgets. Many of these results generalize to $r$-ary gadgets in a straightforward way but for us it suffices to consider the binary setting.

\begin{defn}[Moves and forcing]\label{def:moves} 
	Let $H$ be a graph. We say that a binary gadget $\calJ=(J,L,(x,y))$ is a \emph{move} from $L(x)$ to $L(y)$.
	For $(a,b) \in L(x)\times L(y)$,
	$\edcount(\calJ \to H, (a,b))$ is the size of a smallest set $F \subseteq E(J)$
	that ensures a list homomorphism $\phi$ from $(V(J), E(J)\setminus F)$ to $H$ such that $\phi(x)=a$ and $\phi(y)=b$.
	For each $a\in L(x)$, let $\edcount(\calJ \to H, (a,\ast))=\min_{b\in L(y)}\edcount(\calJ \to H, (a,b))$. Analogously we define $\edcount(\calJ \to H, (\ast,b))$.
	Then the \emph{base cost} of $\calJ$ is $\edcount(\calJ)= \min_{a\in L(x),b\in L(y)}\edcount(\calJ \to H, (a,b))$.
	Finally, let $\calJ(a)$ be the set of vertices that can be reached from $a$ by moving via $\calJ$ with minimum edge deletions, that is, the set of vertices $b$ for which $\edcount(\calJ \to H, (a,b))=\edcount(\calJ \to H, (a,\ast))$.

	Given vertices $a\in L(x)$ and $b\in L(y)$, the move $\calJ$ \emph{allows} $a\to b$ if $b\in \calJ(a)$. The move \emph{forces} $a\to b$ if $\calJ(a)=\{b\}$.
\end{defn}

We will be interested in moves between \emph{incomparable} sets $L(x)$ and $L(y)$.
In \cref{lem:addcost1}, we establish that in this case we can add a cost to moving from or moving to some particular element in $L(x)$ or $L(y)$, respectively. In \cref{lem:moverealizable}, we will use this fact to normalize the costs such that $\edcount(\calJ \to H, (a,\ast))$ is the same (the base cost $\edcount(\calJ)$) for each element $a$ in $L(x)$. This shows that for each gadget $\calJ$ we can realize the relation $R\coloneqq \{(u,v) \mid u\in S, v\in \calJ(u)\}$, and therefore it suffices to focus on what elements are in $\calJ(u)$, rather than on the actual number of required edge deletions $\edcount(\calJ \to H, (u,\ast))$. This is why the concepts of ``forcing'' and ``allowing'' from \cref{def:moves} can be defined with respect to $\calJ(a)$.

\begin{lem}\label{lem:addcost1}
	Let $\calJ=(J,L,\{x,y\})$ be a move from $L(x)$ to $L(y)$ with $a\in L(x)$ and $c\in L(y)$. 
	If $L(x)$ is incomparable then, for each non-negative integer $k$, there is a move $\calJ'$ from $L(x)$ to $L(y)$ such that for all $b\in L(x)\setminus\{a\}$ and $v\in L(y)$ we have
	\begin{myitemize}
		\item $\edcount(\calJ' \to H, (a,v))=\edcount(\calJ \to H, (a,v))+k$, and
		\item $\edcount(\calJ' \to H, (b,v))=\edcount(\calJ \to H, (b,v))$.
	\end{myitemize} 
	
	Analogously, if $L(y)$ is incomparable then, for each non-negative integer $k$, there is a move $\calJ'$ from $L(x)$ to $L(y)$ such that for all $v\in L(x)$ and $d\in L(y)\setminus\{c\}$ we have
	\begin{myitemize}
		\item $\edcount(\calJ' \to H, (v,c))=\edcount(\calJ \to H, (v,c))+k$, and
		\item $\edcount(\calJ' \to H, (v,d))=\edcount(\calJ \to H, (v,d))$.
	\end{myitemize}
\end{lem}
\begin{proof}
	We will prove the first part of the statement. The second part can be shown analogously. 
	
	Let $a\in L(x)$ and let $k$ be a non-negative integer. We obtain $\calJ'$ from $\calJ$ by adding $k$ pendant vertices to $x$, each with list $V(H)\setminus \nh(a)$. Note that for each $b\in L(x)\setminus\{a\}$ there is a vertex $b'\in \nh(b)\setminus \nh(a)\subseteq V(H)\setminus\nh(a)$ by the assumption that $L(x)$ is incomparable.
	Therefore, if $x$ is mapped to $b$ then all the pendant vertices can be mapped to $b'$ without requiring additional edge deletions. So, for each $v\in L(y)$, $\edcount(\calJ' \to H, (b,v))=\edcount(\calJ \to H, (b,v))$.
	However, if $x$ is mapped to $a$ then the pendant vertices cannot be mapped to a neighbor of $a$, which requires $k$ edge deletions, and we obtain $\edcount(\calJ' \to H, (a,v))=\edcount(\calJ \to H, (a,v))+k$ for each $v\in L(y)$.
\end{proof}

Recall the definition of realizing a relation from \cref{def:realizing}.

\begin{lem}\label{lem:moverealizable}
	Let $\calJ$ be a move from $S$ to $S'$. If $S$ is an incomparable set then there is a move $\calJ'$ from $S$ to $S'$ that realizes $R\coloneqq \{(u,v) \mid u\in S, v\in \calJ(u)\}$.
\end{lem}
\begin{proof}
	For $\calJ'$ to realize $R$ we have to ensure that there is some integer $r$ such that, for each $u\in S$, we have $\calJ'(u)=\calJ(u)$ and $\edcount(\calJ' \to H, (u,\ast))=r$. Let $k^* \coloneqq \max_{u\in S} \edcount(\calJ \to H, (u,\ast))$. By applying \cref{lem:addcost1} to some $a\in S$ and $k\coloneqq k^*- \edcount(\calJ \to H, (a,\ast))$ we obtain a gadget $\calJ''$ such that
	\begin{itemize}
		\item $\edcount(\calJ'' \to H, (a,\ast))=k^*$,
		\item for each $u\in S$ we have $\calJ''(u)=\calJ(u)$, and
		\item $\max_{u\in S} \edcount(\calJ'' \to H, (u,\ast))=k^*$.
	\end{itemize} 
	So we can do this iteratively for each $a\in S$ to obtain the sought-for gadget $\calJ'$.
\end{proof}

We will mainly be interested in moves between pairs of vertices and introduce some further shorthand notation for this setting. In the following, we will often write ``pair'' if we mean a $2$-vertex set.

\begin{defn}\label{def:movesforpairs}
	We write $(a,b)\to (c,d)$ if there is a move from $\{a,b\}$ to $\{c,d\}$ that forces both $a\to c$ and $b\to d$. Similarly, we write $\{a,b\}\leadsto \{c,d\}$ if at least one of $(a,b)\to (c,d)$ or $(a,b)\to (d,c)$.
\end{defn}

\begin{lem}\label{lem:movetransitive}
	The shorthand from \cref{def:movesforpairs} is transitive in the sense that if $(a,b)\to (c,d)$ and $(c,d)\to (e,f)$ then $(a,b)\to (e,f)$. Similarly, if $\{a,b\}\leadsto \{c,d\}$ and $\{c,d\}\leadsto \{e,f\}$ then $\{a,b\}\leadsto \{e,f\}$.
\end{lem}
\begin{proof}
	We show that $(a,b)\to (c,d)$ and $(c,d)\to (e,f)$ implies $(a,b)\to (e,f)$. The other statement can be shown analogously.
	Let $\calJ_1$ be a move from $\{a,b\}$ to $\{c,d\}$ that forces $a\to c$ and $b\to d$, and let $\calJ_2$ be a move from $\{c,d\}$ to $\{e,f\}$ that forces $c\to e$ and $d\to f$. 
	By \cref{lem:moverealizable}, there are gadgets $\calJ_1'=(J_1,L_1,(x_1,y_1))$ and $\calJ_2'=(J_2,L_2,(x_2,y_2))$ that realize $R_1= \{(u,v) \mid u\in \{a,b\}, v\in \calJ(u)\}$ and $R_2= \{(u,v) \mid u\in \{c,d\}, v\in \calJ(u)\}$, respectively.

	Then let $J^*$ be the graph obtained by taking a copy of $J'_1$ and a copy of $J'_2$ and identifying $y_1$ with $x_2$. Then 
	$\calJ^*=(J^*, L_1\cup L_2, (x_1, y_2))$ is a move from $\{a,b\}$ to $\{e,f\}$. Note that $\calJ^*$ is well-defined since $L_1(y_1)=L_2(x_2)=\{c,d\}$. Moreover, for $u\in \{a,b\}$ and $v\in \{e,f\}$ we have 
	\[
	\edcount(\calJ^* \to H, (u,v))=\min_{w\in \{c,d\}} \bigl(\edcount(\calJ'_1 \to H, (u,w)) + \edcount(\calJ'_2 \to H, (w,v))\bigr).
	\]
	It is straight-forward to verify that $\calJ^*$ forces $a\to e$: 
	\begin{itemize}
		\item $\edcount(\calJ'_1 \to H, (a,c))<\edcount(\calJ'_1 \to H, (a,d))$ as $\calJ'_1(a)=\calJ(a)=\{c\}$, and
		\item  $\edcount(\calJ'_2)=\edcount(\calJ'_2 \to H, (c,e))=\edcount(\calJ'_2 \to H, (d,f))$ as $\calJ'_2$ realizes $R_2$.
	\end{itemize} 
	Analogously, one can see that $\calJ^*$ forces $b\to f$.
\end{proof}

Note that if $(J,L, (x,y))$ is a move that gives $\{a,b\}\leadsto \{c,d\}$ then $(J,L, (y,x))$ gives $\{c,d\}\leadsto \{a,b\}$.
\begin{obs}\label{obs:reverseforcer}
	If $\{a,b\}\leadsto \{c,d\}$ then $\{c,d\}\leadsto \{a,b\}$.
\end{obs}

It turns out that for moves between \emph{incomparable} pairs $\{a,b\}$ and $\{c,d\}$, in order to ensure $(a,b)\to(c,d)$, it suffices to find a move from $\{a,b\}$ to $\{c,d\}$ that forces $a\to c$ and \emph{allows} $b\to d$.\footnote{In previous works~\cite{DBLP:conf/esa/OkrasaPR20,DBLP:conf/soda/FockeMR22} allowing and forcing was also referred to as distinguishing and forcing.}

\begin{lem}\label{lem:distinguishertoforcer}
	Let $\{a,b\}$ and $\{c,d\}$ be incomparable pairs in $H$.
	If there is a move $\calJ$ from $\{a,b\}$ to $\{c,d\}$ that forces $a\to c$ and allows $b\to d$ then $(a,b)\to (c,d)$.
\end{lem}
\begin{proof}
	Let $\calJ=(J,L,(x,y))$. If $\calJ(b)=\{d\}$ then $\calJ$ actually forces $b\to d$, and we are done. So suppose that $\calJ(b)=\{c,d\}$.
		
	Intuitively, we will put two copies of $\calJ$ in parallel to force both $a\to c$ and $b\to d$. However, this is problematic if $J$ contains the edge $xy$ since we cannot introduce multiple edges. A simple modification salvages this situation:
	If $J$ contains the edge $xy$ then let $J'$ be obtained from $J$ by replacing the edge $xy$ by a path $x-x_1-x_2-y$ where $x_1$ and $x_2$ are new vertices. Since $\{a,b\}$ is incomparable there are $a'\in \nh(a)\setminus \nh(b)$ and $b'\in \nh(b)\setminus \nh(a)$. We set $L(x_1)=\{a', b'\}$ and $L(x_2)=\{a,b\}$.
	In order to obtain a list homomorphism from $(J',L)$ to $H$ the new edges never have to be deleted and $x$ is mapped to $a$ if and only if $x_2$ is mapped to $a$ (and $x_1$ is mapped to $a'$). Thus, it is straight-forward to see that $(J',L,(x,y))$ is also a move from $\{a,b\}$ to $\{c,d\}$ that forces $a\to c$ and allows $b\to d$.
	
	So now let us assume without loss of generality that $J$ does not contain the edge $xy$.  In this case let $\calJ'= (J',L',(x',y'))$ be a copy of the gadget $\calJ$. Let $J''$ be the graph obtained from $J$ and $J'$ by identifying $x$ with $x'$ and $y$ with $y'$. Then let $\calJ''=(J'', L\cup L', \{x,y\})$.
	Recall that  $\calJ(a)=\{c\}$ and therefore $\edcount(\calJ \to H, (a,c))\le\edcount(\calJ \to H, (a,d))-1$. This means that $\edcount(\calJ'' \to H, (a,c))\le\edcount(\calJ'' \to H, (a,d))-2$.
	Furthermore, recall that we have $\calJ(b)=\{c,d\}$, i.e., $\edcount(\calJ \to H, (b,c))=\edcount(\calJ \to H, (b,d))$. Thus, we also have $\edcount(\calJ'' \to H, (b,c))=\edcount(\calJ'' \to H, (b,d))$.

	Then we apply \cref{lem:addcost1} (twice) using the fact that $\{a,b\}$ and $\{c,d\}$ are incomparable to add cost $+1$ to starting from $b$ as well as cost $+1$ to finishing on $c$. So, \cref{lem:addcost1} ensures that there is a gadget $\calJ^{*}$ with $\edcount(\calJ^{*} \to H, (a,c))\le\edcount(\calJ^{*} \to H, (a,d))-1$ and $\edcount(\calJ^{*} \to H, (b,c))= \edcount(\calJ^* \to H, (b,d))+1$. This ensures $\calJ^*(a)=\{c\}$ and $\calJ^*(b)=\{d\}$, as required.
\end{proof}

We will repeatedly use the crucial fact that if two incomparable pairs have a common vertex then there is a forced move from one pair to the other. Here is the formal statement.
\begin{lem}\label{lem:adjincomp}
	Let $\{a,b\}$ and $\{a,c\}$ be  incomparable pairs in $H$. Then $\{a,b\}\leadsto\{a,c\}$.
\end{lem}
\begin{proof}
	Let $b'\in \nh(b)\setminus \nh(a)$ and $c'\in \nh(c)\setminus \nh(a)$.
	\begin{itemize}
		\item If there is $a'\in \nh(a)\setminus (\nh(b)\cup \nh(c))$ then the path gadget $\{a,b\}-\{a',b'\}-\{a,c\}$ forces $a\to a$ and $b\to c$. Note that $b=c$ forms a subcase of this case.
		\item Otherwise we have $\nh(a)\subseteq (\nh(b)\cup \nh(c))$. Then let $a'\in \nh(a)\setminus\nh(b)$ and $a''\in \nh(a)\setminus\nh(c)$ (exist because of incomparability). Note that $a'\in \nh(c)$ and $a''\in \nh(b)$.
		Then $\{a,b\}-\{a,'b'\}-\{c,b\} - \{c',a''\} -\{c,a\}$ forces $a\to c$ (requiring zero edge deletions, whereas $a\to a$ requires at least one) and allows $b\to a$ (also requiring zero edge deletions). Then we can apply \cref{lem:distinguishertoforcer} to obtain a gadget that also forces $b\to a$.
	\end{itemize}\qedhere
\end{proof}

In order to easily apply \cref{lem:adjincomp} we introduce a helpful auxiliary graph.
\begin{defn}[$\aux(H)$]
	Given a graph $H$, let $\aux(H)$ be the graph whose vertices are the incomparable pairs in $H$. Two such sets are adjacent if and only if their intersection is non-empty. 
\end{defn}

\begin{cor}\label{cor:connectedinAux}
	Let $\{a,b\}$ and $\{c,d\}$ be (not necessarily distinct)  incomparable pairs in $H$ that are in the same connected component of $\aux(H)$. Then $\{a,b\}\leadsto \{c,d\}$.
\end{cor}
\begin{proof}
	For adjacent vertices $P,P'$ in $\aux(H)$ we have $P\leadsto P'$ by \cref{lem:adjincomp}. The statement then follows from iterative use of \cref{lem:adjincomp} together with \cref{lem:movetransitive}.
\end{proof}

At this point let us state the main result of this section.

\lemmovingmain*

We split the proof into two parts depending on whether or not the graph $H$ is a \emph{strong split graph}, that is, a graph whose vertices are covered by a reflexive clique and an irreflexive independent set. Let us state two crucial results that we will prove subsequently in two separate parts.

For a graph $H$ let $\aux^*(H)$ be the induced subgraph of $\aux(H)$ whose vertices are the incomparable pairs in $H$ that contain only reflexive vertices. 
\begin{restatable}{lem}{lemauxHssg}
	\label{lem:auxHssg}
	Let $H$ be an undecomposable strong split graph. Then $\aux^*(H)$ is connected.
\end{restatable}

For a graph $H$ we say that a vertex $\{a,b\}$ of $\aux(H)$ is \emph{bad} if $a,b$ are irreflexive non-adjacent vertices such that the sets $\nh(a)\setminus \nh(b)$ and $\nh(b)\setminus \nh(a)$ contain only reflexive vertices. Otherwise, $\{a,b\}$ is \emph{good}. Let $\aux^g(H)$ be the subgraph of $\aux(H)$ induced by its good vertices.
\begin{restatable}{lem}{lemauxHnossg}
	\label{lem:auxHnossg}
	Let $H$ be an undecomposable graph that is not a strong split graph. Then $\aux^g(H)$ is connected.
\end{restatable}

Before proving \cref{lem:auxHssg,lem:auxHnossg} let us show how these results constitute the proof of \cref{lem:movingMain}.

\begin{proof}[Proof of \cref{lem:movingMain}]
	First, suppose that $H$ is a strong split graph. In order for $a$ and $b$ to be incomparable they are either both irreflexive or both reflexive. The same holds for $c$ and $d$.
	Suppose $\{a,b\}$ is irreflexive with private neighbors $a'\in \nh(a)\setminus \nh(b)$ and $b'\in \nh(b)\setminus \nh(a)$. Then the path gadget $\{a,b\}-\{a',b'\}$ gives $(a,b)\to (a',b')$. Then, by \cref{lem:movetransitive}, it suffices to show that $\{a',b'\}\leadsto \{c,d\}$. Since $H$ is a strong split graph $\{a',b'\}$ is reflexive.
	Similarly, if $\{c,d\}$ is irreflexive there is a reflexive pair $\{c',d'\}$ with $\{c',d'\}\leadsto \{c,d\}$.
	Then the statement follows from \cref{cor:connectedinAux} and the fact that $\aux^*(H)$ is connected by \cref{lem:auxHssg}.
	
	Second, suppose that $H$ is not a strong split graph. Suppose $\{a,b\}$ is a bad vertex of $\aux(H)$. This means that $a$ is irreflexive and there is a reflexive $a'\in \nh(a)\setminus \nh(b)$.
	Again, the path gadget $\{a,b\}-\{a',b'\}$ gives $(a,b)\to (a',b')$, where $\{a',b'\}$ now is a good vertex in $\aux^g(H)$ and consequently a vertex in $\aux^g(H)$.
	So we have shown that there is a vertex $P$ of $\aux^g(H)$ such that $\{a,b\}\leadsto P$. Similarly, one can show that if $\{c,d\}$ itself is not a vertex of $\aux^g(H)$ then there is a vertex $P'$ of $\aux^g(H)$ with $P'\leadsto \{c,d\}$.	
	Then the statement follows from \cref{cor:connectedinAux} and the fact that $\aux^g(H)$ is connected by \cref{lem:auxHnossg}.
\end{proof}

\paragraph{{\boldmath Proof of \texorpdfstring{\cref{lem:auxHssg}}{Lemma \ref{lem:auxHssg}}: $H$ is a strong split graph}.}

We are interested in utilizing the fact that some graph is undecomposable. To this end, we investigate how to verify that a graph (here, a strong split graph) has this property. 
Recall that the vertices of a strong split graph are covered by a reflexive clique and an irreflexive independent set. Alternatively, this means that a strong split graph has no irreflexive edge and no reflexive non-edge.

\begin{algorithm}
	\caption{Algorithm to determine whether a strong split graph $H$ has a decomposition.}
	\label{algo:certificate1}
	\begin{algorithmic}
		\Require Graph $H=(V,E)$.
		\State Initialize $B$ as the set of maximal vertices in $V$ and $C$ as the empty set.
		\Loop
			 \If{there is an irreflexive vertex $v\in V\setminus B$ that has a non-edge to $B$}
			 	\State add $v$ to $C$.
			 \ElsIf{there is a reflexive vertex $v\in V\setminus B$ that has an edge to $C$}
			 	\State add $v$ to $B$.
			 \Else
			 	\Break
			 \EndIf
		\EndLoop
		\Ensure $B,C$
	\end{algorithmic}
\end{algorithm}

 Recall from \cref{sec:prelims} that a vertex $v$ in $H$ is universal if $\nh(v)=V(H)$, and it is isolated if $\nh(v)=\emptyset$.
 
\begin{lem}\label{lem:decompalg1}
	Let $H$ be a strong split graph without universal vertices and without isolated vertices.
	Then $H$ has a decomposition if and only if the sets $B$ and $C$ returned by \cref{algo:certificate1} satisfy $B\cup C \neq V(H)$.
\end{lem}
\begin{proof}
	Suppose $H$ has a decomposition $(A',B',C')$. We show that $(B\cup C)\subseteq (B'\cup C')$ --- then $B\cup C \neq V(H)$ follows from the fact that $B'\cup C' \neq V(H)$ as $A'\neq \emptyset$.
	
	Note that if $C'$ is empty then vertices in $B'$ are universal vertices in $H$, a contradiction. Thus, from the fact that $B'\cup C'$ is non-empty it follows that $C'$ is non-empty. Furthermore, there are no isolated vertices in $H$ and so every vertex in $C'$ has at least one neighbor (and neighbors of vertices in $C'$ have to be in $B'$). Thus, there are vertices in $B'$ that \emph{strictly} dominate all vertices in $A$ and $C$ and consequently all maximal vertices of $H$ are contained in $B'$.
	Now assume that at the beginning of some iteration of \cref{algo:certificate1} we have $B\subseteq B'$ and $C\subseteq C'$. Any irreflexive vertex $v\in V$ has to be in one of $A'$ or $C'$. If $v$ has a non-edge to a vertex in $B\subseteq B'$ then it has to be in $C'$. Similarly, any reflexive vertex $v\in V$ has to be in one of $A'$ or $B'$, and if $v$ has an edge to a vertex in $C\subseteq C'$ then it has to be in $B'$. Thus, at the end of this iteration of \cref{algo:certificate1} we still have $B\subseteq B'$ and $C\subseteq C'$. Inductively, we obtain that $B\subseteq B'$ and $C\subseteq C'$, as required.
	
	For the other direction, suppose that \cref{algo:certificate1} returns $B,C$ with $B\cup C\neq V(H)$. From the initialization of $B$ it follows that $B$ is non-empty. Let $A=V(H)\setminus(B\cup C)$. Then $A$ is non-empty as $B\cup C\neq V(H)$. We claim that $(A,B,C)$ is a decomposition. As $H$ is a strong split graph, all maximal vertices are reflexive. So it is clear that \cref{algo:certificate1} places only reflexive vertices in $B$, and only irreflexive vertices in $C$. Furthermore, all reflexive vertices in $H$ form a clique and there are no irreflexive edges, which means that the vertices in $B$ form a reflexive clique and the vertices in $C$ form an irreflexive independent set.
	Suppose that a vertex $v\in A$ has a non-edge to a vertex in $B$. Then $v$ would need to be irreflexive, but then \cref{algo:certificate1} would have placed it in $C$, a contradiction. Similarly, if $v\in A$ has an edge to a vertex in $C$. Then $v$ would need to be reflexive, but then \cref{algo:certificate1} would have placed it in $B$.
\end{proof}

Suppose that $H$ is free of universal and isolated vertices and does not have a decomposition. According to \cref{lem:decompalg1}, \cref{algo:certificate1} will place all vertices of $H$ in one of $B$ or $C$. Essentially, this gives a set of certificates (one for each vertex in $H$) that certify that $H$ is not decomposable. 
It turns out that such a certificate yields a path in $\aux(H)$ that can be used to move from any two incomparable vertices to two maximal incomparable vertices. 

\begin{lem}\label{lem:incompToMaximal}
	Let $H$ be an undecomposable strong split graph without universal or isolated vertices. For every reflexive  incomparable pair $\{a,b\}$ 
	in $H$ there is an incomparable pair of maximal vertices $\{c,d\}$ such that $\{a,b\}$ and $\{c,d\}$ are in the same connected component of $\aux^*(H)$. 
\end{lem}
\begin{proof}
	Since $\{a,b\}$ is reflexive and $H$ is undecomposable, by \cref{lem:decompalg1}, $a$ is in the set $B$ returned by \cref{algo:certificate1}.
	Now consider the reason why $a$ was placed into $B$: it is either maximal, or it is adjacent to some irreflexive vertex $r_1$ that was placed in $C$, $r_1$ in turn was placed into $C$ because it is non-adjacent to some reflexive vertex in $B$, etc. An illustration is given in \cref{fig:witness_ss}. Thus, from the definition of \cref{algo:certificate1} it follows that there is a sequence of vertices $\ell_0, r_1, \ell1, \ldots, r_k, \ell_k$ for some $k\ge 0$ such that
	\begin{enumerate}
		\item $\ell_0=a$ and $\ell_k$ is maximal,
		\item for each $i\in \{0, \dots, k\}$, $\ell_i$ is reflexive and $r_i$ is irreflexive, \label{item:cert1altervertices}
		\item for each $i\in\{0, \dots, k-1\}$, $\ell_i r_{i+1}$ is an edge, \label{item:cert1edge}
		\item for each $i\in \{1, \dots, k\}$, there is no edge between $\ell_i$ and $r_{i}$. \label{item:cert1noedge}
	\end{enumerate}
	Without loss of generality, let this be a shortest sequence with such properties. Then we set $c\coloneqq \ell_k$ (so, if $k=0$ then $a=c$), and $d$ is some maximal vertex with $\nh(d)\neq \nh(c)$. Note that such a vertex $d$ has to exist otherwise $c$ is a universal vertex. 
	
	Any two maximal vertices with different neighborhoods are incomparable, which means that $\{c,d\}$ is a vertex in $\aux^*(H)$. It remains show that $\{a,b\}$ and $\{c,d\}$ are in the same connected component of $\aux^*(H)$.
	
	The fact that there is a path from $\{a,b\}$ to $\{c,d\}$ in $\aux^*(H)$ is trivial if $k=0$, i.e., if $a=c$. So suppose that $k\ge 1$. We will show that for each $i\in \{0, \ldots, k-1\}$ the pair $s_i\coloneqq\{\ell_i, \ell_{i+1}\}$ is incomparable, which implies that for $s_0\coloneqq \{a,b\}$ and $s_k\coloneqq\{c,d\}$, the vertices $s_0,s_1,\ldots, s_{k}$ form a path in $\aux(H)$. According to \cref{item:cert1altervertices} all of these pairs $\{\ell_i, \ell_{i+1}\}$ are reflexive and hence $s_0,s_1,\ldots, s_{k}$ is a path in $\aux^*(H)$.
	
	If $i\in \{0, \ldots, k-2\}$ then $r_{i+1}\in \nh(\ell_i)\setminus \nh(\ell_{i+1})$ by \cref{item:cert1edge,item:cert1noedge}. Also, since $\ell_i$ is not adjacent to $r_{i+2}$ by minimality of the chosen sequence, we have $r_{i+2}\in \nh(\ell_{i+1})\setminus \nh(\ell_{i})$. This shows that $\{\ell_i, \ell_{i+1}\}$ is incomparable. We also have $r_{k}\in \nh(\ell_{k-1})\setminus \nh(\ell_{{k}})$ by \cref{item:cert1edge,item:cert1noedge}. Furthermore, as $c=\ell_k$ is maximal whereas $\ell_{k-1}$ is not maximal according to the minimality of the chosen sequence, there is also some vertex $v\in \nh(\ell_{k})\setminus \nh(\ell_{k-1})$, which means that $\{\ell_{k-1}, \ell_{k}\}$ is also incomparable.
\end{proof}

\lemauxHssg*
\begin{proof}
	Let $\{a,b\}$ and $\{c,d\}$ be (not necessarily disjoint) reflexive incomparable pairs in $H$.
	Let $U$ and $I$ be the set of universal vertices and the set of isolated vertices in $H$, respectively.
	
	Let $H'$ be the graph $H$ from which we remove all vertices in $U\cup I$. Since $\{a,b\}$ is incomparable there are  $a'\in \nh(a)\setminus \nh(b)$ and $b'\in \nh(b)\setminus\nh(a)$, where the neighborhoods are with respect to edges in $H$. Similarly, there are $c'\in \nh(c)\setminus \nh(d)$ and $d'\in \nh(d)\setminus\nh(c)$. This shows that $a,b,c,d$ are vertices in $H'$, and it also shows that $\{a,b\}$ and $\{c,d\}$ are incomparable in $H'$ as well. 
	Furthermore, suppose there is a decomposition $(A,B,C)$ of $H'$. Then $(A,B\cup U, C\cup I)$ is a decomposition of $H$. Thus, $H'$ has to be undecomposable as well.
	
	So $H'$ fulfills the requirements of \cref{lem:incompToMaximal} and we obtain not necessarily disjoint incomparable pairs of maximal vertices $\{x,y\}$ and $\{w,z\}$ such that $\{a,b\}$ and $\{x,y\}$ are in the same connected component of $\aux^*(H')$, and $\{c,d\}$ and $\{w,z\}$ are also in the same connected component of $\aux^*(H')$. We will show that $\{x,y\}$ and $\{w,z\}$ are in the same connected component of $\aux^*(H')$.
	
	If $x$ and $w$ have distinct neighborhoods in $H$ (in particular, this means that they are distinct vertices) then, as they are maximal vertices, they are incomparable. Then $\{x,w\}$ is a vertex in $\aux^*(H')$ and $\{x,y\}-\{x,w\}-\{w,z\}$ forms a path in $\aux^*(H')$.
	If $x$ and $w$ have identical neighborhoods in $H$ then the fact that $\{x,y\}$ is incomparable implies that $\{w,y\}$ is incomparable as well. Consequently, $\{x,y\}-\{w,y\}-\{w,z\}$ forms a path in $\aux^*(H')$.
	
	Summarizing, we have shown that $\{a,b\}$ and $\{c,d\}$ are in the same connected component of $\aux^*(H')$, which means that they are also in the same connected component of $\aux^*(H)$.
\end{proof}

\begin{figure}[h]
	\centering
	\includegraphics[scale=1,page=1]{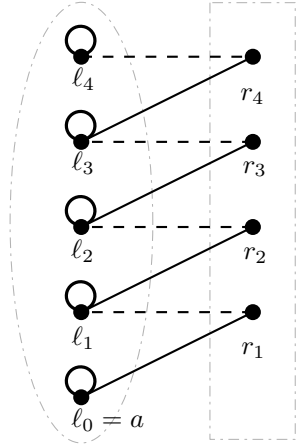}	  	
	\caption{An alternating path certifying that $a$ is moved to $B$ where $H$ is a strong split graph. Vertex $\ell_4$ is maximal.}
	\label{fig:witness_ss}
\end{figure}

\paragraph{{\boldmath Proof of \cref{lem:auxHnossg}: $H$ is not a strong split graph}.}

Note that every graph that is not a strong split graph contains an irreflexive edge or a reflexive non-edge.

\begin{algorithm}[h]
	\caption{Algorithm to determine whether a graph $H$ that contains an irreflexive edge or a reflexive non-edge has a decomposition.}
	\label{algo:certificate2}
	\begin{algorithmic}
		\Require Graph $H=(V,E)$.
		\State Initialize $A$ as the set of all vertices that are part of an irreflexive edge or reflexive non-edge in $H$.
		\Loop
		\If{there is an irreflexive vertex $v\in V\setminus A$ that has an edge to $A$}
		\State add $v$ to $A$.
		\ElsIf{there is a reflexive vertex $v\in V\setminus A$ that has a non-edge to $A$}
		\State add $v$ to $A$.
		\Else
		\Break
		\EndIf
		\EndLoop
		\Ensure $A$
	\end{algorithmic}
\end{algorithm}

\begin{lem}\label{lem:decompalg2}
	Let $H$ be a graph that is not a strong split graph. Then $H$ has a decomposition if and only if the set $A$ returned by \cref{algo:certificate2} satisfies $A\neq V(H)$.
\end{lem}
\begin{proof}
	Let $(A',B',C')$ be a decomposition of $H$. It is straight-forward to see by an inductive argument over the iterations of \cref{algo:certificate2} that $A\subseteq A'$. Then $A\neq V(H)$ follows from the fact that $B'\cup C'$ is non-empty and consequently $A'\neq V(H)$.
	
	In the other direction, let $A$ be the set returned by \cref{algo:certificate2}. We define $B$ and $C$ as the reflexive and irreflexive vertices in $V(H)\setminus A$, respectively. We show that $(A,B,C)$ is a decomposition.
	Note that $A\neq \emptyset$ since $H$ is not a strong split graph and therefore contains at least one irreflexive edge or reflexive non-edge. As $A\neq V(H)$ we have $B\cup C \neq \emptyset$. As every reflexive non-edge was initially placed in $A$ it follows that $B$ forms a reflexive clique. Analogously, $C$ is an irreflexive independent set. Suppose there is a non-edge between a (reflexive) vertex $b$ in $B$ and a vertex in $A$, then \cref{algo:certificate2} would have placed $b$ into $A$, a contradiction. Similarly, a vertex in $C$ with an edge to a vetrex in $A$ would have been placed into $A$.	
\end{proof}

\begin{figure}[t]
	\begin{subfigure}[b]{0.45\textwidth}
		\centering
		\includegraphics[scale=0.8,page=1]{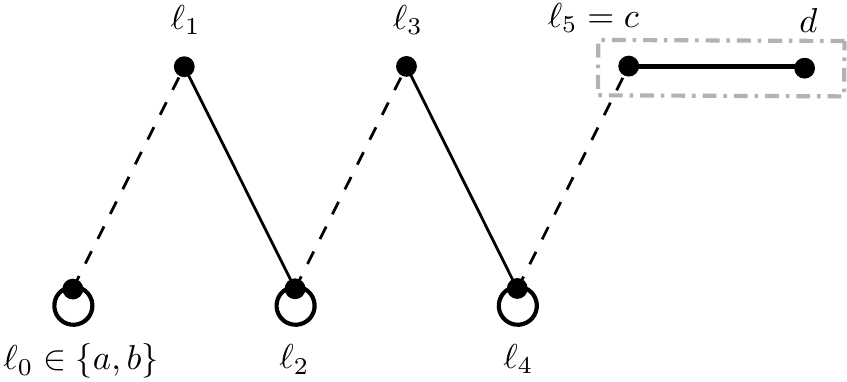}
		\caption{The certificate ends with an irreflexive edge.}
		\label{fig:witness_noss_irref_edge}
	\end{subfigure}
	\hfill
	\begin{subfigure}[b]{0.45\textwidth}
		\centering
		\includegraphics[scale=0.8,page=2]{figures/witness_noss}
		\caption{The certificate ends with a reflexive non-edge.}
		\label{fig:witness_noss_ref_nonedge}
	\end{subfigure}
	\caption{Two possible certificates for $a$ being in $A$ where $H$ is not a strong split graph.}
\end{figure}

\begin{lem}\label{lem:incompToTarget}
	Let $H$ be an undecomposable graph that is not a strong split graph. For every incomparable pair $\{a,b\}$ 
	in $H$ that is a vertex in $\aux^g(H)$ there is a pair $\{c,d\}$ in the same connected component of $\aux^g(H)$ such that $c,d$ form either an irreflexive edge or a reflexive non-edge in $H$. 
\end{lem}
\begin{proof}

The proof structure is similar to that of \cref{lem:incompToMaximal} but a number of modifications are necessary to handle the different certificate structure given by \cref{algo:certificate2}.
First suppose that both $a$ and $b$ are irreflexive. If $ab$ is an edge then we are done by setting $c=a$ and $d=b$. If $a$ and $b$ are not adjacent then, using the fact that $\{a,b\}$ is a vertex in $\aux^g(H)$ and therefore is a \emph{good} vertex of $\aux(H)$, it follows that one of $\nh(a)\setminus \nh(b)$ or $\nh(b)\setminus \nh(a)$ contain an irreflexive vertex. W.l.o.g.~say $a'\in \nh(a)\setminus \nh(b)$ is irreflexive. 
Then $\{a,a'\}$ is an incomparable set, where $aa'$ is an irreflexive edge. Hence $\{a,b\}$ and $\{a,a'\}$ are adjacent vertices in $\aux^g(H)$ and we are done by setting $c=a$ and $d=a'$.

Second, suppose that one of $a$ or $b$, w.l.o.g.~$a$, is reflexive.
Since $H$ is undecomposable, by \cref{lem:decompalg2}, $a$ is in the set $A$ returned by \cref{algo:certificate2}.
Now consider the reason why $a$ (or analogously $b$) was placed into $A$: Then $a$ is either part of a reflexive non-edge, or it is non-adjacent to some irreflexive vertex $\ell_1$ that was placed into $A$. We can then go on to consider why $\ell_1$ is in $A$. Either it is part of an irreflexive edge or it is adjacent to some reflexive vertex that was placed into $A$, etc. An illustration for the case in which the certificate for $a$ being in $A$ leads to an irreflexive edge or reflexive non-edge is given in \cref{fig:witness_noss_irref_edge,fig:witness_noss_ref_nonedge}, respectively.

Thus, from the definition of \cref{algo:certificate2} it follows that there is a sequence of vertices $\boldsymbol{{\ell}}=\ell_0, \ell_1, \ldots, \ell_k$ for some $k\ge 0$ such that
\begin{enumerate}
	\item $\ell_0\in \{a,b\}$ is reflexive,
	\item if $\ell_k$ is irreflexive then $\ell_k$ is part of an irreflexive edge, otherwise $\ell_k$ is part of a reflexive non-edge,\label{item:cert2end}
	\item $\ell_0\ell_1, \ell_1 \ell_2, \ldots, \ell_{k-1} \ell_k$ forms a sequence alternating between non-edges and edges, starting with the non-edge $\ell_0\ell_1$.\label{item:cert2alteredges}
\end{enumerate}

Without loss of generality, let $\boldsymbol{\ell}$ be a shortest sequence with such properties, and let $\ell_0=a$ (by renaming if necessary). Then we say that $\boldsymbol{\ell}$ is a certificate for $\{a,b\}$. The minimality of the sequence ensures another property, namely that
\begin{enumerate}
	\setcounter{enumi}{3}
	\item $\boldsymbol{\ell}$ is alternating between reflexive and irreflexive vertices. \label{item:cert2altervertices}
\end{enumerate}

Then we set $c\coloneqq \ell_k$ (so, if $k=0$ then $a=c$). By \cref{item:cert2end}, if $\ell_k=c$ is irreflexive then it is part of an irreflexive edge $\{c,d\}$, otherwise it is part of a reflexive non-edge $\{c,d\}$. In either case $\{c,d\}$ is an incomparable pair and a vertex in $\aux^g(H)$.
We will use induction on the length $k$ of the sequence $\boldsymbol{\ell}$ to show 
that $\{a,b\}$ and $\{c,d\}$ are in the same connected component of $\aux^g(H)$, as required.

\begin{description}
	\item[$k=0$] then $a=c$, which means that $\{a,b\}$ and $\{c,d\}$ are adjacent in $\aux^g(H)$.
	\item[$k=1$] Then $c=\ell_1$ is irreflexive and consequently $\{c,d\}$ forms an irreflexive edge. Since $\ell_0\ell_1$ is a non-edge $a$ and $c$ are not adjacent. Now we consider two cases. If $a$ and $d$ are adjacent then $\{a,d\}$ is an incomparable set (with private neighbors $a$ and $c$, respectively) and consequently, $\{a,b\} - \{a,d\} - \{c,d\}$ forms a path in $\aux^g(H)$. Otherwise, if $a$ and $d$ are not adjacent then $\{a,c\}$ is an incomparable set and $\{a,b\} - \{a,c\} - \{c,d\}$ is a path in $\aux^g(H)$.
	\item[$k=2$] Then $c=\ell_2$ is reflexive and consequently $\{c,d\}$ forms a reflexive non-edge. Then $a$ is adjacent to $d$ by minimality of the sequence $\boldsymbol{\ell}$. Consequently, $a$ and $c$ have private neighbors $d$ and $\ell_1$, respectively, and therefore $\{a,c\}$ is incomparable and again $\{a,b\} - \{a,c\} - \{c,d\}$ is a path in $\aux^g(H)$.
	\item[$k\ge 3$] Note that $a=\ell_0$ is not adjacent to $\ell_1$ (by \cref{item:cert2alteredges}) but consequently is adjacent to $\ell_3$ by the minimality of $\boldsymbol{\ell}$. From \cref{item:cert2alteredges} it also follows that $\ell_2$ is adjacent to $\ell_1$ but not to $\ell_3$. Thus, $\{a,\ell_2\}$ is incomparable and consequently $\{a,b\}$ and $\{a,\ell_2\}$ are adjacent in $\aux^g(H)$. Now note that a certificate for $\{a,\ell_2\}$ has length $k-2$, and therefore, by the induction hypothesis, we can assume that $\{a,\ell_2\}$ is in the same connected component of $\aux^g(H)$ as $\{c,d\}$, which implies that $\{a,b\}$ is as well.
\end{description}

\end{proof}

\lemauxHnossg*
\begin{proof}
	Let $\{a,b\}$ and $\{c,d\}$ be good vertices of $\aux(H)$.
	
	To shorten notation we say that a pair $\{x,y\}$ is a \emph{target} if $xy$ is either an irreflexive edge or reflexive non-edge.
	The graph $H$ fulfills the requirements of \cref{lem:incompToTarget} and we obtain not necessarily disjoint targets $\{x,y\}$ and $\{w,z\}$ such that $\{a,b\}$ and $\{x,y\}$ are in the same connected component of $\aux^g(H)$, and $\{c,d\}$ and $\{w,z\}$ are also in the same connected component of $\aux^g(H)$. We will show that $\{x,y\}$ and $\{w,z\}$ are in the same connected component of $\aux^g(H)$ as well.
	
	First note that if the intersection of $\{x,y\}$ and $\{w,z\}$ is non-empty then these two sets are adjacent in $\aux^g(H)$. Otherwise, if $x,y,w,z$ are distinct vertices we show that one pair in $\{x,y\}\times \{w,z\}$ is a vertex of $\aux^g(H)$ in order to obtain a path from $\{x,y\}$ to $\{w,z\}$ in $\aux^g(H)$. We consider three cases depending on the form of the targets.
	\begin{itemize}
		\item If $\{x,y\}$ and $\{w,z\}$ form reflexive non-edges then the vertices of each non-edge between $\{x,y\}$ and $\{w,z\}$ form an incomparable set, and thereby a vertex of $\aux^g(H)$ (as all vertices are reflexive). If all edges between the two sets are present then for instance $\{x,w\}$ is an incomparable set (with private neighbors $z$ and $y$, respectively).
		\item If $\{x,y\}$ and $\{w,z\}$ form irreflexive edges then, symmetrically to the previous case, the vertices of each edge between $\{x,y\}$ and $\{w,z\}$ form an incomparable set, and thereby a vertex of $\aux^g(H)$ (as these irreflexive vertices have irreflexive private neighbors). If no edges between the two sets are present then $\{x,w\}$ is an incomparable set (with private neighbors $y$ and $z$, respectively).
		\item Otherwise suppose that $\{x,y\}$ is a reflexive non-edge and $\{w,z\}$ is an irreflexive edge. (The other case is symmetric.) We make another case distinction.
		\begin{itemize}
			\item Suppose $x$ ($y$) has no edges  to $\{w,z\}$ then $\{x,w\}$ ($\{y,w\}$) is incomparable.
			\item Suppose $x$ ($y$) has exactly one edge to $\{w,z\}$, say $xw$ ($yw$), then $\{x,w\}$  ($\{y,w\}$) is incomparable.
			\item If both $x$ and $y$ have a complete set of edges to $\{w,z\}$ then $\{x,w\}$ is incomparable.
		\end{itemize}
		In each of these three cases the obtained incomparable set is a vertex of $\aux^g(H)$ as $x$ and $y$ are reflexive.
	\end{itemize}
\end{proof}

\subsubsection{Proof of \texorpdfstring{\cref{lem:indicator}}{Lemma \ref{lem:indicator}}: Constructing indicators}\label{sec:indicator}

Our goal is to prove \cref{lem:indicator}. This result relates an incomparable set $S$ of size $|S|\ge 2$ to a $2$-vertex set $\{a,b\}$. We start by showing that we can apply the results from \cref{sec:2vertexmoves} also in this setting.

\begin{lem}\label{lem:movefromS}
	Let $S$ be an incomparable set with $a,b\in S$. If $(a,b)\to (c,d)$ then there is a move from $S$ to $\{c,d\}$ that forces $a\to c$ and $b\to d$. 
\end{lem}
\begin{proof}
	Let $(J, L, (x,y))$ be an $(a,b)\to (c,d)$-move. Let $L'$ be the list assignment that is identical to $L$ with the exception that $L'(x)=S$,  whereas $L(x)=\{a,b\}$. It is straight-forward to see that $(J, L', (x,y))$ is a move from $S$ to $\{c,d\}$ that still forces $a\to c$ and $b\to d$ since $\calJ(a)=\calJ'(a)$ and $\calJ(b)=\calJ'(b)$.
\end{proof}

Let us recall the definition of an \emph{indicator} of a set $S$ of vertices from a graph $H$ over a $2$-vertex set $\{a,b\}$ from $H$.
It is a relation $I\subseteq S\times \{a,b\}^{|S|(|S|-1)}$ such that
\begin{myitemize}
	\item $I(u)$ is non-empty for each $u\in S$, and
	\item $I(u)$ and $I(v)$ are disjoint for distinct $u,v\in S$.
\end{myitemize}

We now have all the tools at hand to prove the main result of this section, which we restate for convenience.

\lemindicator*
\begin{proof}
	Let $x,y$ be distinct vertices from $S$ and let $a,b$ be distinct vertices from an obstruction $O$ in $H$. Note that $\{x,y\}$ is incomparable since $S$ is an incomparable set. The fact that $\{a,b\}$ is incomparable can be easily checked with the definition of obstruction (\cref{def:obstruction}).
	From \cref{lem:movingMain} together with \cref{lem:movefromS} it follows that there is a move $\calJ_{x,y}$ from $S$ to $\{a,b\}$ that forces $x$ to one of $a$ or $b$, and forces $y$ to the other. This means that $\calJ_{x,y}(x)$ and $\calJ_{x,y}(y)$ are disjoint.

	According to \cref{lem:moverealizable}, there is a gadget $\calJ'_{x,y}$ that realizes $R_{x,y}\coloneqq \{(u,v) \mid u\in S, v\in \calJ_{x,y}(u)\}$.
	Crucially, we have for each $u\in S$ that $R_{x,y}(u)=\calJ_{x,y}(u)$ is non-empty, and that $R_{x,y}(x)=\calJ_{x,y}(x)$ and $R_{x,y}(y)=\calJ_{x,y}(y)$ are disjoint. We will use these properties later.
	
	For each distinct $x,y\in S$ let $p^1_{x,y}$ be the first and let $p^2_{x,y}$ be the second distinguished vertex of the binary gadget $\calJ'_{x,y}$.
	We define a gadget $\calI$ by identifying all vertices in $\{p^1_{x,y} \mid x,y\in S\}$ to a single vertex $p$. The distinguished vertices of $\calI$ are then $p$ together with the set $\{p^2_{x,y} \mid x,y\in S\}$ for a total of $1+|S|(|S|-1)$ distinguished vertices.
	
	We claim that $\calI$ realizes an indicator of $S$ over $\{a,b\}$, as required.
	Let $k=|S|(|S|-1)$ and let $f$ be some fixed bijection from $[k]$ to the set of distinct pairs $x,y\in S$.
	Then $\boldp=(p, p^2_{f(1)}, \ldots, p^2_{f(k)})$ are the distinguished vertices of $\calI$.
	Let $I$ be the set of tuples of the form $(u, \boldd(u))$ where $u\in S$ and $\boldd(u)=(d_1(u), \ldots, d_{k}(u))$ such that, for each $i\in [k]$, $d_i(u)\in R_{x,y}(u)$ where $f(i)=x,y$.
	As for each such pair $x,y$ the gadget $\calJ'_{x,y}$ realizes $R_{x,y}$, we have $\edcount(\calJ'_{x,y}\to H, (u,d_i(u)))=\edcount(\calJ'_{x,y})$ for each $u\in S$.
	Thus, the required edge deletions $\edcount(\calI \to H, (u,\boldd(u)))$ do not depend on $u$. Moreover, for $\boldz\in S\times \{a,b\}^k$, the number of required edge deletions $\edcount(\calI \to H,\boldz)$ is equal to the minimum number of required edge deletions $\edcount(\calI)$ if and only if $\boldz=(u, \boldd(u))$ for some $u\in S$.
	This shows that $\calI$ realizes $I$.
	
	It remains to show that $I$ is an indicator. We use the two crucial properties that we have verified previously: For each $u\in S$ and each distinct $x,y\in S$, $R_{x,y}(u)=\calJ_{x,y}(u)$ is non-empty and therefore there exists a tuple of the form $(u, \boldd(u))$, and hence $I(u)$ is non-empty. For distinct $u,v\in S$, each tuple in $I(u)$ is of the form $\boldd(u)$ and each tuple in $I(v)$ is of the form $\boldd(v)$, which means that for the index $i$ with $f(i)=u,v$ we have $d_i(u)\in R_{u,v}(u)$ and $d_i(v)\in R_{u,v}(v)$. Therefore, $\boldd(u)$ and $\boldd(v)$ are distinct as $R_{u,v}(u)$ and $R_{u,v}(v)$ are disjoint.	
\end{proof}

\newpage

\newpage

\section{An outlook regarding homomorphism deletion without lists}
\label{sec:without-list}
Recall that in the paper we studied vertex/edge deletion variants of the \emph{list} homomorphism problem.
Similarly we can define variants of the non-list problem \Hom($H$).

By \textsc{HomVD}($H$) (resp. \textsc{HomED}($H$)) we denote the problem which takes as an input a graph $G$ 
and asks to find a smallest set of vertices (resp. edges) of $G$ whose deletion modifies $G$ into a graph that admits a homomorphism to $H$.
The following problems are a natural analogue of our results.

\begin{prob}
For each graph $H$, find the best possible constant $c_{\text{vd}}(H)$,
such that for every constant $\sigma,\delta$, the \textsc{HomVD}($H$) problem can be solved in time $c_{\text{vd}}(H)^p \cdot n^{\bigO(1)}$
on $n$-vertex graphs given along with a \core{\sigma}{\delta} of size $p$.
\end{prob}
\begin{prob}
For each graph $H$, find the best possible constant $c_{\text{ed}}(H)$,
such that for every constant $\sigma,\delta$, the \textsc{HomED}($H$) problem can be solved in time $c_{\text{ed}}(H)^p \cdot n^{\bigO(1)}$
on $n$-vertex graphs given along with a \core{\sigma}{\delta} of size $p$.
\end{prob}

Let us first discuss some easy cases.
Note that if $H$ contains a reflexive vertex $z$, then both \textsc{HomVD}($H$) and \textsc{HomED}($H$) are trivial,
as it is enough to map every vertex of $G$ to $z$ and there is no need to delete anything.
Thus then $c_{\text{vd}}(H) = c_{\text{ed}}(H)=1$.

If $H$ is an edgeless graph, then \textsc{HomVD}($H$) is equivalent to finding a smallest set of vertices whose removal destroys all edges, i.e., to solving \textsc{Vertex Cover} (i.e., \coloringVD{1}).
On the other hand, \textsc{HomED}($H$) in this case is trivial, as we need to remove all edges from $G$.
Thus in this case $c_{\text{vd}}(H) = 2$ (by \cref{thm:vd-coloring-intro}) and $c_{\text{ed}}(H)=1$.

More generally, if every component of $H$ is an irreflexive clique, then \textsc{HomVD}($H$) (resp. \textsc{HomED}($H$)) are equivalent to \coloringVD{q} (resp. \coloringED{q}), where $q$ is the number of vertices in the largest component of $H$.
Thus, by \cref{thm:vd-coloring-intro,} and \cref{thm:ed-coloring-intro} we have $c_{\text{vd}}(H) = q+1$ and $c_{\text{ed}}(H)=q$.

So now consider the case that $H$ is irreflexive, bipartite, and contains at least one edge.
Then our task boils down to making $G$ bipartite by deleting as few vertices/edges as possible.
Thus in this case  \textsc{HomVD}($H$) is equivalent to \textsc{Odd Cycle Transversal} and  \textsc{HomED}($H$) is equivalent to \textsc{Max Cut}.
By \cref{thm:vd-coloring-intro,thm:ed-coloring-intro} we have $c_{\text{vd}}(H) = 3$ and $c_{\text{ed}}(H)=2$.

Finally, consider the case that $H$ is irreflexive, non-bipartite, and contains some component other than a complete graph. Recall that in this case \Hom($H$) is \NP-hard~\cite{DBLP:journals/jct/HellN90}, and thus so are \textsc{HomVD}($H$) and \textsc{HomED}($H$). Similarly to the case of \Hom($H$)~\cite{DBLP:journals/siamcomp/OkrasaR21}, we can assume that $H$ is connected and is a \emph{core}, i.e., it does not admit a homomorphism to its proper subgraph.
However, in the case of \Hom($H$) there is one more algorithmic idea that can be exploited.
The \emph{direct product} of graphs $H_1$ and $H_2$ is a graph $H_1 \times H_2$ with vertex set $V(H_1) \times V(H_2)$ in which
vertices $(x,y)$ and $(x',y')$ are adjacent if and only if $xx' \in E(H_1)$ and $yy' \in E(H_2)$.
It is straightforward to observe that a graph admits a homomorphism to $H_1 \times H_2$ if and only if 
it admits a homomorphism to $H_1$ and a homomorphism to $H_2$, see also~\cite{DBLP:journals/siamcomp/OkrasaR21}.

However, it is not clear how to combine this observation with deleting vertices or edges.
In particular, solving \textsc{HomVD}($H$) (resp. \textsc{HomED}($H$)) in the case that $H = H_1 \times  H_2$
can be equivalently stated as the problem of finding a smallest possible set $X$ of vertices (resp. edges) in a graph $G$
such that $G-X$ (resp. $G \setminus X$)  admits a homomorphism to $H_1$ and to $H_2$. In other words, we are solving  \textsc{HomVD}($H$) (resp. \textsc{HomED}($H$)) \emph{simultaneously} for two target graphs.

\addcontentsline{toc}{section}{References}
\bibliography{main}

\end{document}